\documentclass[journal,onecolumn]{IEEEtran}   
\usepackage{amsmath} 
\usepackage{amsthm} 
\usepackage{graphicx} 
\usepackage{cite} 
\usepackage{amsmath,amsfonts}
\usepackage{algorithmic}
\usepackage{nccmath}
\usepackage[mathlines]{lineno}
\usepackage{cases}
\usepackage{algorithm}
\usepackage{array}
\usepackage[caption=false,font=normalsize,labelfont=sf,textfont=sf]{subfig}
\usepackage{textcomp}
\usepackage{color}
\usepackage{float}
\usepackage{stfloats}
\usepackage{url}
\usepackage{orcidlink}
\usepackage{makecell}
\usepackage{verbatim}
\usepackage{graphicx}
\usepackage{cite}
\usepackage{cite}
\usepackage{amssymb}
\usepackage{blindtext}
\usepackage{nccmath}
\usepackage{amssymb}
\usepackage{mathrsfs}
\usepackage{amsmath}
\newtheorem{theorem}{Theorem}
\newtheorem{assertion}{Assertion}
\newtheorem{lemma}{Lemma}
\newtheorem{Corollary}{Corollary}

\newtheorem{Remark}{Remark}
\newtheorem{example}{Example}

\begin{document}
\title{The dimension and Bose distance of some  BCH codes of  length $\frac{q^{m}-1}{\lambda}$}
\author{Run Zheng \orcidlink{0000-0002-9117-6639},  Nung-Sing Sze  \orcidlink{0000-0003-1567-2654}
 and Zejun Huang \orcidlink{0000-0003-2621-3234}
\thanks{Z. Huang's work was supported by the National Natural Science Foundation of China (No.12171323).  N. S. Sze's work was supported by a HK RGC grant PolyU 15300121 and a PolyU
research grant 4-ZZRN.}
\thanks{R. Zheng is with the School of Mathematical Sciences,
    Shenzhen University, Shenzhen 518060, China and the Department of Applied Mathematics, The Hong Kong Polytechnic University, Hung Hom, Hong Kong (e-mail: zheng-run.zheng@connect.polyu.hk).}

\thanks{N. S. Sze is with the Department of Applied Mathematics,
The Hong Kong Polytechnic University, Hung Hom, Hong Kong (e-mail:
raymond.sze@polyu.edu.hk).}

\thanks{Z. Huang is with the School of Mathematical Sciences,
    Shenzhen University, Shenzhen 518060, China (e-mail: zejunhuang@szu.edu.cn). (Corresponding author)}
}

\date{\today}
\markboth{Journal of \LaTeX\ Class Files,~Vol.~1, No.~2, December~2023}%
{Shell \MakeLowercase{\textit{et al.}}: A Sample Article Using IEEEtran.cls for IEEE Journals}

\maketitle
\begin{abstract}
BCH codes are important error correction codes, widely utilized due to their robust algebraic structure, multi-error correcting capability, and efficient decoding algorithms. Despite their practical importance and extensive study, their  parameters, including dimension, minimum distance and Bose distance, remain largely unknown in general. This paper addresses this challenge by investigating the dimension and Bose distance of BCH codes of length $(q^m - 1)/\lambda$ over the finite field $\mathbb{F}_q$, where $\lambda$ is a positive divisor of $q - 1$.
Specifically, for narrow-sense BCH codes of this length with $m \geq 4$, we derive explicit formulas for their dimension for  designed distance  $2 \leq \delta \leq (q^{\lfloor (2m - 1)/3 \rfloor + 1} - 1)/{\lambda} + 1$. We also provide explicit formulas for their Bose distance in the range $2 \leq \delta \leq (q^{\lfloor (2m - 1)/3 \rfloor + 1} - 1)/{\lambda}$. These ranges for $\delta$ are notably larger  than the previously known results for this class of BCH codes. Furthermore, we extend these findings to determine the dimension and Bose distance for certain non-narrow-sense BCH codes of the same length. Several optimal linear codes can be obtained from these BCH codes.
\end{abstract}

\maketitle
\date{\today}
\section{Introduction}

\IEEEPARstart{T}{hroughout} this paper, let $q$ be a prime power  and   $\mathbb{F}_q$ be the finite field of order $q$. Let $\mathbb{F}_q^n$ denote  the $n$-dimensional linear space over $\mathbb{F}_q.$ A code of length $n$ over $\mathbb{F}_q$ is defined as a nonempty subset of $\mathbb{F}_q^n$. In particular,
an $[n,k,d]$  linear code $\mathcal{C}$ over $\mathbb{F}_q$ is defined as a $k$-dimensional  subspace of $\mathbb{F}_q^n$  with  minimum distance $d$. A  linear code $\mathcal{C}\subseteq \mathbb{F}_{q}^n$  is said to be  \textit{cyclic} if $(c_{0},c_{1},\ldots,c_{n-1})\in \mathcal{C}$ implies that $(c_{n-1},c_{0},\ldots,c_{n-2})\in \mathcal{C}.$ 
Identify each vector $(c_0,c_1,\ldots,c_{n-1})\in  \mathbb{F}_q^n$  with  its polynomial representation \begin{equation}\notag
c_0+c_1x+\cdots+c_{n-1}x^{n-1} \in \mathbb{F}_q[x]/(x^n-1),
\end{equation}
and  each  code $\mathcal{C}\subseteq \mathbb{F}_q^n$  with  a subset of  the quotient ring $\mathbb{F}_q[x]/(x^n-1).$ In this way, a code $\mathcal{C}\subseteq \mathbb{F}_q^n$  is a  cyclic code if and only if it is an ideal of the quotient ring $\mathbb{F}_q[x]/(x^n-1).$ 
 Note that each ideal of $\mathbb{F}_q[x]/(x^n-1)$ is principal.  
 Therefore, a cyclic  code $\mathcal{C}\subseteq \mathbb{F}_q[x]/(x^n-1)$  can be generated by  a monic polynomial $g(x)$,  denoted as  $\mathcal{C}=\left \langle g(x) \right \rangle. $ Moreover,   the polynomial  $g(x)$ is a divisor of $x^n-1$.  The polynomial $g(x)$ is called the \textit{generator polynomial} of $\mathcal{C}$, and $h(x)=(x^n-1)/{g(x)}$ is called the 
 \textit{parity-check polynomial} of $\mathcal{C}$.

Suppose that $n$ is a positive integer such that $\mathrm{gcd}(n,q)=1$.  Let   $m=\mathrm{ord}_n(q)$,  i.e., the smallest integer such that 
$q^m\equiv 1 \ (\mathrm{mod}\ n).$ Let 
$\alpha$ be a primitive element of the field  $\mathbb{F}_{q^m}$.   Then $\beta=\alpha^{\frac{q^m-1}{n}} $ is a primitive $n$-th root of unity.  This leads to the  factorization  $x^n-1=\prod\limits_{i=0}^{n-1}(x-\beta^i).$  
For each integer $i\in [0,n-1]$,  we denote by $m_i(x)$ the minimal polynomial of $\beta^i$ over $\mathbb{F}_q$.  
A cyclic  code  of length  $n$  over $\mathbb{F}_q$ is called a  BCH code with designed distance $\delta$  if its generator polynomial takes the form 
\begin{equation}\notag
\mathrm{lcm}(m_b(x),m_{b+1}(x),\ldots,m_{b+\delta-2}(x)), 
\end{equation}
where  $b$  and $\delta$ are integers with $2\leq \delta\leq n$, 
and  $\mathrm{lcm}$ denotes the least common multiple of the polynomials.    
We denote  by $\mathcal{C}_{(q,  n,\delta, b)}$ such a BCH code.  If $b=1$, it  is called  a \emph{narrow-sense} BCH code, simply denoted by $\mathcal{C}_{(q,n, \delta)}$. If $n=q^m-1$, then  it is called a primitive BCH code. 
 Note that $\mathcal{C}_{(q,n,\delta,b)}$ and $\mathcal{C}_{(q,n,
    \delta^{'}, b)}$ 
may be identical even for distinct $\delta$ and $\delta^{'}$.  The \textit{Bose distance} of $\mathcal{C}_{(q,n,\delta,b)}$, denoted by $d_B$ or $d_B(\mathcal{C}_{(q,n,\delta,b)})$, is the largest integer such that $\mathcal{C}_{(q,n,\delta,b)}=\mathcal{C}_{(q,n,d_B,b)} $.  

 BCH codes were first 
independently discovered by Hocquenghem \cite{H1959} and by Bose and Ray-Chaudhuri \cite{BC1960f,BR1960}. They  occupy a central place in coding theory due to their remarkable properties. First, they offer great flexibility in the choice of code parameters, enabling error correction capabilities to be tailored to specific applications. In addition,  for block lengths up to a few hundred bits, many BCH codes are among the most powerful codes known for given length and dimension.  In addition, efficient encoding and decoding algorithms have been developed,  which make BCH codes highly practical for real-world applications.

Despite their widespread use and extensive study in the literature 
\cite{  CPHTZ2006, DFZ2017, DYFT1997, FTT1986, HS1973, KILS2001, KOLS1999, KTFT1985, KTLS1972, LDL2017, LLDL2017, LRFL2019, RM1991, YH1996, CP1994,WXZ2024,FLD2023,WWL2023,XL2024,DLMQ2023,SLD2023, ZXP2024}, several open problems persist regarding BCH codes, especially concerning the precise determination of their dimension, minimum distance, and Bose distance. These parameters are crucial indicators of a BCH code's performance. Specifically, a BCH code over $\mathbb{F}_q$ with dimension $k$ and minimum distance $d$ can transmit $k$ $q$-ary information symbols and correct up to $\left\lfloor \frac{d-1}{2} \right\rfloor$ $q$-ary symbol errors. Furthermore, the Bose distance $d_B$ provides a fundamental lower bound on the minimum distance, as established by the BCH bound \cite{H1959, BR1960}. Notably, Charpin \cite{CP1994} conjectured that for a narrow-sense primitive BCH code, $d \leq d_B + 4$. Hence,  determining the Bose distance is also invaluable for a deeper understanding of BCH codes and their capabilities. However, as noted  by Charpin \cite{charpin1998open} and Ding \cite{ding2024bch}, the general determination of these parameters remains a challenging problem.

This paper is dedicated to the investigation of the dimension and Bose distance of BCH codes of length $(q^m-1)/\lambda$,    where $\lambda$ is a positive divisor of $q-1$.  This class of   BCH codes includes primitive BCH codes when $\lambda=1$. To date, these fundamental parameters are precisely known only for limited cases. Most existing results  focus on the case where $\lambda = 1$ and $b = 1$, which corresponds to narrow-sense primitive BCH codes. For a comprehensive overview of the parameters of such codes, readers are referred to \cite{zheng2025, liu2017, NLMY2021, GLQ2024,pang2021five,D2015,DDZ2015}.

In contrast, for cases when
$\lambda\neq 1$, the understanding is much more limited. Even for some specific cases, such as $\lambda = 2$ and $\lambda = q-1$,
the dimension and Bose distance are known only for a few designed distances.
For general positive  divisors $\lambda$ of $q-1$, Zhu et al.  \cite{ZSK2019}   determined the
 dimension of narrow-sense BCH codes $\mathcal{C}_{(q,(q^m-1)/\lambda,\delta)}$  for designed distances $2\leq \delta\leq    \frac{q^{\lceil (m+1)/2 \rceil}-1}{\lambda} + 1$. Recently, Sun \cite{S2025} determined the dimension and minimum distance of $\mathcal{C}_{(q,(q^m-1)/\lambda,\delta)}$ for  other specific designed distances. For convenience, we summarize  in Table \ref{table111} some known results on the parameters of  narrow-sense BCH codes $\mathcal{C}_{(q,(q^m-1)/\lambda,\delta)}$ for divisors $\lambda$ of $q-1$ with $\lambda\neq 1$. A check mark \checkmark indicates that the corresponding parameter is known, while a blank entry indicates that it is unknown.  
For non-narrow-sense BCH codes of length $(q^m-1)/\lambda$,
the available results in the literature are even more scarce;
see \cite{liu2017} for some related work.
   Readers may  also refer to \cite{ding2024bch} for an excellent survey on known results regarding the parameters of BCH codes.
\begin{table*}[h]\label{table111}
\renewcommand{\arraystretch}{2.5}
\centering
\caption{Known results on the parameters of   $\mathcal{C}_{(q,(q^m-1)/\lambda,\delta)}$ for divisors $\lambda$ of $q-1$ with  $\lambda\neq 1$}
\begin{scriptsize}
\begin{tabular}{|c|c|c|c|c|c|c
|c|}
\hline
$q$ &$\lambda$ &$m$  & $\delta$ & $k$ & $d_B$ & $d$&  Reference\\ \hline
prime power &$\lambda=q-1$ &$m \geq 4 $ is even  &$2\leq \delta\leq q^{{m}/{2}}$ & \checkmark & & &  \cite{LDL2017}  \\ \hline
prime power &$\lambda=q-1$&$m\geq 5$  is odd & $\delta= a q^{(m-1)/2}+1$  with $1\leq a\leq q-1$& \checkmark   &   & &   \cite{liu2017}           \\  \hline
 prime power & $\lambda=q-1$ &$m\geq 1$   & $\delta=2 $ & \checkmark  & \checkmark& \checkmark&     \cite{LDXG}          \\  \hline
$q\geq 3$ &$\lambda=q-1$&$ m\geq 1$  &  $\delta=3$ &\checkmark &   & partially  solved &  \cite{LDXG}  \\ \hline
$q=3$ &$\lambda=2$& $m\geq 3$  & $\delta=\frac{q^m-q^{m-1}}{2}-\frac{q^{\lfloor (m-3)/2\rfloor+i}+1}{2}$, $i=1,2$  &\checkmark   &\checkmark  & \checkmark &   \cite{LDXG}  \\ \hline
$q$ is odd &$\lambda=2$& $m\geq 3$ & $\delta=\frac{q^m-q^{m-1}}{2}-\frac{q^{\lfloor (m-3)/2\rfloor+i}+1}{2}$, $i=1,2$ &\checkmark & \checkmark&\checkmark & \cite{ZSK2019}\\ \hline 
$q$ is odd &$\lambda=2$ &$m\geq 2$   & \makecell[c]{$\delta= \frac{q^m-q^{m-1}}{2}-\frac{q^{\lfloor (m-3)/2\rfloor+i}+1}{2}$,  $1\leq i\leq \lfloor \frac{(m+11)}{6}\rfloor$} &   & \checkmark&\checkmark & \cite{LMS2020} \\ \hline 

$q$ is odd& $\lambda=2$& $m\geq 4$& $\delta=\frac{q^m-q^{m-1}}{2}-\frac{ q^{\lceil \frac{3m}{4}\rceil-1}+ q^{\lceil  \frac{m}{2}\rceil-1}+q^{\lceil \frac{m}{4}\rceil-1}+1}{2}$ &\checkmark & \checkmark  & & \cite{WLZ2024} \\ \hline  
$q$ is odd& $\lambda=2$ &\makecell[c]{$m\geq 7$,  \\ $m\equiv 3 (\bmod 4)$} & $\delta = \frac{q^m-q^{m-1}}{2} - \frac{q^{\frac{3m-1}{4}} + q^{\frac{m+1}{2}}+q^{\frac{m-11}{4}+t}+1}{2}$, $t=2,3,4$& \checkmark& \checkmark& & \cite{WLZ2024}\\ \hline   
$q$ is odd& $\lambda=2$ &\makecell[c]{$m\geq 7$, \\ $m\not\equiv 3 (\bmod 4)$} & $\delta = \frac{q^m-q^{m-1}}{2} - \frac{q^{\lceil \frac{3m}{4}\rceil}+1}{2}$& \checkmark& \checkmark& & \cite{WLZ2024}\\ \hline 

$q$ is odd& $\lambda=2$   & $m=2$ & \makecell[c]{$\delta = \frac{i_1q+i_0}{2}$, $0\leq i_1\leq q-2$, \\ $\max \{i_0,1\}\leq i_0\leq q-1$ and $i_0\equiv i_1 (\bmod 2)$}  & \checkmark&\checkmark & & \cite{WLZ2024} \\ \hline   

$q$ is odd &$\lambda=2$ & $m=3$ & \makecell[c]{ $\delta=\frac{i_2q^2+i_1q+i_0}{2}$, $0\leq i_2\leq q-2$, $i_2\leq i_1\leq q-1$, \\ $\max\{i_2, 1\}\leq i_0\leq q-1$,  $i_1+i_0 \equiv i_2 (\bmod 2)$} &\checkmark &\checkmark & & \cite{WLZ2024} \\ \hline    
$q$ is odd &$\lambda=2$& $m\geq 4$& $\delta=\frac{q^m-q^{m-1}}{2}-\frac{q^{\lfloor (m-3)/2\rfloor+i}+1}{2}, 1\leq  i\leq \lfloor \frac{m+6}{4} \rfloor$ &\checkmark & \checkmark &\checkmark & \cite{WLZ2024} \\ \hline    

prime power &$\lambda \mid q-1$ & $m\geq 3$ is odd & $2\leq \delta\leq \frac{q^{(m+1)/2}-1}{\lambda}+1$&  \checkmark &  & & \cite{ZSK2019} \\ \hline 
prime power &$\lambda\mid  q-1$& $m\geq 4$ is even & $2\leq \delta \leq \frac{q^{m/2+1}-1}{\lambda}+1$  & \checkmark & & & \cite{ZSK2019} \\ \hline
prime power& $\lambda \mid q-1$& $m\geq 2$& \makecell[c]{$\delta=\frac{(q-\lambda \ell_0)q^{m-1-\ell_1}-1}{\lambda}$, \\ $0\leq \ell_1\leq m-1$ and $0\leq \ell_0<\frac{q-1}{\lambda}$}& \checkmark &\checkmark &\checkmark & \cite{S2025}\\ \hline   
\end{tabular}
\end{scriptsize}
\end{table*}

A deep understanding of $q$-cyclotomic cosets, particularly their sizes and coset leaders, is crucial for determining the dimension and Bose distance of BCH codes. Indeed, the main challenge in determining these parameters of BCH codes often stems from the irregular distribution of coset leaders. Our previous work  \cite{zheng2025} explored the distribution of coset leaders modulo $q^m-1$ within the range $[1, q^{\lfloor (2m - 1)/3 \rfloor + 1}]$ and the sizes of corresponding $q$-cyclotomic cosets,  enabling us to determine the Bose distance and dimension of narrow-sense primitive BCH codes $\mathcal{C}_{(q, q^m - 1, \delta)}$ for $m \geq 4$ and $2\leq \delta \leq  q^{\lfloor (2m - 1)/3 \rfloor + 1}$. 

To extend these results to BCH codes of length $(q^m-1)/\lambda$, we need to investigate $q$-cyclotomic cosets modulo $(q^m-1)/\lambda$. A key and useful observation in this regard is that an integer $a $ is a coset leader modulo $(q^m-1)/\lambda$ if and only if $\lambda a$ is a coset leader modulo $q^m - 1$, and the size of the $q$-cyclotomic coset modulo $(q^m-1)/\lambda$ of $a$ is equal to the size of the $q$-cyclotomic coset modulo $q^m-1$ of  $\lambda a$. Consequently, the problem of finding coset leaders and sizes of   $q$-cyclotomic coset  modulo $(q^m-1)/\lambda$ can be reduced  to identifying integers divisible by $\lambda$ that are coset leaders modulo $q^m - 1$ and determining  the sizes of their corresponding  $q$-cyclotomic coset  modulo $q^m - 1$.

Building upon this crucial observation and our prior analysis of 
 $q$-cyclotomic cosets modulo $q^m-1$, we successfully generalize our  results on primitive BCH codes in \cite{zheng2025} to BCH codes of length $(q^m-1)/\lambda$. Specifically, for any positive divisor $\lambda$ of $q-1$ and positive integer $m\geq 4$, this paper  determines: 
\begin{itemize}
    \item the dimension of $\mathcal{C}_{(q,(q^m-1)/\lambda,\delta)}$ for $2\leq \delta \leq \frac{q^{\lfloor (2m-1)/3\rfloor+1}-1}{\lambda} +1 $;
\item the Bose distance of $\mathcal{C}_{(q,(q^m-1)/\lambda,\delta)}$ for $2\leq \delta \leq \frac{q^{\lfloor (2m-1)/3\rfloor+1}-1}{\lambda}$.   
\end{itemize}
It is important to note that the existing knowledge of the dimension for narrow-sense BCH codes 
 $\mathcal{C}_{(q,(q^m-1)/\lambda,\delta)}$ only covers designed distances  $2\leq \delta\leq    \frac{q^{\lceil (m+1)/2 \rceil}-1}{\lambda} + 1$ and some   specific cases.
Our results significantly extend this range, as evidenced by the inequality $\frac{q^{\lfloor (2m-1)/3\rfloor+1}-1}{\lambda}\geq \frac{q^{\lfloor (m+1)/2 \rfloor }-1}{\lambda}\cdot q^{\lceil (m-4)/6\rceil}$. This implies that the range of 
$\delta$ for which we provide the dimension of 
$\mathcal{C}_{(q,(q^m-1)/\lambda,\delta)}$
is substantially larger than previously established.
  Additionally, we extend these results to certain non-narrow-sense BCH codes of the same length.

This paper is organized as follows.
In Sections \ref{pre} and \ref{review}, we provide essential preliminaries and review our previous results concerning $q$-cyclotomic cosets modulo $q^m-1$. Section \ref{aux} presents several auxiliary lemmas that will be employed in subsequent sections.
Based on the theoretical foundations established in Sections \ref{pre} through \ref{aux}, 
we then determine the dimension of  BCH codes $\mathcal{C}_{(q,(q^m-1)/\lambda,\delta)}$ for $2\leq \delta \leq \frac{q^{\lfloor (2m-1)/3\rfloor+1}-1}{\lambda} +1 $ in Section \ref{dim} and the Bose distance of $\mathcal{C}_{(q,(q^m-1)/\lambda,\delta)}$ for $2\leq \delta \leq \frac{q^{\lfloor (2m-1)/3\rfloor+1}-1}{\lambda}$ in Section \ref{bose}, both by providing explicit formulas. Utilizing these  formulas, we present some  examples of BCH codes $\mathcal{C}_{(q,(q^m-1)/\lambda,\delta)}$  and compare them with the tables of the best known linear codes maintained by Markus Grassl at http://www.codetables.de, which is called \textit{Database} later in this paper.
Furthermore, as an illustration of our main results, Section \ref{illustration} applies these formulas to compute the dimension and Bose distance of BCH codes $\mathcal{C}_{(q,(q^m-1)/\lambda,\delta)}$ specifically when $\lambda \delta$ takes the form $aq^{h+k}+b$, where $h=\lfloor m/2\rfloor$, and $k, a, b$ are integers such that  $m - 2h \leq k \leq \lfloor (2m - 1)/3 \rfloor - h$, $1 \leq a \leq q - 1$, $\lambda \leq b \leq q^{m - h - k}$ and $q \nmid b$.
Following this, Section \ref{non} extends our analysis to determine the dimension and Bose distance of certain non-narrow-sense BCH codes. Moreover, we identify some non-narrow-sense BCH codes that possess optimal parameters.  Finally, Section \ref{conclusion} concludes the paper.

\section{Preliminaries}\label{pre}
Let \( n \) be an integer such that \( \gcd(n, q) = 1 \).  
For each integer \( a \in [0, n-1] \), the \emph{\( q \)-cyclotomic coset} of \( a \) modulo \( n \) is defined as 
\begin{equation}\label{def1}
C_n(a) = \{ a q^{i} \bmod n \mid i = 0, \ldots, \ell_a - 1 \},
\end{equation}
where \( \ell_a \) is the smallest positive integer such that \( a q^{\ell_a} \equiv a \pmod{n} \).  
It is clear that $|C_n(a )|=\ell_a$, where $|\cdot|$ denotes the size of a set.  Moreover, it is well known that $|C_n(a)|=\ell_a$ is  a divisor of \( m = \mathrm{ord}_n(q) \).  
The smallest integer in \( C_n(a) \) is called the \emph{coset leader} of $C_n(a)$.  For convenience, we  occasionally refer to it simply as a
coset leader modulo $n$.
Let \( \ell \) be  a positive integer, 
and  let \( c, d \in [0, n-1] \) be two integers such that \( c \leq d \).  
We define 
\begin{equation}\notag
\mathcal{L}^n(c,d) = \{ a \in [c,d] : a \text{ is the coset leader of } C_n(a) \}
\end{equation}
and
\begin{equation}\notag
\mathcal{L}^n_{\ell}(c,d) = \{ a \in [c,d] : a \text{ is the coset leader of } C_n(a) \text{ and } |C_n(a)| = \ell \}.
\end{equation}
In particular, we simply denote $\mathcal{L}^n_{\ell}(0,n-1)$ by $\mathcal{L}^n_{\ell}$.

It is widely known that the minimal polynomial $m_a(x)$ of $\beta^a$
over $\mathbb{F}_q$ 
is  given by 
\begin{equation*}m_a(x)= \prod\limits_{i\in C_n(a)}(x-\beta^i); \end{equation*} see, for example, \cite[Theorem 4.1.1]{Hou2018}.  
Hence, the generator polynomial  $g(x)$ of the  BCH code $\mathcal{C}_{(q,n,\delta,b)}$ can be expressed as 
\begin{equation}\notag
g(x)= \prod\limits_{i\in \mathcal{G}}(x-\beta^i),\quad \mathcal{G}=\bigcup\limits_{a=b}^{b+\delta-2}C_n(a).
\end{equation}
 Consequently,  
the dimension of $\mathcal{C}_{(q,n,\delta,b)}$ is 
\begin{equation}\label{ndimeq}
    \mathrm{dim}(\mathcal{C}_{(q,n,\delta,b)})=n-\left|\bigcup\limits_{a=b}^{b+\delta-2}C_n(a)\right|,
\end{equation}
and the Bose distance $d_B$ of $\mathcal{C}_{(q,n,\delta,b)}$ is exactly the largest integer such that 
\begin{equation}\label{boseeq}
    \bigcup\limits_{a=b}^{b+\delta-2}C_n(a) = \bigcup\limits_{a=b}^{b+d_B-2}C_n(a).
\end{equation}
In particular, 
the dimension of narrow-sense BCH codes $\mathcal{C}_{(q,n,\delta)}$ can be given by  \begin{equation}\label{dimeq}
    \mathrm{dim}(\mathcal{C}_{(q,n,\delta)})=n-\left|\bigcup\limits_{a=1}^{\delta-1}C_n(a)\right|=n-\sum\limits_{a\in  \mathcal{L}^n(1,\delta-1) }|C_n(a)|. 
    \end{equation}
Thus, determining the coset leaders modulo 
$n$ within 
 $[1,\delta-1]$  
 and the sizes of corresponding $q$-cyclotomic  cosets allows us to compute the dimension of $\mathcal{C}_{(q,n,\delta)}$. 
Additionally, for $\delta' > \delta$, one can easily verify that
\begin{equation*}
   \bigcup_{a=1}^{\delta-1} C_n(a) = \bigcup_{a=1}^{\delta'-1} C_n(a)
\end{equation*}
if and only if none of the integers from $\delta$ to $\delta'-1$ is a coset leader.
Therefore, the Bose distance of $\mathcal{C}_{(q,n,\delta)}$ equals the smallest coset leader that is not less than $\delta$.

We have a useful observation regarding 
$q$-cyclotomic cosets, as presented in the following lemma. This result originated  from  the proof of \cite[Lemma 6]{ZSK2019}. For completeness, we include a proof below.    
\begin{lemma}\label{cos2}
Suppose that  $n$ and $\lambda$ are  two  integers such that  $\mathrm{gcd}(n,q)=1
 $  and  $\mathrm{gcd}(\lambda,q)=1$. Let $a\in [0,n-1]$ be  an integer. Then  
 \begin{itemize}
     \item  $a$ is the   coset leader of  $C_{n}(a)$ if and only if  $\lambda a$  is the coset leader of  $C_{\lambda n}(\lambda a)$;
     \item $|C_{n}(a)|=|C_{\lambda n}(\lambda a)|.$
 \end{itemize}
\end{lemma} 
\begin{proof}
    By definition, the integer $a$ is the coset leader of $C_{n}(a)$ if and only if 
    \begin{equation}\label{fan}
        a \bmod n  \leq aq^i \bmod  n \quad \hbox{for any integer  } i\geq 0. 
    \end{equation}
Noticing that 
    $a \bmod n = a- n\cdot \lfloor \displaystyle\frac{ a}{ n } \rfloor $ and $\lambda a  \bmod \lambda n=\lambda a - \lambda n \cdot \lfloor \displaystyle\frac{\lambda a}{\lambda  n}\rfloor$, 
we can assert that 
\begin{equation}\notag
    \lambda \cdot (a \bmod n) =  \lambda a -  \lambda n \cdot \lfloor \displaystyle\frac{  a}{ n }\rfloor = \lambda a -  \lambda n \cdot \lfloor \displaystyle\frac{ \lambda a}{\lambda n }\rfloor = \lambda a \bmod \lambda n.
\end{equation}
Similarly, we also have 
\begin{equation}\label{qq1}
   \lambda \cdot (aq^i \bmod n) = \lambda a q^i  \bmod \lambda n. 
\end{equation}
Therefore, the inequality in (\ref{fan}) is equivalent to 
\begin{equation}\notag
    \lambda a  \bmod   \lambda n \leq \lambda a q^i  \bmod \lambda n \quad \hbox{for any integer  } i\geq 0. 
\end{equation}
By definition, this  holds if and only if  $\lambda a $ is the coset leader of $C_{\lambda n}(\lambda a)$. Therefore, the first statement follows. 

In addition, the equality in  (\ref{qq1}) also implies that 
\begin{equation}\notag
    aq^{i} \bmod n =a \quad \hbox{if and only if} \quad  \lambda a q^i  \bmod  \lambda n =\lambda a.
\end{equation} 
It follows that  the smallest integer $\ell_a$ such that $aq^{\ell_a} \bmod n = a$ is equal to the smallest integer $\ell_{\lambda a} $ such that $\lambda aq^{\ell_{\lambda a}} \bmod \lambda n=\lambda a$. Therefore, we have $|C_{n}(a)|=|C_{ \lambda n}(\lambda a)|.$ This completes the proof. 
\end{proof}
Let $\mathbb{Z}$ be the set of all integers. For each integer $\lambda$, we denote by 
$\mathcal{D}_{\lambda}$  the set of all the integers that are divisible by $\lambda$, that is, 
 \begin{equation}\notag
    \mathcal{D}_{\lambda}=\{a\in \mathbb{Z}:\lambda\mid a\}.
\end{equation} 
Then one can directly derive  the following corollary from Lemma \ref{cos2}.
\begin{Corollary}\label{cor1}
Let $n$ and $\lambda$ be   two  integers such that  $\gcd(n,q)=1
 $  and  $\mathrm{gcd}(\lambda,q)=1$.
 Then  
 \begin{equation}\notag
     \left |   \mathcal{L}_{\ell}^{n}(b,c)  \right| =  \left |   \mathcal{L}^{\lambda n}_{\ell}(\lambda b,\lambda c)\cap \mathcal{D}_{\lambda}\right|
 \end{equation}
 for any  positive integer  $\ell $, and  integers $b,c$ with $0\leq b\leq c\leq n-1$.  
\end{Corollary}

\section{Some known results on \texorpdfstring{$q$}\ -cyclotomic cosets  modulo \texorpdfstring{$q^m-1$}\ }\label{review} 
Throughout the rest of the paper,  let  $m$ be a positive integer, let  $\lambda$ be a  positive divisor of $q-1$, and  set  $h=\lfloor\frac{m}{2}\rfloor$  and $n=(q^m-1)/\lambda$.   
For  a positive real number $a$, we denote by   $N(a)$ the number of integers in the range $[1,a-1]$ that   are not divisible by $q$.  Equivalently, $
    N(a) = \lfloor a-1\rfloor-\lfloor (a-1)/{q}\rfloor. 
$

In this section, we briefly review some known results on   $q$-cyclotomic cosets modulo 
$\lambda n=q^m -1$.   These foundational results will be essential for deriving the main contributions of this paper. 
To maintain consistency, we follow the notation and terminology used in \cite{zheng2025}, which we now introduce for completeness.

 Let $Z_q =\{0, 1,2,\ldots, q-1\}$ be the set of nonnegative integers strictly less than $q$, and let  
 $Z_q^{m} = \{(a_{m-1}, \ldots, a_0): a_i\in Z_q \}$  be  the set of all length-$m$
sequences of integers in  $Z_q$.   For simplicity, we  denote by $\mathbf{0}_{m}$ the sequence in $Z_q^m$ whose elements are all zero. 
We define an order on $Z_q^m$ by  lexicographic order. Specifically, for any two sequences   $U=(u_{m-1},\ldots,u_1,u_{0})$ and $W=(w_{m-1},\ldots,w_1,w_0)$ in $Z_q^m$, we define:
\begin{itemize}
\item [1.] $U$ is equal to $W$, denoted by $U=W$,  if    $u_{\ell}=w_{\ell}$ for $\ell=0,\ldots,m-1.$
\item  [2.]  $U$ is less than $W$, denoted by $U<W$, if either  $u_{m-1}<w_{m-1}  $  or there exists an integer $i\in [0,m-2]$ such that  
$u_i<v_i$ and $u_{\ell}=w_{\ell}$ for all $\ell=i+1,\ldots,m-1.$ 
\item [3.] $U\leq W$ is denoted   if $U= W$ or $U<W$.
\end{itemize}  
Note that each integer $a \in [0, \lambda n]$ can be uniquely represented by its \textit{$q$-adic expansion} as
$
a = \sum\limits_{\ell=0}^{m-1} a_{\ell} q^{\ell},
$
where $a_{\ell} \in Z_q$ for all $\ell = 0, 1, \dots, m-1$.
We define a map  $V$ from the set of  integers in $[0, \lambda n]$ to $Z_q^{m}$  as 
  \begin{equation}\notag
  V(a)=(a_{m-1}, \ldots,a_1, a_0),
\end{equation}   
where   $\sum\limits_{\ell=0}^{m-1}a_{\ell}q^{\ell}$ forms the  $q$-adic expansion of the integer  $a\in [0,\lambda n].$

Assume that $m\geq 4 $.  Let $k $ be an integer such that   $m-2h\leq k\leq  \lfloor {(2m-1)}/{3}\rfloor-h$. 
If $m$ is an odd integer,  
  for each integer $i\in [-k+1, k]$, we define  $\mathcal{A}_k(i)$ as  the set of  all  integers $a\in [q^{h+k},q^{h+k+1})$ with $q$-adic expansion $\sum\limits_{\ell=0}^{h+k}a_{\ell}q^{\ell}$  satisfying the following conditions:  
\begin{eqnarray}
&&a_{h+i}>0; \label{p31}\\
&&(a_{k+i-1},\ldots,a_{0})\leq (a_{h+k}, \ldots, a_{h-i+1}) \hbox{ and }a_0>0;  \label{p32} \\ 
&&V(a)=(\mathbf{0}_{h-k}, a_{h+k},\ldots,a_{h+i},\mathbf{0}_{h-k},a_{k+i-1},\ldots,a_0).  \label{p33}
\end{eqnarray}
If  $m$ is an even integer,  for each integer $i\in [-k,k]$,  we define  $\mathcal{B}_k(i)$ as the set of all integers $a\in [q^{h+k},q^{h+k+1})$  with $q$-adic expansion $\sum\limits_{\ell=0}^{h+k}a_{\ell}q^{\ell}$ satisfying the condition in (\ref{p31}) and the following:    
\begin{eqnarray}
&&(a_{k+i},\ldots,a_{0})\leq (a_{h+k}, \ldots, a_{h-i}) \hbox{ and }a_0>0;  \label{p35} \\
&&V(a)=(\mathbf{0}_{h-k-1}, a_{h+k},\ldots,a_{h+i},\mathbf{0}_{h-k-1},a_{k+i},\ldots,a_0).  \label{p36}   
\end{eqnarray}

With the above preparations, we now present several known results on
$q$-cyclotomic cosets modulo $\lambda n = q^{m}-1$ as follows.

\begin{lemma}\textnormal{\cite[Theorem~2.3]{YF2000}\label{lll}}
If $m$ is odd, then for any integer $a \le q^{h+1}$, $a$ is the coset
leader of $C_{\lambda n}(a)$ if and only if $q\nmid a$.
If $m$ is even, then for any integer $a \le 2q^{h}$, $a$ is the coset
leader of $C_{\lambda n}(a)$ if and only if $q\nmid a$.
\end{lemma}
Clearly, if an integer $a$ is divisible by $q$, then $a/q \in C_{\lambda n}(a)$. 
Therefore, $q \nmid a$ is a necessary condition for $a$ to be a coset leader modulo 
$\lambda n = q^m - 1$. 
The above lemma shows that this condition is also sufficient for sufficiently small $a$. However, for larger values of $a$, there exist integers $a$ that are neither divisible by $q$ nor coset leaders modulo $\lambda n = q^m - 1$.  
We denote by 
 $\mathcal{S}$  the set such integers, that is,  \begin{equation}\notag
\mathcal{S}=\{a\in [1,\lambda n-1] : q\nmid a \text{ and } a \text{ is not the coset leader of } C_{\lambda n}(a)\},
\end{equation}
In addition, we define $\mathcal{H}$ to be the set of integers in $[1,\lambda n-1]$ that are 
coset leaders modulo $\lambda n$ and whose corresponding $q$-cyclotomic cosets have size $m/2$, i.e.,
\begin{equation*}
\mathcal{H}
=\{a\in [1,\lambda n-1] : a \text{ is the coset leader of } C_{\lambda n}(a)
\text{ and } |C_{\lambda n}(a)|=m/2\}.
\end{equation*}
The following result shows that, within certain ranges, the integers in $\mathcal{S}$ (for odd $m$) 
and in $\mathcal{S}\cup\mathcal{H}$ (for even $m$) can be partitioned into several disjoint classes, where the symbol $\bigsqcup$ denotes disjoint union.
 
\begin{lemma}\textnormal{\cite[Theorem 2]{zheng2025}}\label{th2}
Let $m$ and $k$ be  two integers such that $m\geq 4$ and  $m-2h\leq k\leq   \lfloor (2m-1)/{3}\rfloor -h.$
\begin{itemize}
    \item 
If $m$ is odd,  then 
 \begin{equation}\notag
\mathcal{S}\cap [q^{h+k},q^{h+k+1})=\bigsqcup\limits_{i=-k+1}^k\mathcal{A}_k(i).
\end{equation}
\item If $m $ is  even,   
 then \begin{equation}\notag
(\mathcal{S}\cup \mathcal{H})\cap [q^{h+k},q^{h+k+1}) =\bigsqcup\limits_{i=-k}^k\mathcal{B}_k(i).
\end{equation}
\end{itemize}
\end{lemma}

In \cite[Theorem 1]{zheng2025}, we determined the sizes of certain
$q$-cyclotomic cosets modulo $q^{m}-1$. 
There is a typo in the even-$m$ case:
the upper bound $q^{m-\lfloor m/4\rfloor}$ should be 
$q^{m-\lfloor m/3\rfloor}$. 
For  convenience, we restate this result below with the
corrected range.
 
\begin{lemma}\textnormal{\cite[Theorem 1]{zheng2025}}\label{th1}
If  $m$ is an  odd integer,  then for  any integer $a\in [1,q^{m-\lfloor {m}/{3}\rfloor}),$
\begin{equation}\notag
|C_{\lambda n}(a)|=m. 
\end{equation}
If   $m$ is an even integer,  then  for any integer  $a\in  [1,q^{m-\lfloor{m}/{3}\rfloor}),$ 
\begin{equation*}
|C_{\lambda n}(a)|=\begin{cases}
\frac{m}{2} & \hbox{if }  aq^{h}\mathrm{\ mod\ } \lambda n = a, \\
m & \hbox{if } aq^{h} \mathrm{\ mod\ }\lambda n \neq  a. 
\end{cases}
\end{equation*}
\label{th2e}
\end{lemma}

The following result characterizes the integers $a$ within certain ranges for which the $q$-cyclotomic coset $C_{\lambda n}(a)$ has size $m/2$  when $m$ is even.
\begin{lemma}\textnormal{\cite[Corollary 2]{zheng2025}}\label{corr}
Let $m$ be an even integer and  $k\in \left[0,\lfloor(2m-1)/{3}\rfloor-h\right]$ be an integer. Suppose that $ a\in  \left[q^{h+k},q^{h+k+1}\right)$   is  an integer. Then   $a\in \mathcal{H}$ if and only if $V(a)$ has the form 
\begin{equation}\notag
(\mathbf{0}_{h-k-1},a_{k},\ldots,a_0, \mathbf{0}_{h-k-1},a_{k},\ldots,a_0)
\end{equation}
with  $a_0>0$ and $a_k>0$.   
\end{lemma}

Although the above results only describe the sizes and coset leaders of
$q$-cyclotomic cosets modulo $\lambda n = q^{m}-1$, by applying Lemma \ref{cos2}, they also yield information on
$q$-cyclotomic cosets modulo
$
n=(q^{m}-1)/\lambda.
$
These facts will be used later to determine the dimension and Bose distance of BCH codes of length $n=(q^{m}-1)/\lambda$.

To facilitate the use of these results in subsequent sections, we
introduce additional notation and make several remarks concerning the
definitions of the sets $\mathcal{A}_k(i)$ and $\mathcal{B}_k(i)$.
For each integer $a \in \mathcal{A}_k(i)$, we define
\begin{equation*}
t(a)=\sum_{\ell=h+i}^{h+k} a_{\ell} q^{\ell-h-i}
\quad \text{and} \quad
\alpha(a)=\sum_{\ell=0}^{k+i-1} a_{\ell} q^{\ell}.
\end{equation*}
Similarly, for each integer $a \in \mathcal{B}_k(i)$, we define
\begin{equation*}
t(a)=\sum_{\ell=h+i}^{h+k} a_{\ell} q^{\ell-h-i}
\quad \text{and} \quad
\alpha(a)=\sum_{\ell=0}^{k+i} a_{\ell} q^{\ell}.
\end{equation*}
We then make the following remarks.
\begin{Remark}\label{rm3}
 The condition in (\ref{p31})  can be equivalently represented as $q\nmid t(a)$.
\end{Remark}
\begin{Remark}\label{rr2}
 The condition  in (\ref{p33})  implies that each  integer $a\in \mathcal{A}_k(i)$ can be uniquely  decomposed as \begin{equation*}a=t(a)q^{h+i}+\alpha(a).
 \end{equation*}
 Since $\lambda$ is a divisor $q-1$, it follows that 
for any integer $a\in \mathcal{A}_k(i)$,  \begin{equation*} \lambda \mid a  \quad \hbox{ if and only if } \quad 
\lambda \mid t(a)+\alpha(a).\end{equation*}
Similarly, each integer $a\in \mathcal{B}_k(i)$ can be uniquely decomposed as 
\begin{equation*}
a=t(a)q^{h+i}+\alpha(a),\end{equation*} and 
 \begin{equation*}\lambda\mid a  \quad \hbox{ if and only if } \quad 
\lambda \mid t(a)+\alpha(a).\end{equation*} 
\end{Remark}   

\begin{Remark}\label{rm4}
 The  condition  in    (\ref{p32}) can be equivalently expressed as  \begin{equation}\label{mm}
    1\leq \alpha(a)\leq \sum\limits_{\ell=h-i+1}^{h+k}a_{\ell}q^{\ell-(h-i+1)}\quad \hbox{and}\quad q\nmid \alpha(a).
 \end{equation}
Moreover,  the form of $V(a)$ in (\ref{p33}) implies that 
\begin{equation}\notag
    \sum\limits_{\ell=h-i+1}^{h+k}a_{\ell}q^{\ell-(h-i+1)}=t(a)\cdot q^{2i-1} \quad \hbox{for all integers } i\in [1,k],
\end{equation} 
and 
\begin{equation}\notag
\sum\limits_{\ell=h-i+1}^{h+k}a_{\ell}q^{\ell-(h-i+1)} \nonumber +   \sum\limits_{\ell=h+i}^{h-i}a_{\ell}q^{\ell-(h-i+1)} 
  = t(a)\cdot q^{2i-1}\quad \hbox{for all integers }i\in [-k+1,0]. 
\end{equation} 
Noticing that  $0\leq \sum\limits_{\ell=h+i}^{h-i}a_{\ell}q^{\ell-(h-i+1)}<1$ for all integers $i\in [-k+1,0]$,  we  conclude that  (\ref{mm}) is equivalent to  
\begin{equation}\notag
   1\leq \alpha(a)\leq    \lfloor t(a)
   \cdot q^{2i-1}\rfloor \quad \hbox{and}\quad q\nmid \alpha(a).
\end{equation}
Similarly, 
the condition in  (\ref{p35}) can also   be written as  
 \begin{equation}\notag
 1\leq \alpha(a)\leq \sum\limits_{\ell=h-i}^{h+k}a_{\ell}q^{\ell-(h-i)} \quad\hbox{and}\quad q\nmid \alpha(a),\end{equation} which is further equivalent to 
 \begin{equation}\notag
 1\leq \alpha(a)\leq \lfloor t(a)\cdot q^{2i}\rfloor \quad\hbox{and}\quad q\nmid \alpha(a).
 \end{equation}
\end{Remark}

\begin{Remark}\label{rm6}
If   $a\in \mathcal{A}_k(i)$ is an integer with  $q$-adic expansion $\sum\limits_{\ell=0}^{h+k}a_{\ell}q^{\ell},$ then 
$\sum\limits_{\ell=0}^{h-k}a_{\ell}q^{\ell}\leq \sum\limits_{\ell=h-i+1}^{h+k}a_{\ell}q^{\ell-(h-i+1)}$.
Similarly,  if $a\in \mathcal{B}_k(i)$ is an integer with  $q$-adic expansion $\sum\limits_{\ell=0}^{h+k}a_{\ell}q^{\ell}$,  then   
$\sum\limits_{\ell=0}^{h-k-1}a_{\ell}q^{\ell}\leq \sum\limits_{\ell=h-i}^{h+k}a_{\ell}q^{\ell-(h-i)}$. 
\end{Remark}

\begin{Remark}\label{cor2} 
If   $a\in \mathcal{A}_k(i)$ is an integer with  $q$-adic expansion $\sum\limits_{\ell=0}^{h+k}a_{\ell}q^{\ell},$ then 
$i$ is the smallest integer in $[-k+1,k]$ such that $a_{h+i}>0$.  Consequently, $\mathcal{A}_k(i)\cap \mathcal{A}_k(j)=\varnothing$  for  distinct integers $i,j\in [-k+1,k]$. 

Similarly,  if $a\in \mathcal{B}_k(i)$ is an integer with  $q$-adic expansion $\sum\limits_{\ell=0}^{h+k}a_{\ell}q^{\ell}$,  then   $i$ is the smallest integer in $[-k,k]$ such that $a_{h+i}>0$.  As a result,  $\mathcal{B}_k(i)\cap \mathcal{B}_{k}(j)=\varnothing$ for distinct integers $i,j\in [-k,k]$. 
\end{Remark}

\section{Auxiliary Lemmas}\label{aux}
The following lemmas are needed to establish the main theorems on the dimension and Bose distance of BCH codes in the subsequent sections. Their proofs are given in Appendices A–H.
\begin{lemma}\label{le1}Suppose that  $x$ and $y$ are two positive integers. Then  
   \begin{equation}\notag
       \left|\{\alpha\in [1,x]: q\nmid \alpha \hbox{ and }\lambda\mid \alpha+y \}\right|=\lfloor \frac{x+y}{\lambda}\rfloor -\lfloor\frac{\lfloor {x}/{q}\rfloor+y}{\lambda} \rfloor.
    \end{equation}
\end{lemma}

\begin{lemma}\label{lll8} 
Suppose that  $x$ and $y$ are two integers such that  $x\leq y$. Then  
\begin{equation}\notag
\left| \{\alpha\in [x,y]:\lambda\mid 2\alpha\hbox{ and }q\nmid \alpha \}\right|= \begin{cases}
\lfloor\frac{y}{\lambda}\rfloor -\lfloor\frac{x-1}{\lambda} \rfloor -  \lfloor \frac{\lfloor y/q\rfloor}{\lambda}\rfloor+\lfloor \frac{\lceil {x}/{q} \rceil-1}{\lambda} \rfloor&\hbox{if $\lambda$ is odd,}\\
 \lfloor\frac{2y}{\lambda}\rfloor -\lfloor\frac{2x-2}{\lambda} \rfloor -  \lfloor \frac{2\lfloor y/q\rfloor}{\lambda}\rfloor+\lfloor \frac{2\lceil {x}/{q} \rceil-2}{\lambda} \rfloor&\hbox{if $ \lambda$ is even.}     
\end{cases}
\end{equation}
\end{lemma}

\begin{lemma}\label{le0}
 Suppose that   $x$ and $y$ are two integers. Then  
 \begin{equation}\notag
  \sum\limits_{t=1}^{q-1}\left [\lfloor \frac{t+x}{\lambda}\rfloor- \lfloor \frac{t+y}{\lambda}\rfloor \right] =\frac{(q-1)(x-y)}{\lambda}.
\end{equation}  
\end{lemma}

\begin{lemma}\label{adl}
Suppose that $a $ is a positive  integer. Then  
 \begin{equation}\notag
  \sum\limits_{t=q, q\nmid t }^{aq}\left [\lfloor \frac{\lfloor tq^{-1}\rfloor+t}{\lambda}\rfloor- \lfloor \frac{t}{\lambda}\rfloor \right] =\frac{a(a-1)(q-1)}{2\lambda}.
\end{equation}  
\end{lemma}

\begin{lemma}\label{ll9}  Suppose that $x$ is an even integer. Then
\begin{equation}\notag
{\sum\limits_{t=1}^{q-1}\left[\lfloor \frac{2t+x}{\lambda}\rfloor- \lfloor \frac{t+x}{\lambda}\rfloor\right]} =
  \begin{cases}
  \frac{q(q-1)}{2\lambda}&\hbox{ if $\lambda$ is odd},\\
    \frac{(q-1)(q+1)}{2\lambda}&\hbox{ if $\lambda$ is even}.\\
  \end{cases} 
\end{equation}
\end{lemma}

\begin{lemma}\label{lema12}
Let $k$ and $a$   be two positive integers with $a\leq q$.   
    Then 
\begin{equation}\notag
     \sum\limits_{t=q^{k},q\nmid t}^{aq^{k}-1}\left[\lfloor \frac{2t}{\lambda}\rfloor -\lfloor \frac{\lfloor tq^{-1}\rfloor+t}{\lambda}\rfloor \right]= {\frac{1}{2\lambda}(a^2-1)(q-1)^2q^{2k-2}}+\textstyle{\frac{(3+(-1)^{\lambda})}{4\lambda}(a-1)(q-1)q^{k-1}}.
 \end{equation} 
\end{lemma}

\begin{lemma}\label{lead2}
Let $k$ be an integer, and let $a$ be  a positive   integer with $a\leq q $.  
Then 
\begin{equation}\notag
\sum\limits_{t=q^{k}}^{aq^{k}-1}N(t+1)=
\begin{cases}
\frac{1}{2}a(a-1) &\hbox{if } k=0,\\
\frac{1}{2}(a^2-1)(q-1)q^{2k-1} &\hbox{if }  k \geq 1.\\
\end{cases}
\end{equation}
\end{lemma}

\begin{lemma}\label{le5}
Let $k$  and $a$   be two  positive integers with $a\leq q$. Then
\begin{equation}\label{1}
     \sum\limits_{t=q^{k-i},q\nmid t}^{aq^{k-i}-1}\left[\lfloor \frac{\lfloor tq^{2i-1}\rfloor+t}{\lambda}\rfloor -\lfloor \frac{\lfloor tq^{2i-2}\rfloor+t}{\lambda}\rfloor\right]= \begin{cases}
         \frac{1}{2\lambda}(a^2-1)(q-1)^2q^{2k-3} &\hbox{if } -k+2\leq i\leq  k-1,\\
          \frac{1}{2\lambda}a(a-1)(q-1)q^{2k-2} & \hbox{if } i=k \hbox{ or }-k+1,
     \end{cases}
 \end{equation} 
and 
\begin{equation}\label{1.5}
     \sum\limits_{t=q^{k-i},q\nmid t}^{aq^{k-i}-1} \left[\lfloor \frac{\lfloor tq^{2i}\rfloor+t}{\lambda}\rfloor -\lfloor \frac{\lfloor tq^{2i-1}\rfloor+t}{\lambda}\rfloor\right]= \begin{cases}
         \frac{1}{2\lambda}(a^2-1)(q-1)^2q^{2k-2}  &\hbox{if } -k+1\leq i\leq  k-1\hbox{ and }i\neq 0,\\
         \frac{1}{2\lambda}a(a-1)(q-1)q^{2k-1} & \hbox{if } i=k \hbox{ or } -k.
     \end{cases}
 \end{equation} 
\end{lemma}

\section{The dimension of narrow-sense BCH codes of length \texorpdfstring{$(q^m-1)/\lambda$} \ }\label{dim}

Suppose that  $m\geq 4$.   Let $\delta$ be an integer satisfying $2\leq \delta\leq  \frac{q^{\lfloor (2m-1)/3\rfloor+1}-1}{\lambda}+1$.   Define $k_{\delta} = \lfloor \log_q \lambda (\delta-1) \rfloor - h$. This implies  that   $q^{h+k_{\delta}}\leq \lambda(\delta-1)<q^{h+k_{\delta}+1}$.  Write the $q$-adic expansion of $\lambda(\delta-1)$ as  $ \lambda(\delta-1)= \sum\limits_{\ell=0}^{h+k_{\delta}}\delta_{\ell}q^{\ell}$.   If $k_{\delta}\geq m-2h$, let  $s_{\delta}$ denote  the smallest integer in $[m-2h-k_{\delta}, k_{\delta}]$ such that $\delta_{h+s_{\delta}}>0$ and let   $w_{\delta}=\sum\limits_{\ell=h+s_{\delta}}^{h+k_{\delta}}\delta_{\ell}q^{\ell}$.

If $m$ is odd,  for each integer $\delta\in \left[2, \frac{q^{\lfloor (2m-1)/3\rfloor+1}-1}{\lambda}+1\right]$, we define the function $f(\delta)$ by 
\begin{equation} \label{ff}
f(\delta)=
\begin{cases}                  
0  &    \hbox{if }
\delta \leq \frac{q^{h+1}-1}{\lambda}+1, \\
 { \left(\begin{aligned}&\frac{ (q-1)^2}{\lambda}(k_{\delta}-1)q^{2k_{\delta}-3}+\lfloor \frac{\mu(\delta)+w_{\delta}}{\lambda}\rfloor -\lfloor \frac{\lfloor \mu(\delta)/q\rfloor+w_{\delta}}{\lambda}\rfloor \\
 &+ \sum\limits_{i=-k_{\delta}+1}^{k_{\delta}}\sum\limits_{t\in \mathcal{T}_i(\delta)} \left [\lfloor \frac{\lfloor tq^{2i-1}\rfloor+t}{\lambda}\rfloor -\lfloor \frac{\lfloor tq^{2i-2}\rfloor+t}{\lambda}\rfloor \right]
  \end{aligned}\right)} &\hbox{if }
\delta >\frac{q^{h+1}-1}{\lambda}+1,  
\end{cases}
\end{equation}
where $\mu(\delta)=\min\left \{\sum\limits_{\ell=0}^{h-k_{\delta}}\delta_{\ell}q^{\ell}, \sum\limits_{\ell=h-s_{\delta}+1}^{h+k_{\delta}}\delta_{\ell}q^{\ell-(h-s_{\delta}+1)}  \right\}$ and $\mathcal{T}_i(\delta)=\left\{t\in \mathbb{Z}: q^{k_{\delta}-i} \leq t  < \sum\limits_{\ell=h+s_{\delta}}^{h+k_{\delta}}\delta_{\ell}q^{\ell-h-i}\hbox{ and }q\nmid t\right\}. 
$

If  $m$ is even, for each integer $\delta\in \left[2, \frac{q^{\lfloor (2m-1)/3\rfloor+1}-1}{\lambda}+1\right]$, we define  the  function $\widetilde{f}(\delta)$ by 
 \begin{equation}\label{wf}
    \widetilde{f}(\delta)=\begin{cases}
        0 &\hbox{if }\delta\leq \frac{q^h-1}{\lambda}+1,\\     
         \lfloor \frac{\widetilde{\mu}(\delta)+w_{\delta}}{\lambda}\rfloor-\lfloor \frac{\lfloor\widetilde{\mu}(\delta)/q\rfloor+w_{\delta}}{\lambda}\rfloor
     +\sum\limits_{t=1}^{\delta_h-1}\left[ \lfloor \frac{2t}{\lambda}\rfloor - \lfloor \frac{t}{\lambda}\rfloor \right] &\hbox{if }\frac{q^h-1}{\lambda}+1<\delta\leq \frac{q^{h+1}-1}{\lambda}+1,\\
 \left({\begin{aligned}& q^{2k_{\delta}-2}(k_{\delta}-\textstyle\frac{1}{2})\textstyle\frac{(q-1)^2}{\lambda}  + \frac{q-1}{2\lambda}\left(q^{k_{\delta}-1}+ \textstyle\frac{1+(-1)^{\lambda}}{2}\right)\\& +  \lfloor \frac{\widetilde{\mu}(\delta)+w_{\delta}}{\lambda}\rfloor-\lfloor\textstyle \frac{\lfloor\widetilde{\mu}(\delta)/q\rfloor+w_{\delta}}{\lambda}\rfloor  \\
        &  +\sum\limits_{i=-k_{\delta}}^{k_{\delta}}\sum\limits_{t\in {\mathcal{T}}_i(\delta)}\left[\lfloor\textstyle \frac{\lfloor tq^{2i}\rfloor+t}{\lambda}\rfloor -\lfloor \textstyle\frac{\lfloor tq^{2i-1}\rfloor+t}{\lambda}\rfloor 
        \right] 
   \end{aligned}}\right)  
     &\hbox{if }\delta >\frac{q^{h+1}-1}{\lambda}+1, 
    \end{cases}
\end{equation}
where
$\widetilde{\mu}(\delta)=\min\left \{\sum\limits_{\ell=0}^{h-k_{\delta}-1}\delta_{\ell}q^{\ell}, \sum\limits_{\ell=h-s_{\delta}}^{h+k_{\delta}}\delta_{\ell}q^{\ell-(h-s_{\delta})}  \right\}.$
Define the function  $\tau(\delta)$ for each integer $\delta\in \left[2, \frac{q^{\lfloor (2m-1)/3\rfloor+1}-1}{\lambda}+1\right]$ by 
\begin{equation}\label{tau}
    \tau(\delta)=\begin{cases}
        1 &\hbox{if }\sum\limits_{\ell=h}^{h+k_{\delta}}\delta_{\ell}q^{\ell-h}\leq \sum\limits_{\ell=0}^{h-1}\delta_{\ell}q^{\ell},\\
        & \delta_h>0 \hbox{ and } \lambda \mid 2\sum\limits_{\ell=h}^{h+k_{\delta}}\delta_{\ell},\\
        0 &\hbox{otherwise.}   
    \end{cases}
\end{equation}
Additionally, we define the function $g(\delta)$ for each integer $\delta\in \left[2, \frac{q^{\lfloor (2m-1)/3\rfloor+1}-1}{\lambda}+1\right]$ by 
\begin{equation}\label{g}
    g(\delta)= \begin{cases}
    0&\hbox{if }\delta\leq \frac{q^h-1}{\lambda}+1,\\
\lfloor \frac{(\delta_h-1)\left(3+(-1)^{\lambda}
            \right)}{2\lambda}\rfloor+\tau(\delta)
    & \hbox{if }\frac{q^h-1}{\lambda}+1<\delta\leq \frac{q^{h+1}-1}{\lambda}+1,\\
        \lfloor  \frac{\phi(\delta)\left(3+(-1)^{\lambda}
            \right)}{2\lambda} \rfloor -\lfloor  \frac{\lfloor \phi(\delta)/q \rfloor\left(3+(-1)^{\lambda}
            \right) }{2\lambda}   \rfloor +\tau(\delta) &\hbox{if } \frac{q^{h+1}-1}{\lambda}+1<\delta, 
    \end{cases}
\end{equation}
where $\phi(\delta)=\sum\limits_{\ell=h}^{h+k_{\delta}}\delta_{\ell}q^{\ell-h}-1$.
\begin{theorem}\label{odda}
Let $m\ge 4$ be an integer,  let $\lambda$ be a positive divisor of $q-1$, and  let $\delta$ be an integer such that $
2 \le \delta \le \frac{q^{\lfloor (2m-1)/3\rfloor+1}-1}{\lambda}+1.
$
Set $n=(q^{m}-1)/\lambda$.  
\begin{itemize}
    \item 
If $m$ is odd, then
\begin{equation}\label{th31}
\mathrm{dim}\bigl(\mathcal{C}_{(q,n,\delta)}\bigr)
= n - m\bigl[N(\delta)-f(\delta)\bigr].
\end{equation}
\item 
If $m$ is even, then
\begin{equation}\label{th311}
\mathrm{dim}\bigl(\mathcal{C}_{(q,n,\delta)}\bigr)
= n - m\bigl[N(\delta)-\widetilde{f}(\delta)\bigr]
- \frac{m}{2}g(\delta). 
\end{equation}
Here  $N(\delta) = \delta -1 - \lfloor (\delta-1)/q\rfloor$ denotes the number of integers in $[1,\delta-1] $ not divisible by $q$, and   $f(\delta)$,   $\widetilde{f}(\delta)$ and $g(\delta)$ are defined in (\ref{ff}), (\ref{wf}) and (\ref{g}) respectively. 
\end{itemize}
\end{theorem}

The proof of Theorem~\ref{odda} proceeds via Assertions~\ref{as1}--\ref{as3.5}. Readers may find the proofs for these assertions  in Appendices \ref{API}--\ref{APM}.  

\begin{assertion}\label{as1}
Suppose that    $ \frac{q^{m-h}-1}{\lambda}+1< \delta\leq  \frac{q^{\lfloor (2m-1)/3\rfloor+1}-1}{\lambda}+1.$ 
\begin{itemize}
\item 
If $m$ is odd, then 
\begin{equation}\label{as1e0} 
\left|\left[\sum\limits_{\ell=h+s_{\delta}}^{h+k_{\delta}}\delta_{\ell}q^{\ell}, \lambda(\delta-1)\right]\cap \mathcal{S}\cap \mathcal{D}_{\lambda}\right|=\lfloor \frac{\mu(\delta)+w_{\delta}}{\lambda}\rfloor -\lfloor \frac{\lfloor \mu(\delta)/q\rfloor+w_{\delta}}{\lambda}\rfloor. 
\end{equation}
\item If $m$ is even, then \begin{equation}\label{as1e0-1} 
\left|\left[\sum\limits_{\ell=h+s_{\delta}}^{h+k_{\delta}}\delta_{\ell}q^{\ell}, \lambda(\delta-1)\right]\cap (\mathcal{S}\cup \mathcal{H})\cap \mathcal{D}_{\lambda}\right|=  \lfloor \frac{\widetilde{\mu}(\delta)+w_{\delta}}{\lambda}\rfloor-\lfloor \frac{\lfloor\widetilde{\mu}(\delta)/q\rfloor+w_{\delta}}{\lambda}\rfloor
\end{equation}
and 
\begin{equation}\label{as1e0-2} 
\left|\left[\sum\limits_{\ell=h}^{h+k_{\delta}}\delta_{\ell}q^{\ell}, \lambda(\delta-1)\right]\cap  \mathcal{H}\cap \mathcal{D}_{\lambda}\right|=\tau(\delta).
\end{equation}
\end{itemize}
\end{assertion}

\begin{assertion}\label{as4}
 Suppose that  $ \frac{q^{m-h}-1}{\lambda}+1<\delta\leq  \frac{q^{\lfloor (2m-1)/3\rfloor+1}-1}{\lambda}+1.$    
    \begin{itemize}
 \item   If $m$ is odd, then 
\begin{equation}\label{as2e}
        \left|\left[q^{h+k_{\delta}}, \sum\limits_{\ell=h+s_{\delta}}^{h+k_{\delta}}
        \delta_{\ell}q^{\ell}\right)\cap \mathcal{S}\cap \mathcal{D}_{\lambda}\right|=\sum\limits_{i=-k_{\delta}+1}^{k_{\delta}}\sum\limits_{t\in \mathcal{T}_i(\delta)} \left [\lfloor \frac{\lfloor tq^{2i-1}\rfloor+t}{\lambda}\rfloor -\lfloor \frac{\lfloor tq^{2i-2}\rfloor+t}{\lambda}\rfloor \right]. 
    \end{equation} 
\item If $m$ is even, then 
\begin{equation}\label{as2e2}
        \left|\left[q^{h+k_{\delta}}, \sum\limits_{\ell=h+s_{\delta}}^{h+k_{\delta}}
        \delta_{\ell}q^{\ell}\right)\cap (\mathcal{S}\cup \mathcal{H})\cap \mathcal{D}_{\lambda}\right|=\sum\limits_{i=-k_{\delta}}^{k_{\delta}}\sum\limits_{t\in {\mathcal{T}}_i(\delta)}\left[\lfloor \frac{\lfloor tq^{2i}\rfloor+t}{\lambda}\rfloor -\lfloor \frac{\lfloor tq^{2i-1}\rfloor+t}{\lambda}\rfloor  \right]
    \end{equation} 
    and 
    \begin{equation}\label{as2e3} 
        \left|\left[q^{h+k_{\delta}}, \sum\limits_{\ell=h}^{h+k_{\delta}}
        \delta_{\ell}q^{\ell}\right)\cap  \mathcal{H} \cap \mathcal{D}_{\lambda}\right|= \begin{cases}
            \lfloor \frac{(\delta_h-1)\left(3+(-1)^{\lambda}
            \right)}{2\lambda} \rfloor &\hbox{if }k_{\delta}=0, \\
        \lfloor  \frac{\phi(\delta)\left(3+(-1)^{\lambda}
            \right)}{2\lambda} \rfloor -\lfloor  \frac{\lfloor \phi(\delta)/q \rfloor\left(3+(-1)^{\lambda}
            \right) }{2\lambda}   \rfloor  -\frac{q^{k_{\delta}-1}(q-1)\left(3+(-1)^{\lambda}
            \right)}{2\lambda} &\hbox{if }k_{\delta}\geq 1,
        \end{cases}
    \end{equation} 
 \end{itemize}
where $\phi(\delta)=\sum\limits_{\ell=h}^{h+k_{\delta}}\delta_{\ell}q^{\ell-h}-1$.
 \end{assertion}

\begin{assertion}\label{as2}
  Let $k$ be an integer such that  $1\leq  k\leq \lfloor (2m-1)/3\rfloor-h$. 
    \begin{itemize}
        \item 
If $m$ is odd, then 
    \begin{equation}\label{as3ee0}
        |\mathcal{A}_{k}(i)\cap \mathcal{D}_{\lambda}|=\begin{cases}
         \frac{1}{2\lambda}(q-1)^3(q+1)q^{2k-3}  &\hbox{if } -k+2\leq i \leq  k-1,\\
         \frac{1}{2\lambda}(q-1)^2q^{2k-1} &\hbox{if }i=-k+1 \hbox{ or }k.
     \end{cases}
    \end{equation}     
\item If $m$ is even, then 
    \begin{equation}\label{as3ee05}
        |\mathcal{B}_{k}(i)\cap \mathcal{D}_{\lambda}|=\begin{cases}
         \frac{1}{2\lambda}(q-1)^3(q+1) q^{2k-2}  &\hbox{if } -k+1 \leq i \leq k-1\hbox{ and }i\neq 0,\\
         \frac{1}{2\lambda}(q-1)^2q^{2k} &\hbox{if }i=-k \hbox{ or }k.
     \end{cases}
    \end{equation}  
     \end{itemize}
 \end{assertion}

\begin{assertion}\label{as3}
Suppose that $m$ is even.    Let  $k$ be  an integer such that   $0\leq k\leq \lfloor (2m-1)/3\rfloor-h]$. 
 \begin{itemize}
     \item 
 If 
 $\lambda$ is odd, then 
\begin{equation}\notag
    |B_k(0)\cap \mathcal{D}_{\lambda}| =  \begin{cases}
   \frac{q(q-1)}{2\lambda} &\hbox{if $k=0$,}\\
\frac{(q-1)^2}{2\lambda}(q^{2k}-q^{2k-2}+q^{k-1})
&\hbox{if $k\geq 1$.}
\end{cases}
 \end{equation} 
\item If $\lambda$ is even, then 
 \begin{equation}\notag
    |B_k(0)\cap \mathcal{D}_{\lambda}| =  \begin{cases}
   \frac{(q+1)(q-1)}{2\lambda} &\hbox{if $k=0$,}\\
\frac{(q-1)^2}{2\lambda}(q^{2k}-q^{2k-2}+2q^{k-1})
&\hbox{if $k\geq 1$.}
\end{cases}
 \end{equation} 
 \end{itemize}
\end{assertion}

\begin{assertion}\label{as3.5}
Suppose that $m$ is even. Let $k$ be an integer such that   $0\leq k\leq \lfloor (2m-1)/3\rfloor-h]$. 
 \begin{itemize}
 \item    If $\lambda$ is odd, then 
 \begin{equation}
   \left| \left[q^{h+k}, q^{h+k+1}\right)\cap \mathcal{H}\cap \mathcal{D}_{\lambda}\right| =\begin{cases}
    \frac{q-1}{\lambda}  & \hbox{if }k=0, \\
       \frac{(q-1)^2q^{k-1}}{\lambda} & \hbox{if }k\geq 1.      
   \end{cases} 
\end{equation}
\item If $\lambda$ is even, then 
\begin{equation}
   \left| \left[q^{h+k}, q^{h+k+1}\right)\cap \mathcal{H}\cap \mathcal{D}_{\lambda}\right| =\begin{cases}
\frac{2(q-1)}{\lambda}     & \hbox{if }k=0,\\
         \frac{2(q-1)^2q^{k-1}}{\lambda}& \hbox{if }k\geq 1.    
   \end{cases} 
\end{equation}

 \end{itemize}
\end{assertion}

\noindent\textbf{Proof of Theorem \ref{odda}.} Suppose that $m$ is odd. We first aim to show that 
\begin{equation}\label{dz}
   \left| \left[1,\lambda(\delta-1)\right]\cap \mathcal{S}\cap \mathcal{D}_{\lambda} \right|=f(\delta) \quad \hbox{for } 2\leq \delta\leq \frac{q^{\lfloor (2m-1)/3\rfloor+1}-1}{\lambda}+1.
\end{equation}

If  $ 2\leq \delta \leq  \frac{q^{h+1}-1}{\lambda}+1$, then  $\lambda\leq \lambda(\delta-1)\leq  q^{h+1}-1$. By applying Lemma \ref{lll}, we obtain 
\begin{equation}\notag
    [1,\lambda(\delta-1)]\cap \mathcal{S}\cap \mathcal{D}_{\lambda}=0.
\end{equation}
In particular, 
    \begin{equation}\label{zjd}
  \left|\left[1,q^{h+1}-1\right]\cap \mathcal{S}\cap \mathcal{D}_{\lambda}\right| =0.
\end{equation}

If   $\frac{q^{h+1}-1}{\lambda}+1<\delta\leq \frac{q^{\lfloor (2m-1)/3\rfloor+1}-1}{\lambda}+1$,  then 
$q^{h+1}\leq \lambda(\delta-1)<  q^{\lfloor (2m-1)/3\rfloor+1}$. This  implies  $1\leq k_{\delta}\leq \lfloor(2m-1)/{3} \rfloor-h$. 
By applying Lemma  \ref{th2}, we can conclude from  Assertion \ref{as2} that 
 for any integer $ k\in \left[1,  \lfloor ({2m-1})/{3}\rfloor-h\right],$ 
\begin{equation}
  \begin{split}
 \left|\left[q^{h+k}, q^{h+k+1} \right)\cap \mathcal{S}\cap \mathcal{D}_{\lambda}\right| 
&=\sum\limits_{i=-k+1}^k\left |\mathcal{A}_k(i)\cap \mathcal{D}_{\lambda}\right|  \\
&=  \frac{1}{\lambda}\left[(k-1)q^{2k-3}(q-1)^3(q+1) \nonumber + q^{2k-1}(q-1)^2\right]. 
 \end{split}  
\end{equation}
Therefore, 
\begin{equation}\notag
\begin{split}
\left|\left[q^{h+1},q^{h+k_{\delta}}\right)\cap \mathcal{S}\cap \mathcal{D}_{\lambda}\right|  & =  \sum\limits_{{k}=1}^{k_{\delta}-1}
\left|\left[q^{h+{k}},q^{h+{k}+1}\right)\cap \mathcal{S}\cap \mathcal{D}_{\lambda}\right|  \\ 
&=\sum\limits_{{k}=1}^{k_{\delta}-1}\frac{(q-1)^3}{\lambda}(q+1)({k}-1) q^{2{k}-3}+ \sum\limits_{{k}=1}^{k_{\delta}-1}\frac{q^{2{k}-1}(q-1)^2}{\lambda}.
\end{split}
\end{equation}
It is straightforward  to verify  that 
\begin{equation}\notag
\sum\limits_{{k}=1}^{k_{\delta}-1}\frac{q^{2{k}-1}(q-1)^2}{\lambda}=\frac{(q^{2{k_{\delta}}-1}-q)(q-1)}{\lambda(q+1)}
\end{equation}
and 
\begin{equation}\notag
   \sum\limits_{k=1}^{k_{\delta}-1}\frac{(q-1)^3}{\lambda}(q+1)({k}-1)q^{2{k}-3}=  \frac{\left[q-(k_{\delta}-1)q^{2k_{\delta}-3}+(k_{\delta}-2)q^{2k_{\delta}-1}\right](q-1)}{\lambda(q+1)}.
\end{equation}
By adding these two sums,  we can obtain 
\begin{equation}\label{ssxx}
\left|\left[q^{h+1},q^{h+k_{\delta}}\right)\cap \mathcal{S}\cap \mathcal{D}_{\lambda}\right| =\frac{ (q-1)^2}{\lambda}(k_{\delta}-1)q^{2k_{\delta}-3}.
\end{equation}
Combining equations (\ref{as1e0}), 
(\ref{as2e}), 
(\ref{zjd}) and 
(\ref{ssxx})   we obtain 
\begin{equation}
\begin{split}
\left|\left[1,\lambda(\delta-1)\right]\cap \mathcal{S}\cap \mathcal{D}_{\lambda}\right| 
 = &\frac{ (q-1)^2}{\lambda}(k_{\delta}-1)q^{2k_{\delta}-3}+\lfloor \frac{\mu(\delta)+w_{\delta}}{\lambda}\rfloor -\lfloor \frac{\lfloor \mu(\delta)/q\rfloor+w_{\delta}}{\lambda}\rfloor   \nonumber \\& +  \sum\limits_{i=-k_{\delta}+1}^{k_{\delta}}\sum\limits_{t\in \mathcal{T}_i(\delta)} \left [\lfloor \frac{\lfloor tq^{2i-1}\rfloor+t}{\lambda}\rfloor -\lfloor \frac{\lfloor tq^{2i-2}\rfloor+t}{\lambda}\rfloor \right]. 
  \end{split}
\end{equation}
Recalling the definition of $f(\delta)$, we can now claim that equation (\ref{dz}) holds.

Next, we establish equation (\ref{th31}). 
Noticing that
$\lfloor ({2m-1})/{3}\rfloor+1\leq m-\lfloor {m}/{3}\rfloor$, we have $\lambda(\delta-1)< q^{m-\lfloor {m}/{3}\rfloor}$. Therefore, we can apply     Lemmas \ref{cos2} and  \ref{th1}  to obtain 
\begin{equation}\label{j1}
    |C_{n}(a)|= |C_{\lambda n}(\lambda a)|=m  \quad \hbox{for all integers }a\in [1,\delta-1]
\end{equation}
and 
\begin{equation}\notag
     \mathcal{L}_m^{\lambda n}(1,\lambda(\delta-1))= \mathcal{L}^{\lambda n}(1,\lambda(\delta-1)).
\end{equation}
Recalling the definition of $\mathcal{S}$ and noting  that an integer $a$ cannot be the  coset leader of $C_n(a)$ if  $q\mid a$, it follows that 
\begin{align*}
 \left| \mathcal{L}_m^{\lambda n}(1,\lambda(\delta-1))\cap \mathcal{D}_{\lambda}\right| &= \left|\mathcal{L}^{\lambda n}(1,\lambda(\delta-1))\cap \mathcal{D}_{\lambda}\right|\\
 & =
  \left| \left\{a\in [1,\lambda(\delta-1)]: q\nmid a \right \} \cap \mathcal{D}_{\lambda}\right| - \left|\left [1,\lambda(\delta-1)\right]\cap\mathcal{S}\cap \mathcal{D}_{\lambda}\right|.
\end{align*}
It  can be easily verified that 
\begin{equation}\notag
  \left|   \left\{a\in [1,\lambda(\delta-1)]: q\nmid a \right \} \cap \mathcal{D}_{\lambda} \right| =N(\delta).
\end{equation}
Utilizing Corollary \ref{cor1}
and equation  (\ref{dz}), we have 
 \begin{equation}\label{j2}
   \left|\mathcal{L}_m^{n}(1,\delta-1)\right| =  \left |   \mathcal{L}_m^{\lambda n}(1,\lambda(\delta-1))\cap \mathcal{D}_{\lambda} \right| = N(\delta)-f(\delta).
\end{equation}
Recalling the equality in (\ref{dimeq}), we can now conclude from (\ref{j1}) and  (\ref{j2}) that (\ref{th31})  holds. 

Suppose that  $m$ is even. Our first goal is to show that \begin{equation}\label{even1}
    \left|\left[1,\lambda(\delta-1)\right] \cap(\mathcal{S}\cup \mathcal{H})\cap \mathcal{D}_{\lambda}  \right|=\widetilde{f}(\delta)
\end{equation}
and 
\begin{equation}\label{even2}
       \left|\left[1,\lambda(\delta-1)\right] \cap \mathcal{H} \cap \mathcal{D}_{\lambda}  \right|=g(\delta)
\end{equation} for $2\leq \delta\leq  \frac{q^{\lfloor (2m-1)/3\rfloor+1}-1}{\lambda}+1.$

If $2\leq \delta\leq \frac{q^h-1}{\lambda}+1$, then $\lambda\leq \lambda(\delta-1)\leq q^h-1$. By applying Lemma \ref{lll}, we  conclude that 
\begin{equation}\notag
    \left|\left[1,\lambda(\delta-1)\right] \cap(\mathcal{S}\cup \mathcal{H})\cap \mathcal{D}_{\lambda}  \right|=0
\end{equation}
and 
\begin{equation}\notag
       \left|\left[1,\lambda(\delta-1)\right] \cap \mathcal{H} \cap \mathcal{D}_{\lambda}  \right|=0.
\end{equation}
In particular, we have  
\begin{equation}\label{p1}
    \left|\left[1,q^h-1\right] \cap(\mathcal{S}\cup \mathcal{H})\cap \mathcal{D}_{\lambda}  \right|=0 
\end{equation}
and
\begin{equation}\label{p2}
       \left|\left[1, q^h-1\right] \cap \mathcal{H} \cap \mathcal{D}_{\lambda}  \right|=0.
\end{equation}

If $ \frac{q^h-1}{\lambda}+1<\delta\leq \frac{q^{h+1}-1}{\lambda}+1,$ then we have $q^h\leq \lambda(\delta-1)\leq q^{h+1}-1$. In this case, we have 
\begin{itemize}
    \item $k_{\delta}=s_{\delta}=0$;
    \item $\mathcal{T}_0(\delta)=\{t\in \mathbb{Z}: 1\leq t\leq \delta_h-1\}.$
\end{itemize}
By substituting $k_{\delta}$, $s_{\delta}$ and $\mathcal{T}_0(\delta)$ as above  into equations (\ref{as1e0-1}) and (\ref{as2e2}), we obtain 
\begin{equation}\notag
\left|\left[\delta_h q^h,  \lambda(\delta-1)\right]\cap (\mathcal{S}\cup \mathcal{H})\cap \mathcal{D}_{\lambda}\right|=  \lfloor \frac{\widetilde{\mu}(\delta)+w_{\delta}}{\lambda}\rfloor-\lfloor \frac{\lfloor\widetilde{\mu}(\delta)/q\rfloor+w_{\delta}}{\lambda}\rfloor
\end{equation}
and \begin{equation}\notag
\left|\left[q^h,  \delta_hq^h\right)\cap (\mathcal{S}\cup \mathcal{H})\cap \mathcal{D}_{\lambda}\right|= \sum\limits_{t=1}^{\delta_h-1}\left[ \lfloor \frac{2t}{\lambda}\rfloor -\lfloor \frac{t}{\lambda}\rfloor \right]. 
\end{equation}
Combining the above two equalities with (\ref{p1}),  we obtain 
\begin{equation}\label{485}
\left|\left[1, \lambda(\delta-1)\right]\cap (\mathcal{S}\cup \mathcal{H})\cap \mathcal{D}_{\lambda}\right|= \lfloor \frac{\widetilde{\mu}(\delta)+w_{\delta}}{\lambda}\rfloor-\lfloor \frac{\lfloor\widetilde{\mu}(\delta)/q\rfloor+w_{\delta}}{\lambda}\rfloor + \sum\limits_{t=1}^{\delta_h-1}\left[ \lfloor \frac{2t}{\lambda}\rfloor -\lfloor \frac{t}{\lambda}\rfloor \right]. 
\end{equation}
Additionally, by substituting $k_{\delta}=0$ into (\ref{as1e0-2}) and  (\ref{as2e3}), we have 
\begin{equation}\notag
   \left| \left[ \delta_h q^h,\lambda(\delta-1) \right]\cap \mathcal{H} \cap \mathcal{D}_{\lambda} \right| =\tau(\delta)
\end{equation}
and 
\begin{equation}\notag
    \left|\left[q^h, \delta_hq^h\right)\cap \mathcal{H}\cap \mathcal{D}_{\lambda}\right| =
       \lfloor \frac{\left(3+(-1)^{\lambda}
            \right)(\delta_h-1)}{2\lambda}\rfloor.
\end{equation}
Combining these two equalities with (\ref{p2}), we obtain 
\begin{equation}\label{50}
    \left| \left[1, \lambda(\delta-1)\right]\cap \mathcal{H}\cap \mathcal{D}_{\lambda} \right| = \lfloor \frac{\left(3+(-1)^{\lambda}
            \right)(\delta_h-1)}{2\lambda}\rfloor+\tau(\delta).
\end{equation}

Notice that $\lambda( \delta-1)= q^{h+1}-1=\sum\limits_{\ell=0}^h(q-1)q^{\ell}$ when  $\delta=\frac{q^{h+1}-1}{\lambda}+1$. Therefore, for $\delta=\frac{q^{h+1}-1}{\lambda}+1$,  we have 
\begin{itemize}
    \item $\widetilde{\mu}(\delta)=q-1$;
    \item $ w_{\delta}=(q-1)q^h$;
    \item$ \delta_h=q-1.$
\end{itemize} 
Consequently, by substituting  $\delta=\frac{q^{h+1}-1}{\lambda}+1$ into equation (\ref{485}) and  applying Lemma \ref{ll9},  we obtain   
\begin{equation}\label{v7}
\begin{split}
\left|\left[1, q^{h+1}-1\right]\cap (\mathcal{S}\cup \mathcal{H}) \cap \mathcal{D}_{\lambda} \right|&=  \lfloor \frac{(q-1)(q^h+1)}{\lambda}\rfloor-\lfloor \frac{(q-1)q^{h}}{\lambda}\rfloor + \sum\limits_{t=1}^{q-2}\left[ \lfloor \frac{2t}{\lambda}\rfloor -\lfloor \frac{t}{\lambda}\rfloor \right]\\
& = 
\frac{q-1}{2\lambda}\left(q+\frac{1+(-1)^{\lambda
}}{2}\right).
\end{split}
\end{equation}

If 
$\frac{q^{h+1}-1}{\lambda}+1<\delta\leq  \frac{q^{\lfloor (2m-1)/3\rfloor+1}-1}{\lambda}+1$, then we have $q^{h+1}\leq \lambda(\delta-1)< q^{\lfloor (2m-1)/3\rfloor+1}$. In this case, we have 
 $1\leq k_\delta\leq \lfloor ({2m-1})/{3}\rfloor-h.$
By applying Lemma  \ref{th2}, we can conclude from Assertions \ref{as2} and \ref{as3}  that 
\begin{equation} \notag
\left|\left[q^{h+k},q^{h+k+1}\right)\cap (\mathcal{S}\cup \mathcal{H})\cap \mathcal{D}_{\lambda}\right|  =\sum\limits_{i=-k}^k  |\mathcal{B}_k(i) |
    =\frac{(q-1)^3}{\lambda} (k-\frac{1}{2})q^{2k-2}(q+1) 
+ )\frac{(q-1)^2}{\lambda}(\frac{1}{2}q^{k-1}+q^{2k})
\end{equation}
for  each integer $k\in [1, \lfloor ({2m-1})/{3}\rfloor-h]$. 
It follows that
\begin{eqnarray}
\left|\left[q^{h+1},q^{h+k_{\delta}}\right)\cap (\mathcal{S}\cup \mathcal{H})\right| & =&\sum\limits_{k=1}^{k_{\delta}-1}
\left|\left[q^{h+k},q^{h+k+1}\right)\cap  (\mathcal{S}\cup \mathcal{H})\right|  \label{w3} \\ 
&=&\frac{(q-1)^2}{\lambda} (k_{\delta}-\frac{1}{2})q^{2k_{\delta}-2} +\frac{q-1}{2\lambda}(q^{k_{\delta}-1}-q).  \nonumber
\end{eqnarray}
Combining (\ref{as1e0-1}), (\ref{as2e2}), (\ref{v7}) and (\ref{w3}), we derive   
\begin{equation}\notag
\begin{split}
    \left| \left[1,\lambda(\delta-1)\right]\cap (\mathcal{S}\cup \mathcal{H})\cap \mathcal{D}_{\lambda} \right| = &\frac{q-1}{2\lambda}\left(q^{k_{\delta}-1}+ \frac{1+(-1)^{\lambda}}{2}\right) +  \lfloor \frac{\widetilde{\mu}(\delta)+w_{\delta}}{\lambda}\rfloor-\lfloor \frac{\lfloor\widetilde{\mu}(\delta)/q\rfloor+w_{\delta}}{\lambda}\rfloor  \\
        & + \frac{(q-1)^2}{\lambda} (k_{\delta}-\frac{1}{2})q^{2k_{\delta}-2}  +\sum\limits_{i=-k_{\delta}}^{k_{\delta}}\sum\limits_{t\in {\mathcal{T}}_i(\delta)}\left[\lfloor \frac{\lfloor tq^{2i}\rfloor+t}{\lambda}\rfloor -\lfloor \frac{\lfloor tq^{2i-1}\rfloor+t}{\lambda}\rfloor  
        \right] 
        \end{split}
\end{equation}

On the other hand, we conclude from Assertion \ref{as3.5} that 
\begin{equation}\notag
    \left|\left[q^{h}, q^{h+k_{\delta}} \right)\cap \mathcal{H}\cap \mathcal{D}_{\lambda}   \right| =\sum\limits_{k=0}^{k_{\delta}-1}\left|
    \left[q^{h+k}, q^{h+k+1}\right)\cap \mathcal{H}\cap \mathcal{D}_{\lambda}\right| = \frac{(3+(-1)^{\lambda})(q-1)q^{k_{\delta}-1}}{2\lambda}.
\end{equation}
With  (\ref{as1e0-2}) and (\ref{as2e3}), it follows that 
\begin{equation}\notag
     \left|\left[1, \lambda(\delta-1) \right)\cap \mathcal{H}\cap \mathcal{D}_{\lambda}   \right| =  \lfloor  \frac{\phi(\delta)\left(3+(-1)^{\lambda}
            \right)}{\lambda} \rfloor -\lfloor  \frac{\lfloor \phi(\delta)/q \rfloor\left(3+(-1)^{\lambda}
            \right) }{2\lambda}   \rfloor +\tau(\delta),
\end{equation}
where $\phi(\delta)=\sum\limits_{\ell=h}^{h+k_{\delta}}\delta_{\ell}q^{\ell-h}-1$. 

By now, we have already demonstrated that both  (\ref{even1}) and (\ref{even2}) hold for $ 2 \leq \delta\leq  \frac{q^{\lfloor (2m-1)/3\rfloor+1}-1}{\lambda}+1.$  Next, we show that equation (\ref{th311}) holds. 
Notice that
 $\lambda(\delta-1)< q^{m-\lfloor {m}/{3}\rfloor}$. We can apply    Lemmas \ref{cos2} and \ref{th1}  to obtain  \begin{equation}\label{evenm}
    |C_{n}(a)|= |C_{\lambda n}(\lambda a)|= m  \hbox{ or } \frac{m}{2} \quad  \hbox{for all  integers }  a\in [1,\delta-1].
    \end{equation}
By applying Corollary \ref{cor1}, we can derive from (\ref{even1}) and (\ref{even2}) that 
 \begin{equation}\label{jj1}
   \left |   \mathcal{L}^n_m(1,\delta-1)\right| = 
   \left |   \mathcal{L}^{\lambda n}_m\left(1,\lambda(\delta-1)\right)\cap \mathcal{D}_{\lambda} \right| = N(\delta)-\widetilde{f}(\delta)
\end{equation}
and 
 \begin{equation}\label{jj2}
  \left |   \mathcal{L}^{n}_{\frac{m}{2}}(1,\delta-1)\right|= \left |   \mathcal{L}^{\lambda n}_{\frac{m}{2}}\left(1,\lambda(\delta-1)\right)\cap \mathcal{D}_{\lambda} \right| = g(\delta).
\end{equation}
Recalling the equality in (\ref{dimeq}), we can now conclude from (\ref{evenm}), (\ref{jj1}) and (\ref{jj2}) that (\ref{th311})  holds.  This completes the proof. \qed

Zhu et al.\ determined the dimension of narrow-sense BCH codes of length $(q^m-1)/\lambda$ for designed distance $2\leq \delta\leq \frac{q^{\lceil (m+1)/2 \rceil}-1}{\lambda} + 1$ in \cite[Theorem 1]{ZSK2019} for odd $m$ and in \cite[Theorem 3]{ZSK2019} for even  $m$, respectively.
As an illustration of   Theorem \ref{odda} in this paper, we present the dimension of narrow-sense BCH codes of length $(q^m-1)/\lambda$ for  designed distances  $2\leq \delta \leq \frac{q^{\lceil m/2\rceil+1}-1}{\lambda}+1$ in the following two corollaries. 
Notably, for even $m$,  although the range of designed distances considered in Corollary  \ref{CC} coincides with that in
 \cite[Theorem 3]{ZSK2019}, our derivation leads to a simpler expression, involving only four cases compared to  eleven cases in \cite[Theorem 3]{ZSK2019}.

\begin{Corollary}\label{CC}
 Let $m\geq 4$ be an even integer, and let $\delta$ be an integer with $2\leq \delta \leq  \frac{q^{h+1}-1}{\lambda}+1$. 
Then  \begin{equation}\label{cor2e1}
           \mathrm{dim}(\mathcal{C}_{(q,n,\delta)}) =\begin{cases}  n-mN(\delta) \\   \quad \quad\quad \hbox{if } \delta \leq \frac{q^h-1}{\lambda}+1,\\    n-mN(\delta)+m\sum\limits_{t=1}^{\delta_h}\left[ \lfloor \frac{2t}{\lambda} \rfloor -\lfloor \frac{t}{\lambda} \rfloor\right]
      -\frac{m}{2}  \lfloor \frac{\left(3+(-1)^{\lambda}\right)(\delta_h-1)}{2\lambda}\rfloor- \frac{m}{2}\\
         \quad \quad\quad  \hbox{if } \delta > \frac{q^h-1}{\lambda}+1, \delta_h\leq \sum\limits_{\ell=0}^{h-1}\delta_{\ell}q^{\ell} \hbox{ and }  \lambda\mid 2\delta_h, \\
             n-mN(\delta)+m\sum\limits_{t=1}^{\delta_h}\left[ \lfloor \frac{2t}{\lambda} \rfloor -\lfloor \frac{t}{\lambda} \rfloor\right]
      -\frac{m}{2}  \lfloor \frac{\left(3+(-1)^{\lambda}\right)(\delta_h-1)}{2\lambda}\rfloor
            \\ \quad \quad\quad \hbox{if }\delta > \frac{q^h-1}{\lambda}+1, \delta_h\leq \sum\limits_{\ell=0}^{h-1}\delta_{\ell}q^{\ell} 
            \hbox{ and } \lambda\nmid 2\delta_h,\\
n-mN(\delta)+m \left[\sum\limits_{t=1}^{\delta_h}\left[ \lfloor\textstyle \frac{2t}{\lambda} \rfloor -\lfloor\textstyle \frac{t}{\lambda} \rfloor\right] +\lfloor \textstyle\frac{\delta_0+\delta_h}{\lambda}\rfloor-\lfloor\textstyle \frac{2\delta_h}{\lambda}\rfloor\right]
      -\textstyle\frac{m}{2} \lfloor \textstyle\frac{\left(3+(-1)^{\lambda}\right)(\delta_h-1)}{2\lambda}\rfloor
           \\ \quad \quad\quad \hbox{if }\delta > \frac{q^h-1}{\lambda}+1\hbox{ and }\delta_h> \sum\limits_{\ell=0}^{h-1}\delta_{\ell}q^{\ell}.
     \end{cases}
 \end{equation}  
\end{Corollary}
  \begin{proof}  
 If $2\leq \delta \leq \frac{q^h-1}{\lambda}+1$, then we have $\widetilde{f}(\delta)=0$ and $g(\delta)=0$. 
 
If  $\delta>\frac{q^h-1}{\lambda}+1  $, then  we distinguish the following two cases. 

\textbf{Case 1.} Suppose that  $\delta_h\leq \sum\limits_{\ell=0}^{h-1}\delta_{\ell}q^{\ell}$. Then we have  $\widetilde{\mu}(\delta)=\delta_h$ and $w_{\delta}=\delta_hq^h$. By substituting them  into equation (\ref{wf}),  we obtain  
\begin{equation}\notag
\begin{split}
\widetilde{f}(\delta)&=\lfloor \frac{\delta_h+\delta_{h}q^h}{\lambda}\rfloor-\lfloor \frac{\delta_{h}q^h}{\lambda}\rfloor + \sum\limits_{t=1}^{\delta_h-1}\left[ \lfloor \frac{2t}{\lambda} \rfloor -\lfloor \frac{t}{\lambda} \rfloor\right]   \\
&=\sum\limits_{t=1}^{\delta_h}\left[ \lfloor \frac{2t}{\lambda} \rfloor -\lfloor \frac{t}{\lambda} \rfloor\right].
\end{split}
\end{equation}
In addition,  it is straightforward to verify that  
\begin{equation}\notag
\tau(\delta)= 
\begin{cases}
  1 &\hbox{if }\lambda \mid 2\delta_h,\\
      0  & \hbox{if }\lambda \nmid 2\delta_h.
\end{cases}
\end{equation}
Recalling equation (\ref{g}), it follows that  
\begin{equation}\notag
    g(\delta)=\begin{cases}
\lfloor \frac{\left(3+(-1)^{\lambda}\right)(\delta_h-1)}{2\lambda}\rfloor+  1 &\hbox{if }\lambda \mid 2\delta_h,\\
     \lfloor \frac{\left(3+(-1)^{\lambda}\right)(\delta_h-1)}{2\lambda}\rfloor  &  \hbox{if }\lambda \nmid 2\delta_h.
\end{cases}
\end{equation}

\textbf{Case 2.}  Suppose that  $\delta_h>  \sum\limits_{\ell=0}^{h-1}\delta_{\ell}q^{\ell}$.  Then  $\widetilde{\mu}(\delta)=\sum\limits_{\ell=0}^{h-1}\delta_{\ell}q^{\ell}=\delta_0$,   $w_{\delta}=\delta_hq^h$ and $\tau(\delta)=0$. 
By substituting these values  into equations (\ref{wf}) and (\ref{g}),  we obtain  
\begin{equation}\notag
\widetilde{f}(\delta)=\lfloor \frac{\delta_0+\delta_h}{\lambda}\rfloor-\lfloor \frac{2\delta_h}{\lambda}\rfloor+\sum\limits_{t=1}^{\delta_h}\left[ \lfloor \frac{2t}{\lambda} \rfloor -\lfloor \frac{t}{\lambda} \rfloor\right] 
\end{equation}
and 
\begin{equation}\notag
   g(\delta)= \lfloor \frac{\left(3+(-1)^{\lambda}\right)(\delta_h-1)}{2\lambda}\rfloor.
\end{equation}

 Finally, substituting the values of $g(\delta)$ and $\widetilde{f}(\delta)$ for corresponding cases into equation   (\ref{th311}), we derive the desired equation (\ref{cor2e1}). 
  \end{proof}  
 
\begin{Corollary}\label{hlx}
     Let $m\geq 5$ be an odd integer, and let $\delta$  be an integer with $2\leq \delta \leq  \frac{q^{h+2}-1}{\lambda}+1.$    Then  
\begin{equation}\label{cor3e0}
           \mathrm{dim}(\mathcal{C}_{(q,n,\delta)})\! = \!
           \begin{cases} 
           n-mN(\delta)\\
           \quad \quad\quad \hbox{if }\delta \leq \frac{q^{h+1}-1}{\lambda}+1,\\
            n\!-\!mN(\delta)\!+\!m \left(\lfloor\frac{\delta_0+\delta_{h+1}+\delta_h}{\lambda} \rfloor - \lfloor\frac{2\delta_{h+1}+\delta_h}{\lambda} \rfloor
           \!+\! \frac{\delta_{h+1}^2(q-1)}{\lambda}\!+\!\sum\limits_{i=1}^{\delta_h}\left[\lfloor \frac{2\delta_{h+1}+i}{\lambda}\rfloor \!-\! \lfloor \frac{\delta_{h+1}+i}{\lambda}\rfloor \right] \right)  \\
           \quad \quad \quad 
                \hbox{if } \delta> \frac{q^{h+1}-1}{\lambda}+1,  \delta_h>0  \hbox{ and }\delta_{h+1}> \sum\limits_{\ell=0}^{h-1}\delta_{\ell}q^{\ell},\\
            n-mN(\delta)+m \left(\frac{\delta_{h+1}^2(q-1)}{\lambda} +\sum\limits_{i=1}^{\delta_h}\left[\lfloor \frac{2\delta_{h+1}+i}{\lambda}\rfloor - \lfloor \frac{\delta_{h+1}+i}{\lambda}\rfloor \right]\right)
             \\ \quad \quad \quad   \hbox{if } \delta> \frac{q^{h+1}-1}{\lambda}+1,  \delta_h>0  \hbox{ and }\delta_{h+1}\leq  \sum\limits_{\ell=0}^{h-1}\delta_{\ell}q^{\ell},\\
n-mN(\delta)+m\left(\frac{\left(\delta_{h+1}^2-\delta_{h+1}+\delta_1\right)(q-1)}{\lambda}+\lfloor \frac{\delta_1+\delta_0+\delta_{h+1}}{\lambda}\rfloor -\lfloor \frac{\delta_1+\delta_{h+1}}{\lambda}\rfloor\right)\\
\quad \quad \quad  \hbox{if
      }\delta> \frac{q^{h+1}-1}{\lambda}+1, \delta_h=0 \hbox{ and }\delta_{h+1}q>\sum\limits_{\ell=0}^{h-1}\delta_{\ell}q^{\ell},\\
   n-mN(\delta) +m   \frac{\delta_{h+1}^2(q-1)}{\lambda} \\ \quad \quad \quad  \hbox{if
      }\delta> \frac{q^{h+1}-1}{\lambda}+1, \delta_h=0 \hbox{ and } \delta_{h+1}q\leq \sum\limits_{\ell=0}^{h-1}\delta_{\ell}q^{\ell}.
     \end{cases}
 \end{equation}
\end{Corollary}

\begin{proof}
   If  $\delta \leq \frac{q^{h+1}-1}{\lambda}+1, $ we  directly have $f(\delta)=0$. 
   
     If  $\delta >\frac{q^{h+1}-1}{\lambda}+1, $ then we have $k_{\delta}=1$.  We distinguish following cases. 

 \textbf{Case 1.} Suppose that $\delta_h>0$.  In this case, we have $s_{\delta}=0$. This leads to  
 \begin{itemize}
     \item $\mathcal{T}_0(\delta)=\left\{t\in \mathbb{Z}: q\leq  t\leq  \delta_{h+1}q+ \delta_h-1\hbox{ and }q\nmid t\right\}$; 
     \item 
     $\mathcal{T}_1(\delta)=\left\{t\in \mathbb{Z}:1 \leq t \leq \delta_{h+1}\right\}$;
     \item $w_{\delta}=\delta_{h+1}q^{h+1}+\delta_hq^h$;
      \item  $\mu(\delta)=\min\left\{\delta_{h+1},\sum\limits_{\ell=0}^{h-1}\delta_{\ell}q^{\ell}\right\}$.
      \end{itemize}
     By substituting these values into equation (\ref{ff}), we obtain  
\begin{equation}\label{cor3e1} 
\begin{split}
f(\delta)=&\lfloor \frac{\mu(\delta)+w_{\delta}}{\lambda}\rfloor - \lfloor \frac{{\lfloor \mu(\delta})/q\rfloor+w_{\delta}}{\lambda}\rfloor  + \sum\limits_{t=q, q\nmid t}^{\delta_{h+1}q+\delta_{h}-1}\left[ \lfloor\frac{\lfloor tq^{-1}\rfloor+t}{\lambda}\rfloor - \lfloor\frac{t}{\lambda}\rfloor\right]+ \sum\limits_{t=1}^{\delta_{h+1}}\left[ \lfloor\frac{ tq+t}{\lambda}\rfloor - \lfloor\frac{2t}{\lambda}\rfloor\right]. 
\end{split}
\end{equation} 
Note that  $\mu(\delta)=\delta_{h+1}$  if $\delta_{h+1}\leq \sum\limits_{\ell=0}^{h-1}\delta_{\ell}q^{\ell}$, and $\mu(\delta)=\sum\limits_{\ell=0}^{h-1}\delta_{\ell}q^{\ell}=\delta_0$ if  $\delta_{h+1}> \sum\limits_{\ell=0}^{h-1}\delta_{\ell}q^{\ell}$.  Thus, 
\begin{equation}\label{588}
    \lfloor \frac{\mu(\delta)+w_{\delta}}{\lambda}\rfloor - \lfloor \frac{{\lfloor \mu(\delta})/q\rfloor+w_{\delta}}{\lambda}\rfloor =\begin{cases}
\lfloor\frac{\delta_0+\delta_{h+1}+\delta_h}{\lambda} \rfloor - \lfloor\frac{\delta_{h+1}+\delta_h}{\lambda} \rfloor& \hbox{if } \delta_{h+1}> \sum\limits_{\ell=0}^{h-1}\delta_{\ell}q^{\ell},\\
 \lfloor\frac{2\delta_{h+1}+\delta_h}{\lambda} \rfloor - \lfloor\frac{\delta_{h+1}+\delta_h}{\lambda} \rfloor  & \hbox{if } \delta_{h+1}\leq \sum\limits_{\ell=0}^{h-1}\delta_{\ell}q^{\ell}.
    \end{cases}
\end{equation}
By applying Lemma \ref{adl}, we derive \begin{equation}\notag
     \sum\limits_{t=q,q\nmid t}^{\delta_{h+1}q}\left[ \lfloor\frac{\lfloor tq^{-1}\rfloor+t}{\lambda}\rfloor - \lfloor\frac{t}{\lambda}\rfloor\right] =  \frac{\delta_{h+1}(\delta_{h+1}-1)(q-1)}{2\lambda}. 
\end{equation}
Since each integer $t\in \left[\delta_{h+1}q+1,\delta_{h+1}q+\delta_h-1\right]$ can be uniquely expressed as $t=\delta_{h+1}q+i$ with an integer $i\in \left[1,\delta_h-1\right]$,  we have 
    \begin{equation} \notag  \sum\limits_{t=\delta_{h+1}q+1}^{\delta_{h+1}q+\delta_h-1}\left[ \lfloor\frac{\lfloor tq^{-1}\rfloor+t}{\lambda}\rfloor - \lfloor\frac{t}{\lambda}\rfloor\right] =    \sum\limits_{i=1}^{\delta_h-1}\left[\lfloor \frac{2\delta_{h+1}+i}{\lambda}\rfloor - \lfloor \frac{\delta_{h+1}+i}{\lambda}\rfloor \right].
\end{equation}
Adding above two sums, we get  
\begin{equation}\label{59}
     \sum\limits_{t=q, q\nmid t}^{\delta_{h+1}q+\delta_{h}-1}\left[ \lfloor\frac{\lfloor tq^{-1}\rfloor+t}{\lambda}\rfloor - \lfloor\frac{t}{\lambda}\rfloor\right]= \frac{\delta_{h+1}(\delta_{h+1}-1)(q-1)}{2\lambda} + \sum\limits_{i=1}^{\delta_h-1}\left[\lfloor \frac{2\delta_{h+1}+i}{\lambda}\rfloor - \lfloor \frac{\delta_{h+1}+i}{\lambda}\rfloor \right].
\end{equation}
Furthermore, it is straightforward to verify that 
\begin{equation} \label{60}
\begin{split} 
\sum\limits_{t=1}^{\delta_{h+1}}\left[ \lfloor\frac{ tq+t}{\lambda}\rfloor - \lfloor\frac{2t}{\lambda}\rfloor\right]&= \sum\limits_{t=1}^{\delta_{h+1}}\frac{t(q-1)}{\lambda}\\
&= \frac{\delta_{h+1}(\delta_{h+1}+1)(q-1)}{2\lambda}. 
\end{split}
\end{equation}
We can now conclude from (\ref{cor3e1}) -- (\ref{60}) that 
\begin{equation}\notag
f(\delta)= 
\begin{cases} 
\lfloor\frac{\delta_0+\delta_{h+1}+\delta_h}{\lambda} \rfloor - \lfloor\frac{\delta_{h+1}+\delta_h}{\lambda} \rfloor  + \frac{\delta_{h+1}^2(q-1)}{\lambda}+\sum\limits_{i=1}^{\delta_h-1}\left[\lfloor \frac{2\delta_{h+1}+i}{\lambda}\rfloor - \lfloor \frac{\delta_{h+1}+i}{\lambda}\rfloor \right]&  \hbox{if } \delta_h>0\hbox{ and }\delta_{h+1}> \sum\limits_{\ell=0}^{h-1}\delta_{\ell}q^{\ell},\\
 \frac{\delta_{h+1}^2(q-1)}{\lambda}+\sum\limits_{i=1}^{\delta_h}\left[\lfloor \frac{2\delta_{h+1}+i}{\lambda}\rfloor - \lfloor \frac{\delta_{h+1}+i}{\lambda}\rfloor \right]&  \hbox{if } \delta_h>0\hbox{ and }\delta_{h+1}\leq  \sum\limits_{\ell=0}^{h-1}\delta_{\ell}q^{\ell}.
\end{cases}
\end{equation}

 \textbf{Case 2.} Suppose that  $\delta_h=0$. Then we have $s_{\delta}=1$.  This  leads to 
\begin{itemize}
     \item $\mathcal{T}_0(\delta)=\left\{t\in \mathbb{Z}: q\leq t\leq  \delta_{h+1}q-1\hbox{ and }q\nmid t\right\}$;
     
    \item  $\mathcal{T}_1(\delta)=\{t\in \mathbb{Z}:1\leq t\leq \delta_{h+1}-1\}$;
    \item $w_{\delta}=\delta_{h+1}q^{h+1}$;
    \item  $\mu(\delta)=\min\left\{ \delta_{h+1}q, \sum\limits_{\ell=0}^{h-1}\delta_{\ell}q^{\ell} \right\}.$
    
\end{itemize}
Note that  $\mu(\delta)=\sum\limits_{\ell=0}^{h-1}\delta_{\ell}q^{\ell}=\delta_1q+\delta_0$ if $\delta_{h+1}q>\sum\limits_{\ell=0}^{h-1}\delta_{\ell}q^{\ell}$, and $\mu(\delta)=\delta_{h+1}q$ if $\delta_{h+1}q\leq \sum\limits_{\ell=0}^{h-1}\delta_{\ell}q^{\ell}$. Thus,
by  substituting these  terms into equation (\ref{ff}), we can use  an analogous argument as in {Case 1}  to obtain 
\begin{equation}\notag
    f(\delta)= \begin{cases}
      \frac{(\delta_{h+1}^2-\delta_{h+1}+\delta_1)(q-1)}{\lambda}+\lfloor \frac{\delta_1+\delta_0+\delta_{h+1}}{\lambda}\rfloor -\lfloor \frac{\delta_1+\delta_{h+1}}{\lambda}\rfloor& \hbox{if
      } \delta_h=0\hbox{ and }\delta_{h+1}q>\sum\limits_{\ell=0}^{h-1}\delta_{\ell}q^{\ell},\\
      \frac{\delta_{h+1}^2(q-1)}{\lambda} &\hbox{if
      }\delta_h=0\hbox{ and }\delta_{h+1}q\leq \sum\limits_{\ell=0}^{h-1}\delta_{\ell}q^{\ell}.
    \end{cases}
\end{equation} 

Finally, substituting the value of $f(\delta)$ into equation (\ref{th31}) for each corresponding case, we can conclude that  (\ref{cor3e0}) holds. This completes the proof.
\end{proof}

\section{The Bose distance  of narrow-sense BCH codes of length \texorpdfstring{$ (q^m-1)/\lambda $}\  }\label{bose}
In this section, we investigate the Bose distance of narrow-sense BCH codes of length $n=(q^m-1)/\lambda$ for designed distances  $2\leq \delta \leq \frac{q^{\lfloor (2m-1)/3\rfloor+1}-1}{\lambda}$.
Note that if $q \mid \delta$, then $\delta$ is not a coset leader modulo $n$. Therefore, we have $d_B(\mathcal{C}_{(q,n,\delta)})=d_B(\mathcal{C}_{(q,n,\delta+1)})$ for $q \mid \delta$.
Given this property, it is sufficient to focus on BCH codes $\mathcal{C}_{(q,n,\delta)}$ with $q \nmid \delta$. Their Bose distances are established in the following three theorems.

\begin{theorem}\label{th5}
    Let $m$ be a positive integer, and let $\lambda$ be a positive divisor of $q-1$. Set $n=(q^{m}-1)/\lambda$ and $h=\lfloor \frac{m}{2}\rfloor$. Let $\delta$ be an integer with   $2\leq \delta\leq \frac{q^{m-h}-1}{\lambda}$ and $q\nmid \delta$.    Then     $d_B(\mathcal{C}_{(q,n,\delta)})=\delta.$
\end{theorem}
\begin{proof}

    Note that $\lambda \mid q-1$ implies that $\mathrm{gcd}(q,\lambda)=1$. Since $q\nmid \delta$, it follows that $q\nmid \lambda\delta$. Additionally, it is clear that $2\leq \lambda\delta<q^{m-h}$.   By Lemma \ref{lll}, it follows that $\lambda \delta$ is  a coset leader modulo $\lambda n$. Then applying Lemma \ref{cos2}, we conclude that $\delta$ is  a coset leader modulo $n$. Consequently, $d_B(\mathcal{C}_{(q,n,\delta)})=\delta.$ 
\end{proof}

\begin{theorem}\label{bosth}
Let  $m\geq 5$ be an  odd integer, and let $\lambda$ be a positive divisor of $q-1$.  Set $n=(q^m-1)/\lambda$ and $h=\lfloor\frac{m}{2}\rfloor$. Let    $\delta$  be an integer with    $\frac{q^{h+1}-1}{\lambda}+1 \leq \delta \leq \frac{q^{\lfloor (2m-1)/3\rfloor+1}-1}{\lambda}$ and  $q\nmid \delta$. Define $j_{\delta} = \lfloor \log_q (\lambda \delta) \rfloor - h$, and  
write the $q$-adic expansion of $\lambda \delta$ as
$\lambda \delta=\sum\limits_{\ell=0}^{h + j_{\delta}} \delta_{\ell} q^{\ell}$. Let $r_{\delta}$ be   the smallest integer in $[-j_{\delta} + 1,\, j_{\delta}]$ such that $\delta_{h + r_{\delta}} > 0$
and let 
$
\hat{\delta} = \sum\limits_{\ell = h + r_{\delta}}^{h + j_{\delta}} \delta_{\ell} q^{\ell} + \sum\limits_{\ell =h - r_{\delta}+1}^{h + j_{\delta}} \delta_{\ell} q^{\ell - (h - r_{\delta}+1)}.
$
Then 
\begin{equation}\label{bose1}
d_B(\mathcal{C}_{(q,n,\delta)}) = \begin{cases}
\delta &\hbox{if }\sum\limits_{\ell=0}^{h-j_{\delta}} \delta_{\ell} q^{\ell} > \sum\limits_{\ell= h - r_{\delta}+1}^{h + j_{\delta}} \delta_{\ell} q^{\ell - (m - h - r_{\delta})}, \\
\lfloor \frac{\hat{\delta} }{\lambda}\rfloor+1 &\hbox{if } \sum\limits_{\ell=0}^{h-j_{\delta}} \delta_{\ell} q^{\ell} \leq \sum\limits_{\ell=h - r_{\delta}+1}^{h + j_{\delta}} \delta_{\ell} q^{\ell - (h - r_{\delta} + 1)} \hbox{ and  } \delta_{h - r_{\delta}+1}+\lambda -q  \neq \hat{\delta} \bmod \lambda, \\
\lfloor \frac{\hat{\delta} }{\lambda}\rfloor+ 2 &\hbox{if } \sum\limits_{\ell=0}^{h-j_{\delta}} \delta_{\ell} q^{\ell} \leq \sum\limits_{\ell=h - r_{\delta}+1}^{h + j_{\delta}} \delta_{\ell} q^{\ell - (h - r_{\delta} + 1)} \hbox{ and  } 
\delta_{h - r_{\delta}+1}+\lambda -q= \hat{\delta} \bmod \lambda.
\end{cases}
\end{equation}
\end{theorem}
\begin{proof}
We show that equation (\ref{bose1}) holds through the following cases. 

\textbf{Case 1.} Suppose that   $\sum\limits_{\ell=0}^{h-j_{\delta}} \delta_{\ell} q^{\ell} > \sum\limits_{\ell=h - r_{\delta} + 1}^{h + j_{\delta}} \delta_{\ell} q^{\ell - (h - r_{\delta} + 1)}$. By Remarks \ref{rm6} and  \ref{cor2}, the inequality implies that   $\lambda \delta\not\in \mathcal{A}_{j_{\delta}}(i)$ for every integer $i\in [-j_{\delta}+1,j_{\delta}]$. Therefore,  by Lemma  \ref{th2}, we conclude  $\lambda\delta\not \in \mathcal{S}$. Note that the assumptions   $q\nmid\delta$ and $\lambda \mid q-1$ imply $q\nmid \lambda \delta$. It follows that   $\lambda \delta$ is a coset leader modulo $\lambda n$. Applying Lemma \ref{cos2}, we further obtain that   $\delta$ is  a coset leader modulo $n$. Consequently, we have $ d_B(\mathcal{C}_{(q,n,\delta)})=\delta$.

\textbf{Case 2.}
Suppose that $\sum\limits_{\ell=0}^{h-j_{\delta}} \delta_{\ell} q^{\ell} \leq \sum\limits_{\ell=h - r_{\delta} + 1}^{h + j_{\delta}} \delta_{\ell} q^{\ell - (h - r_{\delta} + 1)}$. With an argument similar to that in  the proof of \cite[Theorem 6]{zheng2025}, we can  show that each integer in $ [\lambda\delta, \hat{\delta}]$ is not  a coset leader modulo $\lambda n$.  Applying Lemma \ref{cos2}, it follows  that each integer in $\left[\delta, \lfloor \frac{\hat{\delta}}{\lambda}\rfloor \right]$ is not a coset leader modulo $n$. Notice  that when $m$ is odd,  \begin{equation*}
    \hat{\delta}=\lambda \lfloor\frac{\hat{\delta}}{\lambda} \rfloor+\hat{\delta}\bmod \lambda= \sum\limits_{\ell = h + r_{\delta}}^{h + j_{\delta}} \delta_{\ell} q^{\ell} + \sum\limits_{\ell =h - r_{\delta}+1}^{h + j_{\delta}} \delta_{\ell} q^{\ell - ( h- r_{\delta}+1)}.\end{equation*} Thus, we have 
\begin{equation*}
    \lambda \lfloor\frac{\hat{\delta}}{\lambda} \rfloor+ \lambda = \sum\limits_{\ell = h + r_{\delta}}^{h + j_{\delta}} \delta_{\ell} q^{\ell} + \sum\limits_{\ell = h - r_{\delta}+2}^{h + j_{\delta}} \delta_{\ell} q^{\ell - (h - r_{\delta}+1)} +\delta_{h-r_{\delta}+1}-\hat{\delta}\bmod \lambda+\lambda.
\end{equation*}
This implies that \begin{equation*}
    q\mid  \lambda \lfloor\frac{\hat{\delta}}{\lambda} \rfloor+ \lambda\quad  \hbox{if and only if}\quad  
 \delta_{h-r_{\delta}+1}+\lambda-q=\hat{\delta} \bmod \lambda.\end{equation*}  
Then we distinguish the following two subcases: 

\textit{Subcase 1.} If $\delta_{h-r_{\delta}+1}+\lambda-q\neq \hat{\delta} \bmod \lambda$, then $q\nmid \lambda\lfloor \frac{\hat{\delta}}{\lambda}\rfloor+\lambda$. Applying a similar argument as in Case 1, we can show that $\lambda \lfloor \frac{\hat{\delta}}{\lambda}\rfloor+\lambda$ is a coset leader modulo $\lambda n=q^m-1$.  By applying Lemma \ref{cos2}, it follows that $\lfloor \frac{\hat{\delta}}{\lambda}\rfloor+1$ is a coset leader modulo $n$. Consequently, $d_B(\mathcal{C}_{(q,n,\delta)}) =\lfloor \frac{\hat{\delta}}{\lambda}\rfloor+1.$

\textit{Subcase 2.}
 If $\delta_{h-r_{\delta}+1}+\lambda -q= \hat{\delta} \bmod \lambda$, then $q\mid \lambda \lfloor \frac{\hat{\delta}}{\lambda}\rfloor+\lambda$ and $q\nmid \lambda \lfloor \frac{\hat{\delta}}{\lambda}\rfloor+2\lambda $. This implies that $\lambda \lfloor \frac{\hat{\delta}}{\lambda}\rfloor+\lambda $ is not a coset leader of $\lambda n$. Furthermore, using a similar argument as in Case 1 again, we can conclude that  while $\lambda\lfloor \frac{\hat{\delta}}{\lambda}\rfloor+2\lambda $ is a coset leader modulo $\lambda  n$. By applying Lemma \ref{cos2}, it follows that  $\lfloor \frac{\hat{\delta}}{\lambda}\rfloor+1$ is not a coset leader modulo $n$, while $\lfloor \frac{\hat{\delta}}{\lambda}\rfloor+2$ is  a coset leader modulo $n$. Therefore, we have  $d_B(\mathcal{C}_{(q,n,\delta)}) =\lfloor \frac{\hat{\delta}}{\lambda}\rfloor+2.$
 \end{proof}

\begin{theorem}\label{bosth2}
Let   $m\geq 4$ be an  even integer,  and  let $\lambda$ be a positive divisor of $q-1$.  
Set $n=(q^m-1)/\lambda$ and $h=\lfloor\frac{m}{2}\rfloor$.  
Let  $\delta$ be an integer such that  $ \frac{q^{h}-1}{\lambda}+1 \leq \delta \leq \frac{q^{\lfloor (2m-1)/3\rfloor+1}-1}{\lambda}$ and  $q\nmid \delta$.  Define $j_{\delta} = \lfloor \log_q (\lambda \delta) \rfloor - h$,  and write   the $q$-adic expansion of $\lambda \delta$ as $\lambda \delta=\sum\limits_{\ell=0}^{h + j_{\delta}} \delta_{\ell} q^{\ell}$. Let $r_{\delta}$ be   the smallest integer in $[-j_{\delta},  j_{\delta}]$ such that $\delta_{h + r_{\delta}} > 0$ and let 
$
\hat{\delta} = \sum\limits_{\ell = h + r_{\delta}}^{h + j_{\delta}} \delta_{\ell} q^{\ell} + \sum\limits_{\ell = h - r_{\delta}}^{h + j_{\delta}} \delta_{\ell} q^{\ell - (h- r_{\delta})}.
$
\begin{itemize}
\item 
 If $r_{\delta} \neq 0$, then
\begin{equation}\label{64}
d_B(\mathcal{C}_{(q,n, \delta)}) = \begin{cases}
\delta &\hbox{if }  \sum\limits_{\ell=0}^{h-j_{\delta}-1} \delta_{\ell} q^{\ell} > \sum\limits_{\ell=h - r_{\delta}}^{h + j_{\delta}} \delta_{\ell} q^{\ell - (h - r_{\delta})}, \\
\lfloor \frac{\hat{\delta}}{\lambda}\rfloor + 1 &\hbox{if }  \sum\limits_{\ell=0}^{h-j_{\delta}-1} \delta_{\ell} q^{\ell} \leq \sum\limits_{\ell=h - r_{\delta}}^{h + j_{\delta}} \delta_{\ell} q^{\ell - (h - r_{\delta})} \text{ and } \delta_{h - r_{\delta}}+\lambda -q  \neq \hat{\delta} \bmod \lambda, \\
\lfloor \frac{\hat{\delta}}{\lambda}\rfloor + 2 &\hbox{if }  \sum\limits_{\ell=0}^{h-j_{\delta}-1} \delta_{\ell} q^{\ell} \leq \sum\limits_{\ell=h - r_{\delta}}^{h + j_{\delta}} \delta_{\ell} q^{\ell - (h - r_{\delta})}\text{ and } \delta_{h - r_{\delta}}+\lambda -q  = \hat{\delta} \bmod \lambda.
\end{cases}
\end{equation}

 \item If $r_{\delta} = 0$, then
\begin{equation}\label{65}
d_B(\mathcal{C}_{(q,n,\delta)}) = \begin{cases}
\delta &\hbox{if } \sum\limits\limits_{\ell=0}^{h-j_{\delta}-1} \delta_{\ell} q^{\ell} \geq  \sum\limits_{\ell=h}^{h + j_{\delta}} \delta_{\ell} q^{\ell - h}, \\
 \frac{\hat{\delta}}{\lambda}  &\hbox{if } \sum\limits_{\ell=0}^{h-j_{\delta}-1} \delta_{\ell} q^{\ell} < \sum\limits_{\ell=h}^{h + j_{\delta}} \delta_{\ell} q^{\ell - h}\hbox{ and }\lambda\mid \hat{\delta},\\
\lfloor \frac{\hat{\delta}}{\lambda} \rfloor +1 &\hbox{if } \sum\limits_{\ell=0}^{h-j_{\delta}-1} \delta_{\ell} q^{\ell} < \sum\limits_{\ell=h}^{h + j_{\delta}} \delta_{\ell} q^{\ell - h}, \lambda\nmid \hat{\delta} \hbox{ and }  \delta_h+\lambda-q\neq \hat{\delta} \bmod \lambda,\\
\lfloor \frac{\hat{\delta}}{\lambda} \rfloor +2 &\hbox{if } \sum\limits_{\ell=0}^{h-j_{\delta}-1} \delta_{\ell} q^{\ell} < \sum\limits_{\ell=h}^{h + j_{\delta}} \delta_{\ell} q^{\ell - h}, \lambda\nmid \hat{\delta} \hbox{ and }\delta_h+\lambda-q=  \hat{\delta} \bmod \lambda. 
\end{cases}
\end{equation}
\end{itemize}
\end{theorem}
\begin{proof}
We can use an argument analogous to that in the proof of Theorem \ref{bosth} to conclude that \ref{64}) holds if $r_{\delta}\neq 0$.  Thus, we now only demonstrate equation (\ref{65}) by considering the following cases. 

 \textbf{Case 1.}
Suppose  that $\sum\limits_{\ell=0}^{h-j_{\delta}-1}\delta_\ell q^\ell
          \ge \sum\limits_{\ell=h}^{h+j_{\delta}}\delta_\ell q^{\ell-h}$.
If the inequality is strict, by Lemma~\ref{th2} and 
Remarks~\ref{rm6} -- \ref{cor2}, we have 
$\lambda\delta\notin\mathcal{S}\cup\mathcal{H}$; 
If equality holds, by Lemma~\ref{corr}, we 
$\lambda\delta\in\mathcal{H}$. Note that the assumptions   $q\nmid\delta$ and $\lambda \mid q-1$ imply $q\nmid \lambda \delta$. Therefore, 
in either case, $\lambda\delta$ is a coset leader modulo
$\lambda n$.
By Lemma~\ref{cos2}, it follows that $\delta$ is a coset leader of
modulo $n$. Therefore,
$d_B(\mathcal{C}_{(q,n,\delta)})=\delta$.

\textbf{Case 2.} 
Suppose that $\sum\limits_{\ell=0}^{h-j_{\delta}-1}\delta_\ell q^\ell<
    \sum\limits_{\ell=h}^{h+j_{\delta}}\delta_\ell  q^{\ell-h}$. We can  use an analogous argument as in the proof of \cite[Theorem 6]{zheng2025}  to show that each integer in $ [\lambda\delta, \hat{\delta})$ is not  a coset leader modulo $\lambda n$.   By Lemma \ref{cos2}, it follows that each integer $a\in \left[\delta, \lfloor \frac{\hat{\delta}}{\lambda}\rfloor \right)$ is not a coset leader modulo $n$. Then we distinguish the following two subcases:

\textit{Subcase 1.} Assume that $\lambda\mid \hat{\delta}$.  Using Lemma \ref{corr}, it is straightforward to verify that $\hat{\delta}\in \mathcal{H}$, and hence $\hat{\delta}$ is a coset leader modulo $\lambda n.$   
    It follows that $\frac{\hat{\delta}}{\lambda}$ is a coset leader modulo $n$. Therefore, we can obtain $d_B(\mathcal{C}_{(q,n,\delta)})=\frac{\hat{\delta}}{{\lambda}}$.  

\textit{Subcase 2.}
Assume that $\lambda \nmid \hat{\delta}$.  Then we have $\lambda\lfloor \frac{\hat{\delta}}{\lambda}\rfloor<\hat{\delta}$. Recalling that each integer in $[\lambda \delta, \hat{\delta} )$ is not a coset leader modulo $\lambda n$, this implies that $\lambda \lfloor \frac{\hat{\delta}}{\lambda}\rfloor$ is not a coset leader modulo $\lambda n$. By Lemma \ref{cos2}, it follows that $\lfloor \frac{\hat{\delta}}{\lambda}\rfloor$ is not a coset leader modulo $n$.  Then we can use an analogous argument as in the proof of Case 2 of Theorem \ref{bosth} to conclude that $d_B(\mathcal{C}_{(q,n,\delta)})=\lfloor \frac{\hat{\delta}}{\lambda}\rfloor + 1 
$ if $\delta_h+\lambda-q\neq \hat{\delta} \bmod \lambda$,
 and $d_B(\mathcal{C}_{(q,n,\delta)})=\lfloor \frac{\hat{\delta}}{\lambda}\rfloor + 2 
$ if  $\delta_h+\lambda-q=  \hat{\delta} \bmod \lambda$. 
\end{proof}

An $[n,k,d]$ linear code over $\mathbb{F}_q$ is called \emph{optimal} if there does not exist an $[n,k',d]$ linear code with $k' > k$, or an $[n,k,d']$ linear code with $d' > d$ over $\mathbb{F}_q$.
Using the formulas given in Theorems~\ref{odda} -- \ref{bosth2}, we can readily determine the dimension and Bose distance of narrow-sense BCH codes of length 
$(q^m-1)/\lambda$ for designed distances  $
2 \le \delta \le \frac{q^{\lfloor (2m-1)/3\rfloor+1}-1}{\lambda}$, where $\lambda$ is a positive divisor of $q-1$. 
  By comparing these codes with the \textit{Database}, we identify several BCH codes that are either optimal or have parameters matching those of the best-known linear codes.
But most of these codes arise in the primitive case $\lambda = 1$. See Tables~\uppercase\expandafter{\romannumeral 2}--\uppercase\expandafter{\romannumeral 3} in our previous work \cite{zheng2025} for examples of such primitive BCH codes. 

In addition, we list further examples of narrow-sense BCH codes of length $(q^m-1)/\lambda$ with $\lambda \neq 1$ in Table~\ref{table1}. The parameters of these codes are nearly optimal or best-known in the \textit{Database}. All parameters were verified by Magma, and for all BCH codes listed in Table~\ref{table1}, the Bose distance $d_B$ coincides with the minimum distance $d$. The Python programs implementing these formulas, together with additional examples, are available at \url{https://zheng-run.github.io}.

\begin{table}[htbp]
\centering
\caption{Examples of BCH code $\mathcal{C}_{(q,(q^m-1)/\lambda,\delta)}$ with $\lambda\neq 1$} 
\label{table1}
\centering
\begin{tabular}{|c|c|c|c|c|c|c|c|c|}
\hline
$q$&	$m$&	$\lambda$&	$n$&	$\delta$&	$k$&	$d_B$& 	Optimality\\ \hline
3&	4&	2&	40&	2&	36&	2&	$d_{\text{optimal}}=3$\\ \hline
3&	4&	2&	40&	3 -- 4&	32&	4	& $d_{\text{optimal}}=5$\\ \hline
3&	4&	2&	40&	5&	28&	5&	$d_{\text{best}}=6$\\ \hline
3&	4&	2&	40&	6 -- 7&	26&	7	&Best known \\ \hline
3&	4&	2&	40&	8&	22&	8&	$d_{\text{best}}=9$\\ \hline
3&	5&	2&	121&	6 -- 7&	101&	$7$&	$d_{\text{best}}=8$\\ \hline
3&	5&	2&	121&	9 -- 10&	91&	$10$	&$d_{\text{best}}=11$\\ \hline
4&	4&	3&	85&	4 -- 5 &	73 &	$5$	&$d_{\text{best}}=6$\\ \hline
\end{tabular} 
\end{table}

 \section{BCH codes with designed distance \texorpdfstring{$\delta  $ } \  satisfying \texorpdfstring{$\lambda \delta = aq^{h+k}+b$. } \ }\label{illustration} 

 As an illustration of   Theorems~\ref{odda} -- \ref{bosth2} established in the previous two sections, we  derive the following two corollaries in this section. These corollaries present the dimension and Bose distance of narrow-sense BCH code
 $\mathcal{C}_{(q,n,\delta)}$ for 
 $\delta$ satisfying  $\lambda \delta = a q^{h+k} + b$, where $k,a,b$ are integers such that $m - 2h \leq k \leq \lfloor (2m - 1)/3 \rfloor - h$, $1 \leq a \leq q - 1$, $\lambda \leq b \leq q^{m - h - k}$  and $q \nmid b$.
 
\begin{Corollary}\label{example1}
   Let $m\geq 5$ be an odd integer. Let $\delta$ be an integer such that   $\lambda\delta=a q^{h+k}+b$ for some  integers $k$, $a$, $b$ with 
$1 \le k \le \lfloor (2m-1)/3 \rfloor - h$, 
$1 \le a \le q-1$, 
$\lambda \le b \le q^{h-k+1}$ and $q \nmid b$.
 Then 
  \begin{equation}\label{cor51}
    d_B(\mathcal{C}_{(q,n,\delta)})=
    \begin{cases}
     \delta& \hbox{if }  b> a q^{2k-1},\\
 \lfloor \frac{aq^{h+k}+aq^{2k-1}}{\lambda}\rfloor+1 
 &\hbox{if }   b\leq a q^{2k-1},
    \end{cases}
    \end{equation}
  and 
 \begin{equation}\label{cor52}
    \mathrm{dim}(\mathcal{C}_{(q,n,\delta)})=
    \begin{cases}
      n-mN(\delta)+  \frac{q-1}{\lambda}ma^2q^{2k-3}\left[(q-1)k+1\right]& \hbox{if }  b\geq  a q^{2k-1}+\lambda,\\
 n-\frac{q-1}{\lambda}m a q^{h+k-1}+ \frac{q-1}{\lambda}ma  q^{2k-3}\left[ a(q-1)k+a-q 
 \right]
 &\hbox{if }   b< a q^{2k-1}+\lambda.
    \end{cases}
    \end{equation}   
\end{Corollary}
\begin{proof}
Note that $q\nmid b$ implies that $q\nmid \lambda \delta$. 
Since  $\lambda$ is a divisor of  $q-1$, it follows that $q\nmid \delta$.   Then 
 by applying Theorem \ref{bosth}, we can directly derive equation (\ref{cor51}).
 
  We now demonstrate that (\ref{cor52}) holds.
It is clear that 
$V(\lambda(\delta-1))$ has the form 
\begin{equation}\notag
    (\mathbf{0}_{h-k}, a, \mathbf{0}_{2k_{}-1},\delta_{h-k},\ldots,\delta_0)
\end{equation}
with $\sum\limits_{\ell=0}^{h-k}\delta_{\ell}q^{\ell}=b-\lambda$.  
In this case, we have 
\begin{itemize}
    \item $s_{\delta}=k_{\delta}=k$; 
    \item $w_{\delta}=aq^{h+k}$;
    \item $\mu(\delta)=\min\left\{b-\lambda, a q^{2k-1}\right\}$;
    \item $\mathcal{T}_i(\delta)=\left[q^{k-i},a q^{k-i}\right]\cap \{t\in \mathbb{Z}: q\nmid t\}$ for each integer $i\in [-k+1,k]$.
\end{itemize}
Then applying Lemma  \ref{le5}, we can   obtain  
\begin{equation}\notag
    \sum\limits_{t\in \mathcal{T}_i(\delta)} \left[\lfloor \frac{\lfloor tq^{2i}\rfloor+t}{\lambda}\rfloor -\lfloor \frac{\lfloor tq^{2i-1}\rfloor+t}{\lambda}\rfloor\right]  = 
    \begin{cases}
        \frac{1}{2\lambda}(a^2-1)(q-1)^2q^{2k-3}& \hbox{if }-k+2 \leq i\leq    k-1,\\
        \frac{1}{2\lambda} a(a-1)(q-1)q^{2k-2}& \hbox{if }i=k \hbox{ or }-k+1.
        \end{cases}       
\end{equation}
It follows that 
\begin{equation}\notag
\begin{split}
\sum\limits_{i=-k_{\delta}+1}^{k_{\delta}}\sum\limits_{t\in \mathcal{T}_i(\delta)}N(tq^{2i-1}+1)&= \sum\limits_{i=-k+1}^{k}\sum\limits_{t\in \mathcal{T}_i(\delta)}N(tq^{2i-1}+1) \\
&=\frac{1}{\lambda}\left[a (a-1)(q-1)q^{2k-2}+(a^2-1)(k-1)(q-1)^2q^{2k-3}\right].
    \end{split}
\end{equation}
By substituting the values of $k_{\delta}$ and  $\mu(\delta)$, and the above expression into (\ref{ff}), we obtain 
\begin{equation}\notag
     f(\delta)=
     \begin{cases}        
  \frac{q-1}{\lambda}a^2q^{2k-3}\left[k(q-1)+1\right]& \hbox{if }b\geq  a q^{2k-1}+\lambda,\\
\frac{q-1}{\lambda}  a  q^{2k-3}\left[ ak(q-1)+a-q 
 \right]+\frac{b+a-\lambda}{\lambda} - \lfloor\frac{\lfloor (b-\lambda+aq)/q\rfloor}{\lambda} \rfloor  &\hbox{if } b< a q^{2k-1}+\lambda.
     \end{cases}
 \end{equation}
 Furthermore, notice that \begin{equation}\notag\begin{split}   
     N(\delta)- \left[\frac{b+a-\lambda}{\lambda} - \lfloor\frac{\lfloor (b-\lambda+aq)/q\rfloor}{\lambda} \rfloor\right]& = \frac{aq^{h+k}+b-\lambda}{\lambda} -\lfloor \frac{aq^{h+k}+b-\lambda}{\lambda q}\rfloor -\frac{b+a-\lambda}{\lambda}+\lfloor\frac{\lfloor (b-\lambda+aq)/q\rfloor}{\lambda} \rfloor\\
     &= \frac{aq^{h+k-1}(q-1)}{\lambda} -\lfloor \frac{aq+b-\lambda}{\lambda q}\rfloor + \lfloor\frac{\lfloor (b-\lambda+aq)/q\rfloor}{\lambda} \rfloor\\
     &= \frac{aq^{h+k-1}(q-1)}{\lambda}. 
     \end{split}
 \end{equation}
 Therefore,  we can apply Theorem  \ref{odda} to derive the desired equality in (\ref{cor52}).   
\end{proof}

\begin{Remark}
In Corollary~\ref{example1}, we assume $q\nmid b$. Since $\lambda$ is a divisor of $q-1$, we have $\gcd(q,\lambda)=1$. From $\lambda \delta = a q^{h+k} + b$, it follows that $q\mid b$ if and only if $q\mid \lambda\delta$, which is further equivalent to $q\mid \delta$. In this case,
we have $
\mathcal{C}_{(q,n,\delta)} = \mathcal{C}_{(q,n,\delta+1)},
$
while clearly $q\nmid (\delta+1)$. Hence it suffices to consider the case $q\nmid b$. Similarly, in Corollary~\ref{example2} below we also impose the condition $q\nmid b$.
\end{Remark}

\begin{Corollary}\label{example2}
 Let $m\geq 4$ be an even integer. Let $\delta$ be an integer such that $\lambda\delta=a q^{h+k}+b$ for some integers $k, a, b $ with    $0 \leq k\leq \lfloor (2m-1)/3\rfloor-h$,   
  $1\leq a \leq q-1$,     $\lambda\leq b\leq q^{h-k} $ and  $q\nmid b$. 
 \begin{itemize}
\item If $k=0$, then 
\begin{equation}\label{db1}
    d_B(\mathcal{C}_{(q,n,\delta)}) =\begin{cases}
        \delta&\hbox{if } b\geq a,\\
        \frac{aq^h+a}{\lambda}&\hbox{if }b<a \hbox{ and }\lambda \mid 2a,\\
      \lfloor  \frac{aq^h+a}{\lambda}\rfloor+1&\hbox{if } b<a,\lambda\nmid 2a \hbox{ and }a+\lambda-q\neq 2a \bmod \lambda,\\
      \lfloor  \frac{aq^h+a}{\lambda}\rfloor+2&\hbox{if } b<a,\lambda\nmid 2a \hbox{ and }a+\lambda-q = 2a \bmod \lambda,
    \end{cases}
\end{equation}
and 
\begin{equation}\label{69}
   \mathrm{dim}(\mathcal{C}_{(q,n,\delta)})\!=\!
    \begin{cases}
      n\!-\!mN(\delta)\!+\!m\left[\frac{a+b-\lambda}{\lambda}\!-\!\lfloor\frac{2a}{\lambda}\rfloor+\sum\limits_{t=1}^a\left[\lfloor \frac{2t}{\lambda}\rfloor\!-\![\lfloor \frac{t}{\lambda}\rfloor\right]\right] \!-\!\frac{m(a-1)\left(3+(-1)^{\lambda}\right)}{4\lambda}&\hspace{-6pt} \hbox{if }  b<a+\lambda, \\
n\!-\!m N(\delta)+m\sum\limits_{t=1}^a\left[\lfloor\frac{2t}{\lambda}\rfloor\!-\!\lfloor\frac{t}{\lambda}\rfloor\right]-\frac{m(a-1)\left(3+(-1)^{\lambda}\right)}{4\lambda}
 &\hspace{-6pt}\hbox{if } b\geq a+\lambda \hbox{ and }  \lambda\nmid 2a,\\
  n\!-\!m N(\delta)\!+\!m\sum\limits_{t=1}^a\left[\lfloor\frac{2t}{\lambda}\rfloor\!-\!\lfloor\frac{t}{\lambda}\rfloor\right]-\frac{m(a-1)\left(3+(-1)^{\lambda}\right)}{4\lambda}\!-\!\frac{m}{2}
 &\hspace{-6pt}\hbox{if } b\geq a+\lambda \hbox{ and }  \lambda\mid 2a.     
    \end{cases}
\end{equation}
 
     \item If $k\geq 1$, then 
\begin{equation}\label{db2}
    d_B(\mathcal{C}_{(q,n,\delta)}) =\begin{cases}
        \delta&\hbox{if } b> aq^{2k},\\
        \lfloor\frac{aq^{h+k}+aq^{2k}}{\lambda}\rfloor+1&\hbox{if }b\leq aq^{2k} \hbox{ and }a+\lambda-q \neq 2a \bmod \lambda,\\
      \lfloor\frac{aq^{h+k}+aq^{2k}}{\lambda}\rfloor+2&\hbox{if }b\leq aq^{2k} \hbox{ and }a+\lambda-q = 2a \bmod \lambda,\\
    \end{cases}
\end{equation}
  and    
\begin{equation}\label{cor5e1}
 \mathrm{dim}(\mathcal{C}_{(q,n,\delta)})=
  \begin{cases}
        n-mN(\delta)+m\frac{q-1}{\lambda}a^2q^{2k-2}\left[\left(k-\frac{1}{2}\right)(q-1)+q\right] \\ 
        \quad \quad \quad \hbox{if }  b> aq^{2k} + \lambda \hbox{ and }\lambda \hbox{ is odd},\\
        n-mN(\delta)+m\frac{q-1}{\lambda}a^2q^{2k-2}\left[\left(k-\frac{1}{2}\right)(q-1)+q\right]-\frac{m(q-1)(q^{k-1}-1)}{2\lambda}\\
        \quad \quad \quad \hbox{if }  b> aq^{2k} + \lambda \hbox{ and }\lambda \hbox{ is even},\\
n+m\frac{q-1}{\lambda}aq^{2k-2}   \left[a(k-\frac{1}{2})(q-1)+(a-1)q\right] -\frac{maq^{h+k-1}(q-1)}{\lambda}
\\
 \quad \quad \quad \hbox{if } b\leq  aq^{2k}+\lambda
   \hbox{ and } \lambda \hbox{ is odd,} \\
\textstyle{n+m\frac{q-1}{\lambda}aq^{2k-2}   \left[a(k-\frac{1}{2})(q-1)+(a-1)q\right]}
\textstyle{-\frac{maq^{h+k-1}(q-1)}{\lambda}-\frac{m(q-1)(q^{k-1}-1)}{2\lambda}}
\\
 \quad \quad \quad\hbox{if } b\leq  aq^{2k}+\lambda
   \hbox{ and } \lambda \hbox{ is even.}   \end{cases}
\end{equation}

 \end{itemize}
\end{Corollary}

\begin{proof}
By Theorem \ref{bosth2}, it is easy to conclude that (\ref{db1}) holds if $k=0$, and (\ref{db2}) holds if $k\geq 1$. Therefore, we only demonstrate that (\ref{69}) holds if $k=0$, and   (\ref{cor5e1}) holds if $k\geq 1$.

    If $k=0$, then we have $\lambda(\delta-1) = a q^h+b-\lambda$ with  $0\leq b-\lambda<q^{h-1}$. In this case, we have  
    \begin{itemize}
   \item  $\frac{q^h-1}{\lambda}+1<\delta\leq \frac{q^{h+1}-1}{\lambda}+1$;
   \item $\delta_h=a$;
        \item $\widetilde{\mu}(
    \delta)=\min\{b-\lambda, a\}$;
    \item $w_{\delta}=aq^h$;
    \item $\tau(\delta)=1$ if $a \leq b-\lambda$ and $\lambda\mid 2a$, and $\tau(\delta)=0$ if otherwise. 
    \end{itemize}
By substituting   the above values  for the corresponding  case into equations (\ref{wf}) and (\ref{g}), we derive  \begin{equation}\notag
    \widetilde{f}(\delta)=\begin{cases}
    \frac{a+b-\lambda}{\lambda}-\lfloor\frac{2a}{\lambda}\rfloor+ \sum\limits_{t=1}^a\left[ \lfloor \frac{2t}{\lambda}\rfloor -  \lfloor \frac{t}{\lambda}\rfloor \right]& \hbox{if } b<a+\lambda,\\
   \sum\limits_{t=1}^a\left[ \lfloor \frac{2t}{\lambda}\rfloor -  \lfloor \frac{t}{\lambda}\rfloor \right]& \hbox{if }  b\geq  a+\lambda, 
    \end{cases}
\end{equation}
and 
\begin{equation}\notag
g(\delta)=\begin{cases}
    \frac{(a-1)(3+(-1)^{\lambda})}{2\lambda}+1& \hbox{if }b\geq a+\lambda \hbox{ and } \lambda\mid  2a,\\
  \frac{(a-1)(3+(-1)^{\lambda})}{2\lambda}& \hbox{otherwise}. 
    \end{cases}
\end{equation}
Then by Theorem \ref{odda}, it follows that (\ref{69})  holds. 

If $k\geq 1$, then we have  $k_{\delta}=k$ and $V(\lambda(\delta -1))$ has the form 
\begin{equation}\notag
    (\mathbf{0}_{h-k-1}, a, \mathbf{0}_{2k}, \delta_{h-k-1}, \ldots,\delta_0)
\end{equation}with $\sum\limits_{\ell=0}^{h-k}\delta_{\ell}q^{\ell}=b-\lambda$.
In this case, we have 
\begin{itemize}
\item $w_{\delta}=aq^h$;
    \item $s_{\delta}=k_{\delta}=k$;
    \item $\tau(\delta)=0$;
    \item $\widetilde{\mu}(\delta)=\min\left\{b-\lambda, a q^{2k}\right\}$;
    \item $\phi(\delta)=aq^k-1$;
   \item $\mathcal{T}_i(\delta)=\left[q^{k-i},a q^{k-i}-1\right]\cap \{t\in \mathbb{Z}: q\nmid t\}$ for each integer $i\in [-k,k]$.
\end{itemize}
Then applying Lemmas \ref{lema12} and \ref{le5}, we obtain 
\begin{equation}\notag
    \sum\limits_{t\in \mathcal{T}_i(\delta)} \left[ \lfloor\frac{\lfloor tq^{2i}\rfloor+t}{\lambda}\rfloor -  \lfloor\frac{\lfloor tq^{2i-1}\rfloor+t}{\lambda}\rfloor\right]= 
    \begin{cases}
         \frac{1}{2\lambda}(a^2-1)(q-1)^2q^{2k-2}  &\hbox{if } -k+1\leq i \leq  k-1 \hbox{ and } i\neq 0,\\
         \frac{1}{2\lambda}a(a-1)(q-1)q^{2k-1} &\hbox{if } i=k \hbox{ or } -k,\\
         {\left(\begin{aligned}
             &\textstyle{\frac{1}{2\lambda}(a^2-1)(q-1)^2q^{2k-2}}\\&+\textstyle{\frac{(3+(-1)^{\lambda})}{4\lambda}(a-1)(q-1)q^{k-1}}
         \end{aligned}\right)} & \hbox{if }i=0.
     \end{cases}      
\end{equation}
It follows that 
\begin{equation}\notag
\begin{split}
\sum\limits_{i=-k_{\delta}}^{k_{\delta}}\sum\limits_{t\in \mathcal{T}_i(\delta)}\left[ \lfloor\frac{\lfloor tq^{2i}\rfloor+t}{\lambda}\rfloor -  \lfloor\frac{\lfloor tq^{2i-1}\rfloor+t}{\lambda}\rfloor\right]=&
   \frac{1}{\lambda}(k-\frac{1}{2})(a^2-1)(q-1)^2q^{2k-2}+\frac{1}{\lambda}a(a-1)(q-1)q^{2k-1}\\
   & + \frac{(3+(-1)^{\lambda})(a-1)q^{k-1}(q-1)}{4\lambda}.
    \end{split}
\end{equation}
By substituting  $k_{\delta}$, $\tau(\delta)$ $w_{\delta}$, $\widetilde{\mu}(\delta)$, $\phi(\delta)$ and the above expression into (\ref{wf}) and (\ref{g}), we obtain   
\begin{equation}\notag
    \widetilde{f}(\delta)=\begin{cases}
        \frac{q-1}{\lambda}a^2q^{2k-2}\left[\left(k-\frac{1}{2}\right)(q-1)+q\right]+\frac{q-1}{2\lambda}aq^{k-1}& \hbox{if }  b> aq^{2k} + \lambda \hbox{ and }\lambda \hbox{ is odd},\\
        \frac{q-1}{\lambda}a^2q^{2k-2}\left[\left(k-\frac{1}{2}\right)(q-1)+q\right]+\frac{q-1}{2\lambda}\left[(2a-1)q^{k-1}+1\right]& \hbox{if }  b> aq^{2k} + \lambda \hbox{ and }\lambda \hbox{ is even},\\
{\left(\begin{aligned}
&\textstyle\frac{q-1}{\lambda}aq^{2k-2}   \left[a(k-\textstyle\frac{1}{2})(q-1)+(a-1)q\right] \\\
&+ \textstyle\frac{q-1}{2\lambda}aq^{k-1}+\frac{a+b-\lambda}{\lambda}- \lfloor\textstyle \frac{\lfloor(a+bq-\lambda)/q\rfloor}{\lambda} \rfloor\end{aligned}\right)}
&\hbox{if } b\leq  aq^{2k}+\lambda
   \hbox{ and } \lambda \hbox{ is odd,} \\
{\left(\begin{aligned}
&\textstyle\frac{q-1}{\lambda}aq^{2k-2}   \left[a(k-\textstyle\frac{1}{2})(q-1)+(a-1)q\right]+\textstyle\frac{a+b-\lambda}{\lambda}\\
&+\textstyle\frac{q-1}{2\lambda}\left[(2a-1)q^{k-1}+1\right]- \lfloor \textstyle\frac{\lfloor(a+bq-\lambda)/q\rfloor}{\lambda} \rfloor\end{aligned}\right)}
&\hbox{if } b\leq  aq^{2k}+\lambda
   \hbox{ and } \lambda \hbox{ is even,}   \end{cases}
\end{equation}
and 
\begin{equation}\notag
    g(\delta)=\begin{cases} 
    \frac{a(q-1)q^{k-1}}{\lambda},& \hbox{if $\lambda$ is odd,}\\
    \frac{2a(q-1)q^{k-1}}{\lambda},& \hbox{if $\lambda$ is even.}
    \end{cases}
\end{equation}
By Theorem \ref{odda}, it follows that (\ref{cor5e1}) holds. 
\end{proof}

\section{Some non-narrow-sense BCH codes of length  \texorpdfstring{$ (q^m-1)/\lambda $} \ }\label{non}
\linenumbers
In the previous sections, we study the parameters of narrow-sense BCH codes of length $n=(q^m-1)/\lambda$. For the dimension of narrow-sense 
BCH codes $\mathcal{C}_{(q,n,\delta)}$, the
equality in \eqref{dimeq} allows us to compute $\dim(\mathcal{C}_{(q,n,\delta)})$ by 
determining the coset leaders modulo $n$ within the range $[1,\delta-1]$ and the sizes of their corresponding $q$-cyclotomic cosets.
However, this approach cannot be directly applied to non-narrow-sense BCH codes $\mathcal{C}_{(q,n,\delta,b)}$ with $b \ge 2$, since we do not have an analogue of \eqref{dimeq} for these codes. In general, for $b\geq 2$, we do not have 
\begin{equation*}
   \dim(\mathcal{C}_{(q,n,\delta,b)}) 
   =  n - \sum_{a \in \mathcal{L}^n(b,b+\delta-2)} |C_n(a)|,
\end{equation*}
because the equality
$
   \left|     \bigcup\limits_{a=b}^{b+\delta-2} C_n(a)
   \right|
   =
   \sum\limits_{a \in \mathcal{L}^n(b,b+\delta-2)} |C_n(a)|
$
may not hold  when $b \ge 2$. 
Consequently, determining the dimension of non-narrow-sense BCH codes $\mathcal{C}_{(q,n,\delta,b)}$ with $b \ge 2$ is in general more difficult, and there are very few results on the parameters of such codes in the literature.

Nevertheless, although the above equality does not hold in general, the results on the coset leaders  and sizes of  $q$-cyclotomic cosets given  in Sections~\ref{pre} and~\ref{review} enable us to determine the dimension and Bose distance of certain classes of non-narrow-sense BCH codes, as shown in the following result.

\begin{theorem}
Let $m \ge 4$ be an integer, and let $\lambda$ be a positive divisor of $q-1$. Set $n = (q^m-1)/\lambda$ and  $h=\lfloor \frac{m}{2}\rfloor$. 
Let $b$ be an integer with 
$
   \frac{q^{m-h}-1}{\lambda}+1 \le b < \frac{q^{\lfloor (2m-1)/3 \rfloor +1}-1}{\lambda}.
$    
Define
$
  j_b = \lfloor \log_q (\lambda b) \rfloor - h,
$
and write the $q$-adic expansion of $\lambda b$ as
$
  \lambda b = \sum\limits_{\ell=0}^{h+j_b} b_\ell q^\ell.
$
Let $r_b$ be the smallest integer in $[\,m-2h-j_b,\; j_b\,]$ such that $b_{h+r_b} > 0$. If
$
  \sum\limits_{\ell=0}^{h-j_b-1} b_\ell q^\ell
   \;>\;
  \sum\limits_{\ell=h-r_b}^{h+j_b} b_\ell q^{\ell-(h-r_b)},
$
then for any integer $\delta$ with
$
   2 \le \delta \le 
   \frac{q^{h-j_b}-1}{\lambda}
   - \left\lfloor 
\sum\limits_{\ell=0}^{h-j_b-1} b_\ell q^\ell/ \lambda
     \right\rfloor
   + 1,
$
we have
\begin{equation}\label{baobao}
   \dim(\mathcal{C}_{(q,n,\delta,b)}) = n - m(\delta-1)
\end{equation}
and
\begin{equation}\label{baobao2}
   d_B(\mathcal{C}_{(q,n,\delta,b)}) = \delta.
\end{equation}
\end{theorem}

\begin{proof}
We demonstrate the proof only for the case where 
$m$ is even, as the case   $m $ is odd follows similarly. First,
observe that  
\begin{equation*}
\begin{split}    V\left(\sum\limits_{\ell=h+r_b}^{h+j_b}b_{\ell}q^{\ell}+q^{h-j_{b}}-1\right) 
    = (\mathbf{0}_{h-1-j_b}, b_{h+j_b},\ldots,b_{h+r_b}, \mathbf{0}_{2r_b}, q-1,\ldots,q-1),
    \end{split}
\end{equation*} and 
 \begin{equation*}
     V(\lambda b)\leq V(\lambda a)\leq V\left(\sum\limits_{\ell=h+r_b}^{h+j_b}b_{\ell}q^{\ell}+q^{h-j_{b}}-1\right)
     \end{equation*}
     for all integers  $a\in \left[b,b+ \frac{q^{h-j_b}-1}{\lambda}-\lceil \frac{\sum\limits_{\ell=0}^{h-j_{b}-1}b_{\ell}q^{\ell}}{\lambda}\rceil  \right]$.  
 Since $\sum\limits_{\ell=0}^{h-j_b-1}b_{\ell}q^{\ell}>\sum\limits_{\ell=h-r_{b}}^{h+j_{b}}b_{\ell}q^{\ell-(h-r_{b})}$, it follows that  $V(\lambda a)$ must have the form 
\begin{equation}\label{eee}  
V(\lambda a)=(\mathbf{0}_{h-1-j_b}, b_{h+j_b}, \ldots, b_{h+r_b}, \mathbf{0}_{2r_{b}},a_{h-j_b-1},\ldots,a_0)\end{equation}
with $\sum\limits_{\ell=0}^{h-j_b-1}a_{\ell}q^{\ell}>\sum\limits_{\ell=h-r_{b}}^{h+j_{b}}b_{\ell}q^{\ell-(h-r_{b})}$ for all integers  $a\in \left[b,b+ \frac{q^{h-j_b}-1}{\lambda}-\lfloor \frac{\sum\limits_{\ell=0}^{h-j_{b}-1}b_{\ell}q^{\ell}}{\lambda}\rfloor   \right]$.  

Next, we 
define a  function $f$ for the  integers in $ \left[b,b+ \frac{q^{h-j_b}-1}{\lambda}-\lfloor \frac{\sum\limits_{\ell=0}^{h-j_{b}-1}b_{\ell}q^{\ell}}{\lambda}\rfloor   \right]$  as 
\begin{equation*}
    f(a)= \frac{a}{q^{t_a}},
\end{equation*}
where $t_a$ is largest integer such that $q^{t_a}\mid a$. Recall that  $\mathcal{L}^n_{\ell}(0,n-1) $ is simply denoted by $\mathcal{L}^n_{\ell}$. 
 We now show that \begin{equation}\label{HAHA}
f(a)\in \mathcal{L}^n_m \quad \hbox{for all } a\in  \left[b,b+ \frac{q^{h-j_b}-1}{\lambda}-\lfloor \frac{\sum\limits_{\ell=0}^{h-j_{b}-1}b_{\ell}q^{\ell}}{\lambda}\rfloor   \right]
\end{equation} by considering the following cases.

\textbf{Case 1.} Suppose that $t_a=0$, i.e., $q\nmid a$. Then  $ f(a)= a$. It follows that   $V(\lambda f(a))$ has the form given in  (\ref{eee}). 
By Remarks \ref{rm6} and \ref{cor2},  this implies that 
$\lambda f(a)\not \in \mathcal{B}_{j_b}(i)$ for any $i\in [-j_b,j_b]$. Applying Lemma  \ref{th2}, we  conclude that  
$\lambda f(a) \not \in \mathcal{H}\cup \mathcal{S}$.  
Therefore,   $\lambda f(a) \in \mathcal{L}_m^{\lambda n}$.  By Lemma \ref{cos2}, this implies that $f(a)\in \mathcal{L}^n_m$. 

\textbf{Case 2.} Suppose that $1\leq t_a\leq  j_b$. In this case, we first have 
$f(a)\leq b-1$. 
Since $V(\lambda a )$ has the form given in (\ref{eee}),  we derive    
\begin{equation}\notag
    V\left( \lambda f(a) \right) = (\mathbf{0}_{h-1-j_b+t_a}, b_{h+j_b}, \ldots, b_{h+r_b}, \mathbf{0}_{2r_{b}},a_{h-j_b-1},\ldots,a_{t_a}). 
\end{equation}
Furthermore, the inequality $\sum\limits_{\ell=0}^{h-j_b-1}a_{\ell}q^{\ell}>\sum\limits_{\ell=h-r_{b}}^{h+j_{b}}b_{\ell}q^{\ell-(h-r_{b})}$ implies that $a_{j_b+r_b}>0$.  It follows that 
  $\lambda f(a) \not \in \mathcal{B}_{j_b-t_a}(i)$ for every integer $i\in [t_a-j_b, j_b-t_a]$, and hence  $\lambda f(a) \not \in \mathcal{H}\cup \mathcal{S}$. This implies that  $\lambda f(a) \in \mathcal{L}^{\lambda n}_m$, and hence    $f(a)\in \mathcal{L}^n_m$. 

\textbf{Case 3.}
Suppose   $t_a\geq j_b+1$. In this case, we have  $\lambda f(a)<q^h$ and 
$\lambda f(a) q^h \bmod \lambda n \neq \lambda f(a)$. By Lemma \ref{th1}, we have $|C_{\lambda n}(\lambda f(a))|=m$. Moreover, 
by applying Lemma \ref{lll}, we  conclude that $\lambda f(a)$ is a coset leader modulo $\lambda  n$. It follows  that $f(a)\in \mathcal{L}^n_m$.

By now we have demonstrated that (\ref{HAHA}) holds. Note that $C_n(a)=C_n(f(a))$ and   $f$ is injective.  Therefore, we have 
\begin{equation*}
\begin{split}
   \left|\bigcup\limits_{a=b}^{b+\delta-2}C_n(a) \right|=  \left|\bigcup\limits_{a=b}^{b+\delta-2}C_n(f(a))\right|=\sum\limits_{a=b}^{b+\delta-2} |C_n(f(a))|=m(\delta-1) 
   \end{split}
\end{equation*} 
and 
\begin{equation*}
\begin{split}
   \bigcup\limits_{a=b}^{b+\delta-2}C_n(a)  =  \bigcup\limits_{a=b}^{b+\delta-2}C_n(f(a)) \neq \bigcup\limits_{a=b}^{b+\delta-1}C_n(f(a)) =  \bigcup\limits_{a=b}^{b+\delta-1}C_n(a)
   \end{split}
\end{equation*}
for any integer $\delta$ such that  $2\leq \delta \leq \frac{q^{h-j_b}-1}{\lambda}-\lfloor \frac{\sum\limits_{\ell=0}^{h-j_{b}-1}b_{\ell}q^{\ell}}{\lambda}\rfloor +1$.
Recalling equation  (\ref{ndimeq}) and the fact that the Bose distance $d_B$ of $\mathcal{C}_{(q,n,\delta,b)}$ is the largest integer for which (\ref{boseeq}) holds, we deduce that (\ref{baobao}) and (\ref{baobao2}) both hold. 
\end{proof}

Applying the above theorem, we find the following   BCH codes. 
\begin{example}
Let $q=3$, $m=4$ and $\lambda=1$. We have $n=80$.  The BCH code 
$\mathcal{C}_{(3,80,2,b)}$ has parameters $[80,76,2]$ for all integers $ b\in [11,17]\cup [21,25]$. 
\end{example}
\begin{example}
Let $q=4$, $m=4$ and $\lambda=1$. We have $n=255$.  The BCH code 
$\mathcal{C}_{(4,255,2,b)}$ has parameters $[255,251,2]$ for all integers $b\in [18, 30]\cup [35,46]\cup[52,62]$. 
\end{example}
By the Hamming bound,  no linear code over $\mathbb{F}_3$
with parameters $[80,76,d]$ exists for $d > 32$, nor does a linear code over
$\mathbb{F}_4$ with parameters $[255,251,d]$ exist for $d > 2$.
Therefore,  the BCH   codes in the above examples are optimal. Readers can also find the bounds for the minimum distance of  linear codes of length up to 256 in the \textit{Database}.

\section{Conclusion}\label{conclusion}
In this paper, we investigate the dimension and Bose distance of  BCH codes of length $(q^m-1)/\lambda$, where $\lambda$ is a positive divisor of $q-1$. Our main contribution is to provide explicit formulas for the dimension and   Bose distance of narrow-sense BCH codes of length $(q^m-1)/\lambda$
  for a much larger range of designed distances than previously known.  In addition, we 
 extend these results to some non-narrow-sense BCH codes of the same length. Applying our results, we find some BCH codes with good parameters.

\appendices
\section{Proof of Lemma \ref{le1}}
\begin{proof}
It is clear that \begin{equation}\label{l70}
    \{a\in [1,x]: q\nmid a \hbox{ and }\lambda\mid a+y \}= \{a\in [1,x]: \lambda\mid a+y \}-\{a\in [1,x]: q\mid a \hbox{ and }\lambda\mid a+y \}.
\end{equation}
Next, we show that 
\begin{equation}\label{l71}
 \left|   \{a\in [1,x]: q\mid a \hbox{ and }\lambda\mid a+y\}\right|=\left |\left\{a\in \left[1,\lfloor {x}/{q}\rfloor\right]: \lambda\mid a+y\right\}\right|.
\end{equation}
Notice that there exists a bijective $a\mapsto{a}/{q}$ between the integers in $[1,x]$ that are divisible by $q$ and the integers in $\left[1,\lfloor x/q \rfloor\right]$. Furthermore, for any integer $a$ such that $q\mid a$, we have 
\begin{equation}\notag
    a+y=\frac{a}{q}(q-1)+\frac{a}{q}+y.
\end{equation}
Since $\lambda\mid q-1$, it follows that $\lambda\mid a+y$ if and only if $\lambda\mid \frac{a}{q}+y$ for any integer $a$ such that $q\mid a$. Therefore, we can conclude that there exists a one-to-one correspondence between the two sets in (\ref{l71}) by mapping $a$ to ${a}/{q}$, and hence the equality in (\ref{l71}) holds. 
With (\ref{l70}), it follows that 
\begin{equation}\notag
\begin{split}
     \left| \{a\in [1,x]: q\nmid a \hbox{ and }\lambda\mid a+y \}\right|&= \left| \{a\in \left[\lfloor x/q\rfloor+1,x\right]:  \lambda\mid a+y \}\right| \\
     &=\lfloor \frac{x+y}{\lambda}\rfloor -\lfloor\frac{\lfloor {x}/{q}\rfloor+y}{\lambda} \rfloor.
\end{split}
\end{equation}
This completes the proof.
 \end{proof}

\section{Proof of Lemma \ref{lll8}}
\begin{proof}
First, we can apply a similar argument as utilized in the proof of Lemma \ref{le1} to conclude that 
    \begin{equation}\notag
   \left|  \{\alpha\in [x,y]:\lambda\mid 2\alpha \hbox{ and }q\mid \alpha \}\right|=\left| \left\{\alpha\in \big[\lceil {x}/{q}\rceil,\lfloor {y}/{q}\rfloor\big]: \lambda\mid 2\alpha\right\}\right|.
 \end{equation}
It follows that 
\begin{equation}\label{35}
    \left|  \{\alpha\in [x,y]:\lambda\mid 2\alpha \hbox{ and }q\nmid \alpha \}\right| =\left|  \{\alpha\in [x,y]:\lambda\mid 2\alpha \}\right|-\left| \left\{\alpha\in \big[\lceil {x}/{q}\rceil,\lfloor {y}/{q}\rfloor\big]: \lambda\mid 2\alpha\right\}\right|.
\end{equation}
Then we distinguish the   following two cases:

\textbf{Case 1.}
 Suppose that $\lambda$ is odd. Then $\lambda\mid 2\alpha$ holds if and only if $\lambda \mid \alpha.$ Therefore, 
\begin{equation}\notag
    \left|  \{\alpha\in [x,y]: \lambda\mid 2\alpha\}      \right| = \left| \{\alpha\in [x,y]: \lambda\mid \alpha\}   \right|=\lfloor\frac{y}{\lambda}\rfloor -\lfloor\frac{x-1}{\lambda} \rfloor.
\end{equation}
 Similarly, we also have 
\begin{equation}\notag
    \left|  \{\alpha\in [\lceil {x}/{q} \rceil, \lfloor {y}/{q} \rfloor]: \lambda\mid 2\alpha\}      \right| = \lfloor \frac{\lfloor y/q\rfloor}{\lambda}\rfloor-\lfloor \frac{\lceil {x}/{q} \rceil-1}{\lambda} \rfloor.
\end{equation}

\textbf{Case 2.} Suppose that $ \lambda$ is even. Then  $\lambda\mid 2\alpha$ holds if and only if $\frac{\lambda}{2}\mid \alpha$. Consequently, 
 \begin{equation}\notag
    \left|  \{\alpha\in [x,y]: \lambda\mid 2\alpha\}      \right| = \left| \left\{\alpha\in [x,y]: \frac{\lambda}{2} \mid \alpha\right\}   \right|=\lfloor\frac{2y}{\lambda}\rfloor -\lfloor\frac{2x-2}{\lambda} \rfloor.
\end{equation}
 Similarly, we also have 
\begin{equation}\notag
    \left|  \{\alpha \in \big[\lceil {x}/{q} \rceil, \lfloor 
{y}/{q} \rfloor\big]: \lambda\mid 2\alpha \}      \right| = \lfloor \frac{2\lfloor  y/q\rfloor}{\lambda}\rfloor-\lfloor \frac{2\lceil {x}/{q} \rceil-2}{\lambda} \rfloor.
\end{equation}

 We now can derive the desired equation 
from equation (\ref{35}) and the discussion for the above two cases.  This completes the proof. 
\end{proof}

\section{Proof of Lemma \ref{le0}}
\begin{proof}
For simplicity, we set  $\beta_1= x - \lambda\cdot \lfloor \frac{x}{\lambda} \rfloor $ and $\beta_2= y - \lambda\cdot \lfloor \frac{y}{\lambda} \rfloor $.  Then   
 \begin{equation}\label{9}
  \sum\limits_{t=1}^{q-1} \left [\lfloor \frac{t+x}{\lambda}\rfloor- \lfloor \frac{t+y}{\lambda}\rfloor \right]= \sum\limits_{t=1}^{q-1} \left[\lfloor \frac{x}{\lambda} \rfloor -\lfloor \frac{y}{\lambda}\rfloor+ 
 \lfloor  \frac{t+\beta_1}{\lambda}\rfloor-\lfloor \frac{t+\beta_2}{\lambda}\rfloor\right].
\end{equation}  
  It is straightforward to verify that  
\begin{equation}\notag
\lfloor\frac{t+\beta_1}{\lambda}\rfloor =\begin{cases}
    0, &\hbox{for } t\in [1, \lambda-\beta_1-1],    \\ 
    i,&\hbox{for }  t\in [i\lambda-\beta_1,(i+1)\lambda-1-\beta_1],\\
    \frac{q-1}{\lambda},&\hbox{for } t\in [q-1-\beta_1,q-1],
\end{cases}    
\end{equation}
where $i$ can be any integer in  $[1,\frac{q-1}{\lambda}-1].$
Therefore, we can obtain  
\begin{equation}\notag
    \sum\limits_{t=1}^{q-1}\lfloor \frac{t+\beta_1}{\lambda}\rfloor=\frac{(q-1)(\beta_1+1)}{\lambda}+\sum\limits_{i=1}^{\frac{q-1}{\lambda}-1}i\lambda.
\end{equation}
Similarly, we can also derive
   \begin{equation}\notag
    \sum\limits_{t=1}^{q-1}\lfloor \frac{t+\beta_2}{\lambda}\rfloor=\frac{(q-1)(\beta_2+1)}{\lambda}+\sum\limits_{i=1}^{\frac{q-1}{\lambda}-1}i\lambda.
\end{equation}
Combining the above two equalities, we obtain 
\begin{equation}\notag
    \sum\limits_{t=1}^{q-1} \left[\lfloor \frac{t+\beta_1}{\lambda}\rfloor-\lfloor \frac{t+\beta_2}{\lambda}\rfloor\right]=\frac{(q-1)(\beta_1-\beta_2)}{\lambda}.
\end{equation}
With the equality in (\ref{9}), it follows that 
\begin{equation}\notag
      \sum\limits_{t=1}^{q-1}\left[\lfloor \frac{t+x}{\lambda}\rfloor- \lfloor \frac{t+y}{\lambda}\rfloor \right] = \frac{q-1}{\lambda}\left(\lambda \cdot  \lfloor \frac{x}{\lambda} \rfloor +\beta_1\right)- \frac{q-1}{\lambda}\left(\lambda \cdot  \lfloor \frac{y}{\lambda} \rfloor +\beta_2\right)=\frac{(q-1)(x-y)}{\lambda}.
    \end{equation}This completes the proof.
\end{proof}

\section{Proof of Lemma \ref{adl}}
\begin{proof}
Notice that each integer $t\in [q,aq]$ such that $q\nmid t$ admits a unique decomposition \begin{equation*}
t=iq+j \quad \hbox{with } i\in [1,a-1] \hbox{ and } j\in [1,q-1].\end{equation*}  Substituting this decomposition, we have 
\begin{equation}\notag
\begin{split}
    \lfloor \frac{\lfloor tq^{-1}\rfloor+t}{\lambda}\rfloor- \lfloor \frac{t}{\lambda}\rfloor &= \lfloor \frac{i+iq+j}{\lambda}\rfloor-\lfloor \frac{iq+j}{\lambda}\rfloor\\
    &= \lfloor \frac{j+2i}{\lambda}\rfloor-\lfloor \frac{j+i}{\lambda}\rfloor. 
    \end{split}
    \end{equation}
Then by applying Lemma \ref{le0}, we obtain 
   \begin{equation}\notag
   \begin{split}
       \sum\limits_{t=q, q\nmid t }^{aq}\left [\lfloor \frac{\lfloor tq^{-1}\rfloor+t}{\lambda}\rfloor- \lfloor \frac{t}{\lambda}\rfloor \right]
       &=\sum\limits_{i=1}^{a-1}\sum\limits_{j=1}^{q-1}\left[ \lfloor \frac{j+2i}{\lambda}\rfloor-\lfloor \frac{j+i}{\lambda}\rfloor  \right]\\
       &=\sum\limits_{i=1}^{a-1}\frac{(q-1)i}{\lambda}\\
       &=\frac{a(a-1)(q-1)}{2\lambda}. 
       \end{split}
   \end{equation} 
   This completes the proof.
\end{proof}

\section{Proof of Lemma \ref{ll9}}
\begin{proof}
Notice that each integer $t\in [0,q-2]$ admits a  unique decomposition 
    \begin{equation*}t=i\lambda+j  \quad \hbox{with }  i\in \left[0,\frac{q-1}{\lambda}-1\right] 
\hbox{ and }j\in [0,\lambda-1].\end{equation*} Therefore, 
\begin{equation}\label{le8e1}
\begin{split}
    \sum\limits_{t=1}^{q-1}\left[\lfloor  \frac{2t+x}{\lambda}\rfloor- \lfloor \frac{t+x}{\lambda}\rfloor\right]&=\sum\limits_{t=0}^{q-1}\left[\lfloor  \frac{2t+x}{\lambda}\rfloor- \lfloor \frac{t+x}{\lambda}\rfloor\right]\\  &=\sum\limits_{i=0}^{\frac{q-1}{\lambda}-1}\sum\limits_{j=0}^{\lambda-1}\left[ \lfloor \frac{2i\lambda+2j+x}{\lambda}\rfloor-\lfloor \frac{i\lambda+j+x}{\lambda}\rfloor    \right]+\lfloor \frac{2(q-1)+x}{\lambda}\rfloor-\lfloor \frac{q-1+x}{\lambda}\rfloor\\
    &=\sum\limits_{i=0}^{\frac{q-1}{\lambda}-1}\sum\limits_{j=0}^{\lambda-1}\left[ i+ \lfloor \frac{2j+x}{\lambda} \rfloor-\lfloor \frac{j+x}{\lambda}\rfloor \right]+\frac{q-1}{\lambda}\\
    &=\frac{(q-1)(q+1)}{2\lambda}-\frac{q-1}{2}+\frac{q-1}{\lambda}\sum\limits_{j=0}^{\lambda-1}\left[ \lfloor \frac{2j+x}{\lambda} \rfloor-\lfloor \frac{j+x}{\lambda}\rfloor \right].
    \end{split}
\end{equation}
Let $y=x-\lambda\cdot\lfloor \frac{x}{\lambda} \rfloor$. Then we have $0\leq y\leq \lambda-1$ and 
\begin{equation}\label{le8e2}
   \lfloor \frac{2j+x}{\lambda} \rfloor-\lfloor \frac{j+x}{\lambda}\rfloor = \lfloor \frac{2j+y}{\lambda} \rfloor-\lfloor \frac{j+y}{\lambda}\rfloor. 
\end{equation}
Noting that $\lfloor  \frac{j+y}{\lambda} \rfloor =0$ for each integer $j\in [0,\lambda-y-1]$, and $\lfloor  \frac{j+y}{\lambda} \rfloor =1$ for each integer $j\in [\lambda-y, \lambda-1]$, we can derive 
\begin{equation}\label{le8e3}
\sum\limits_{j=0}^{\lambda-1}\lfloor \frac{j+y}{\lambda}\rfloor = y.\end{equation} 
Next, we determine the value of $\sum\limits_{j=0}^{\lambda-1}\lfloor  \frac{2j+y}{\lambda} \rfloor$ through the following two cases:

\textbf{Case 1.} Suppose that $\lambda$ is even. Since $x$ is even, it follows that  $y$ is also even. We can obtain 
\begin{equation}\notag
    \lfloor \frac{2j+y}{\lambda} \rfloor =\begin{cases}
        0&\hbox{if }0\leq j\leq \frac{\lambda-y}{2}-1, \\
        1&\hbox{if }\frac{\lambda-y}{2}\leq j \leq \frac{2\lambda-y}{2}-1, \\
        2 &\hbox{if }  \frac{2\lambda-y}{2} \leq j \leq \lambda -1. 
   \end{cases}   
\end{equation}
It follows that 
\begin{equation}\notag
\begin{split}
\sum\limits_{j=0}^{\lambda-1}\lfloor \frac{2j+y}{\lambda} \rfloor &=  \left(\frac{2\lambda-y}{2}- \frac{\lambda-y}{2}\right)+2\left(\lambda - \frac{2\lambda-y}{2}\right)  \\ &=\frac{\lambda}{2}+y.
\end{split}
\end{equation}

\textbf{Case 2.} Suppose that $\lambda$ is odd. Then we can obtain 
\begin{equation}\notag
    \lfloor \frac{2j+y}{\lambda} \rfloor =\begin{cases}
        0&\hbox{if }0\leq j\leq \frac{\lambda-y}{2}-1, \\
        1&\hbox{if }\frac{\lambda+1}{2}-\lfloor \frac{y+1}{2}\rfloor\leq j \leq \lambda-\lceil \frac{y+1}{2}\rceil, \\
        2 &\hbox{if } \lambda -\lceil \frac{y+1}{2} \rceil +1 \leq j \leq \lambda -1. 
   \end{cases}   
\end{equation}
Consequently,  we have 
\begin{equation}\notag
\begin{split}
\sum\limits_{j=0}^{\lambda-1}\lfloor \frac{2j+y}{\lambda} \rfloor  &= \left( \lambda-\lceil \frac{y+1}{2}\rceil -\frac{\lambda+1}{2}+ \lfloor \frac{y+1}{2}\rfloor +1 \right) + 2 \left(\lceil \frac{y+1}{2}\rceil  - 1 \right) \\
&= \frac{\lambda-1}{2}-1+\lceil \frac{y+1}{2}\rceil+\lfloor \frac{y+1}{2}\rfloor +1\\
&= \frac{\lambda-1}{2}+y.
\end{split}
\end{equation}
From equations (\ref{le8e2}) and (\ref{le8e3}), and the discussion for the above two cases, we can conclude that 
\begin{equation}\notag    \sum\limits_{j=0}^{\lambda-1}\left[\lfloor \frac{2j+x}{\lambda} \rfloor-\lfloor \frac{j+x}{\lambda}\rfloor \right] =\begin{cases}
        \frac{\lambda}{2},& \hbox{if  }
\lambda \hbox{ is even,}   \\
\frac{\lambda-1}{2},& \hbox{if  }
\lambda \hbox{ is odd.}
\end{cases}
\end{equation}
Then by applying (\ref{le8e1}), we obtain the desired equation.  This completes the proof. 
\end{proof}

\section{Proof of Lemma \ref{lema12}}
\begin{proof}
Note that 
each  integer $t\in \left[q^k,aq^{k}\right]$ such that $q\nmid t$ can be uniquely decomposed as \begin{equation*}  
t=iq+j \quad \hbox{with } i\in [q^{k-1}, aq^{k-1}-1] \hbox{ and }j\in [1,q-1].\end{equation*}
Therefore, by substituting this decomposition,  we obtain 
 \begin{equation}\notag
 \begin{split}
     \sum\limits_{t=q^{k},q\nmid t}^{aq^{k}-1}\left[\lfloor \frac{2t}{\lambda}\rfloor -\lfloor \frac{\lfloor tq^{-1}\rfloor+t}{\lambda}\rfloor\right] &= \sum\limits_{i=q^{k-1}}^{aq^{k-1}-1}\sum\limits_{j=1}^{q-1} \left[\lfloor \frac{2iq+2j}{\lambda}\rfloor -\lfloor \frac{i+iq+j}{\lambda}\rfloor\right]\\
     &= \sum\limits_{i=q^{k-1}}^{aq^{k-1}-1}\sum\limits_{j=1}^{q-1} \left[\frac{i(q-1)}{\lambda}+\lfloor \frac{2j+2i}{\lambda}\rfloor -\lfloor \frac{j+2i}{\lambda}\rfloor \right]\\ 
     &= \sum\limits_{i=q^{k-1}}^{aq^{k-1}-1}\frac{i(q-1)^2}{\lambda}+\sum\limits_{i=q^{k-1}}^{aq^{k-1}-1}\sum\limits_{j=1}^{q-1}\left[\lfloor \frac{2j+2i}{\lambda}\rfloor-\lfloor\frac{j+2i}{\lambda} \rfloor\right].
 \end{split}
\end{equation} 
It is straightforward to obtain 
\begin{equation}\notag
    \sum\limits_{i=q^{k-1}}^{aq^{k-1}-1}\frac{i(q-1)^2}{\lambda} = \frac{q^{k-1}(q-1)^2(aq^{k-1}+q^{k-1}-1)}{2\lambda}.
\end{equation}
In addition, by applying Lemma \ref{ll9}, we have 
\begin{equation}\notag
   \sum\limits_{i=q^{k-1}}^{aq^{k-1}-1}\sum\limits_{j=1}^{q-1}\left[\lfloor \frac{2j+2i}{\lambda}\rfloor-\lfloor\frac{j+2i}{\lambda} \rfloor\right]=\begin{cases}
       \frac{q^k(q-1)(a-1)}{2\lambda}& \hbox{if }\lambda \hbox{ is odd,} \\
        \frac{q^{k-1}(q-1)(q+1)(a-1)}{2\lambda}& \hbox{if }\lambda \hbox{ is even.} 
   \end{cases}
\end{equation}  
Combining the above three equations, we obtain the desired equation. This completes the proof. 
\end{proof}

\section{Proof of Lemma \ref{lead2}}

\begin{proof}
By definition, we first have 
\begin{equation}\label{l6ee1}
\sum\limits_{t=q^{k}}^{aq^{k}-1}N(t+1)=\sum\limits_{t=q^{k}}^{aq^{k}-1}t-\sum\limits_{t=q^{k}}^{aq^{k}-1}\lfloor t/q\rfloor.
\end{equation}
Through direct computation, we obtain 
\begin{equation}\label{l6ee2}
\sum\limits_{t=q^{k}}^{aq^{k}-1}t=
\begin{cases}
\frac{1}{2}a(a-1)&\hbox{ if } k=0,\\
\frac{1}{2}(a^2-1)q^{2k}-\frac{1}{2}(a-1)q^k&\hbox{ if } k\geq 1.
\end{cases}
\end{equation}

We now  determine the value of $\sum\limits_{t=q^{k}}^{aq^{k}-1}\lfloor t/q\rfloor.$ First, note that 
$\lfloor t/q\rfloor=0$ for all $t\in \left[q^k,aq^{k}-1\right]$  when  $k=0$. This implies that     
$\sum\limits_{t=q^k}^{aq^{k}-1}\lfloor t/q\rfloor =0$ if $k=0$.   
 When  $k\geq 1,$  the interval 
$\left[q^{k}, q^{k+1}-1\right]$ can be partitioned as \begin{equation*}\left[q^{k}, q^{k+1}-1\right]=\bigsqcup\limits_{i=0}^{q^{k-1}(a-1)-1}
\left[q^{k}+i q, q^{k}+(i+1)q-1\right].\end{equation*}  Moreover, for each integer $i\in \left[0,q^{k-1}(a-1)-1\right]$ and each  integer $t\in \left[q^{k}+i q, q^{k}+(i+1)q-1\right]$, we have   $\lfloor t/q\rfloor =q^{k-1}+i$.  Therefore,  
  \begin{equation}\notag
\begin{split}\sum\limits_{t=q^{k}}^{aq^{k}-1}\lfloor t/q\rfloor&=\sum\limits_{i=0}^{q^{k-1}(a-1)-1}\sum\limits_{t=q^{k}+i q}^{q^k+(i+1)q-1}(q^{k-1}+i)\\
     &=\frac{1}{2}(a^2-1)q^{2k}-\frac{1}{2}(a-1)q^k 
     \end{split}
 \end{equation}
  if   $k\geq 1.$  
Combining (\ref{l6ee1}), (\ref{l6ee2}) and the value of $\sum\limits_{t=q^k}^{q^{k+1}-1}\lfloor t/q\rfloor$ given as above,  we obtain the desired equality. 
\end{proof}

\section{Proof of Lemma \ref{le5}}
\begin{proof}
The arguments used to derive (\ref{1}) and (\ref{1.5}) are analogous, so we  only demonstrate equation (\ref{1}) holds through the following two cases. 

\textbf{Case 1.} Suppose that $1\leq i\leq k$.  
In this case, we first have 
\begin{equation}\notag
   \lfloor \frac{\lfloor tq^{2i-1}\rfloor+t}{\lambda}\rfloor= \frac{t(q^{2i-1}-1)}{\lambda}+\lfloor \frac{2t}{\lambda}\rfloor \quad \hbox{and}\quad \lfloor \frac{\lfloor tq^{2i-2}\rfloor+t}{\lambda}\rfloor= \frac{t(q^{2i-2}-1)}{\lambda}+\lfloor \frac{2t}{\lambda}\rfloor. 
\end{equation}
It follows that 
\begin{equation}\notag
    \lfloor \frac{\lfloor tq^{2i-1}\rfloor+t}{\lambda}\rfloor -\lfloor \frac{\lfloor tq^{2i-2}\rfloor+t}{\lambda}\rfloor =\frac{tq^{2i-2}(q-1)}{\lambda}.
\end{equation}
This  leads to 
\begin{equation}\label{3}
     \sum\limits_{t=q^{k-i},q\nmid t}^{aq^{k-i}-1}\left[\lfloor \frac{\lfloor tq^{2i-1}\rfloor+t}{\lambda}\rfloor -\lfloor \frac{\lfloor tq^{2i-2}\rfloor+t}{\lambda}\rfloor\right]=   \sum\limits_{t=q^{k-i}}^{aq^{k-i}-1}  \frac{tq^{2i-2}(q-1)}{\lambda} - \sum\limits_{t\in \mathcal{W}_i}\frac{tq^{2i-2}(q-1)}{\lambda}
\end{equation}
with $\mathcal{W}_i=\left[q^{k-i},aq^{k-i}-1
\right]\cap \{t\in \mathbb{Z}: q\mid t\}$.

Note that   $\mathcal{W}_i=\varnothing$ if  $i=k$.   Thus,  
 \begin{equation}\label{4}
 \sum\limits_{t\in \mathcal{W}_i}\frac{tq^{2i-2}(q-1)}{\lambda}=0 \quad \hbox{if } i=k.
 \end{equation} 
On the other hand, it can be easily verified that    
$\mathcal{W}_i =\{q^{k-i}+ j q \mid j =0,\ldots, q^{k-i-1}(a-1)-1\}$  if  $1\leq i \leq  k-1.$
Therefore, 
 \begin{equation}\label{5}
 \begin{split}
\sum\limits_{t\in \mathcal{W}_i}\frac{tq^{2i-2}(q-1)}{\lambda}
& = \sum\limits_{j=0}^{q^{k-i-1}(a-1)-1}\frac{(q^{k-i}+ qj )q^{2i-2}(q-1)}{\lambda} \\
& = \frac{1}{2\lambda}(q-1)(a-1)\left[(a+1)q^{2k-3}+q^{k+i-2}\right] 
 \end{split}
\end{equation} if $1\leq i\leq k-1$. 
Additionally, by straightforward computation, 
we can obtain 
\begin{equation}\label{6}
\begin{aligned}
\sum\limits_{t=q^{k-i}}^{aq^{k-i}-1}\frac{tq^{2i-2}(q-1)}{\lambda} 
 =\begin{cases}\frac{1}{2\lambda}a(a-1)(q-1)^2q^{2k-2} \quad  \hbox{ if } i=k,\\
 \frac{1}{2\lambda}(q-1)(a-1)\left[(a+1)q^{2k-2}-q^{k+i-2}\right] \\
 \quad \quad \hbox{if }1\leq i\leq k-1.
 \end{cases}\end{aligned}
\end{equation}
By substituting equations (\ref{4})-(\ref{6}) into equation (\ref{3}), we conclude that equation (\ref{1}) holds for  $1\leq i\leq k$.

\textbf{Case 2.} Suppose that $-k+1\leq i\leq 0$. 
Then for any integer   $t\in \left[q^{k-i}, aq^{k-i}-1\right]\cap \{t\in \mathbb{Z}: q\nmid t\}$   with  $q$-adic expansion $\sum\limits_{\ell=0}^{k-i}t_{\ell}q^{\ell}$, 
we have 
\begin{equation}\label{lhx}
    \lfloor tq^{2i-1} \rfloor = \lfloor tq^{2i-2} \rfloor \cdot q+t_{-2i+1}. 
\end{equation}
This leads to 
\begin{equation}\notag
\begin{split}
    \lfloor \frac{t+\lfloor tq^{2i-1}\rfloor}{\lambda}\rfloor = \frac{ \sum\limits_{\ell=0}^{k-i}t_{\ell} (q^{\ell}-1)}{\lambda} + \frac{\lfloor tq^{2i-2}\rfloor(q-1)}{\lambda} +
    \lfloor  \frac{\sum\limits_{\ell=0}^{k-i}t_{\ell}+\lfloor tq^{2i-2}\rfloor+t_{-2i+1}}{\lambda} \rfloor. 
    \end{split}
\end{equation}   
We also have 
\begin{equation}\notag
\begin{split}
    \lfloor \frac{t+\lfloor tq^{2i-2}\rfloor}{\lambda}\rfloor = \frac{ \sum\limits_{\ell=0}^{k-i}t_{\ell} (q^{\ell}-1)}{\lambda} +  \lfloor  \frac{\sum\limits_{\ell=0}^{k-i}t_{\ell}+\lfloor tq^{2i-2}\rfloor}{\lambda} \rfloor. 
    \end{split}
    \end{equation}
Combining the above two equalities, we can obtain  \begin{equation}\notag
    \lfloor \frac{t+\lfloor tq^{2i-1} \rfloor}{\lambda}\rfloor-\lfloor \frac{t+\lfloor tq^{2i-2}\rfloor}{\lambda}\rfloor=   \lfloor  \frac{\sum\limits_{\ell=0}^{k-i}t_{\ell}+\lfloor tq^{2i-2}\rfloor+t_{-2i+1}}{\lambda} \rfloor 
     -\lfloor  \frac{\sum\limits_{\ell=0}^{k-i}t_{\ell}+\lfloor tq^{2i-2}\rfloor}{\lambda} \rfloor +\frac{\lfloor tq^{2i-2}\rfloor(q-1)}{\lambda}.
    \end{equation}
By applying Lemma \ref{le0} with $x=\sum\limits_{\ell=1}^{k-i}t_{\ell}+\lfloor tq^{2i-2}\rfloor+t_{-2i+1}$ and $ y= \sum\limits_{\ell=1}^{k-i}t_{\ell}+\lfloor tq^{2i-2}\rfloor
  $, we  have 
\begin{equation}\notag
\sum\limits_{t_0=1}^{q-1}
    \left[ \lfloor  \frac{\sum\limits_{\ell=0}^{k-i}t_{\ell}+\lfloor tq^{2i-2}\rfloor+t_{-2i+1}}{\lambda} \rfloor 
     -\lfloor  \frac{\sum\limits_{\ell=0}^{k-i}t_{\ell}+\lfloor tq^{2i-2}\rfloor}{\lambda} \rfloor \right]= \frac{(q-1)t_{-2i+1}}{\lambda}.
\end{equation}
Noticing that the value of $\lfloor tq^{2i-2}\rfloor$ is independent of $t_0$, we get 
\begin{equation}\notag
    \sum\limits_{t_0=1}^{q-1}\frac{\lfloor tq^{2i-2}\rfloor(q-1)}{\lambda}  = \frac{\lfloor tq^{2i-2}\rfloor(q-1)^2}{\lambda}.
\end{equation}
Noting that $\lfloor tq^{2i-2}\rfloor=\lfloor\frac{\lfloor tq^{2i-1}\rfloor}{q}\rfloor$ and  
recalling equation (\ref{lhx}), we can add  
the above two sums  to  obtain 
\begin{equation}\label{whlx}
\begin{split}
    \sum\limits_{t_0=1}^{q-1}\left[
     \lfloor \frac{t+\lfloor tq^{2i-1} \rfloor}{\lambda}\rfloor-\lfloor \frac{t+\lfloor tq^{2i-2}\rfloor}{\lambda}\rfloor  \right]&= 
     \frac{\lfloor tq^{2i-2}\rfloor(q-1)^2}{\lambda} + \frac{(q-1)t_{-2i+1}}{\lambda}\\
    &=\frac{q-1}{\lambda}N(\lfloor tq^{2i-1}\rfloor+1).
    \end{split}
\end{equation}

In addition, each  integer  $t\in \left[q^{k-i}, aq^{k-i+1}-1\right]\cap \{t\in \mathbb{Z}: q\nmid t\}$ with $q$-adic expansion $\sum\limits_{\ell=0}^{k-i}t_{\ell}q^{\ell}$ can be uniquely decomposed  as   
 \begin{equation}\notag
t=\lfloor tq^{2i-1} \rfloor\cdot q^{-2i+1}+\sum\limits_{\ell=1}^{-2i}t_{\ell}q^{\ell}+t_0.
\end{equation}
Furthermore, as $ t $ ranges over integers from $ q^{k-i} $ to $q^{k-i+1} - 1$ that are not divisible by $ q $, the value of $ \lfloor tq^{2i-1} \rfloor $ ranges from $ q^{k+i-1} $ to $ aq^{k+i-1}-1$. Additionally, for each fixed value of $ \lfloor tq^{2i-1} \rfloor $,  the sum $ \sum\limits_{\ell=1}^{-2i} t_{\ell} q^{\ell} $ ranges over integers from $ 0 $ to $q^{-2i}-1$, while $ t_0 $ varies from $ 1 $ to $ q - 1 $.  Therefore, we can conclude that 
\begin{equation}\notag
\begin{split}
      \sum\limits_{\substack{t=q^{k-i}, q\nmid t}}^{ q^{k-i+1}-1}\left[\lfloor \frac{\lfloor tq^{2i-1}\rfloor+t}{\lambda}\rfloor -\lfloor \frac{\lfloor tq^{2i-2}\rfloor+t}{\lambda}\rfloor\right]&= \sum\limits_{\lfloor tq^{2i-1} \rfloor=q^{k+i-1}}^{aq^{k+i-1}-1 } 
\sum\limits_{\sum_{\ell=1}^{-2i}t_{\ell}q^{\ell}=0}^{q^{-2i}-1}\sum\limits_{t_0=1}^{q-1}\left[  \lfloor \frac{t+\lfloor tq^{2i-1} \rfloor}{\lambda}\rfloor-\lfloor \frac{t+\lfloor tq^{2i-2}\rfloor}{\lambda}\rfloor \right]\\
&=\sum\limits_{\lfloor tq^{2i-1} \rfloor=q^{k+i-1}}^{aq^{k+i-1}-1 }  \frac{q^{-2i}(q-1)}{\lambda}N(\lfloor tq^{2i-1}\rfloor+1),
\end{split}
\end{equation}
where the second equality follows from equation (\ref{whlx}).  We can now apply   Lemma \ref{lead2} to conclude that  (\ref{1})  holds for $-k+1\leq i\leq 0$.  

By now, we have established equation (\ref{1})
\end{proof}

\section{Proof of Assertion \ref{as1}}\label{API}
\begin{proof}
Suppose that  $m$ is odd.  We first aim to  show that 
\begin{equation}\label{ha2}
\left[\sum\limits_{\ell=h+s_{\delta}}^{h+k_{\delta}}\delta_{\ell}q^{\ell}, \lambda(\delta-1)\right]\cap \mathcal{S}\cap \mathcal{D}_{\lambda}=\left[\sum\limits_{\ell=h+s_{\delta}}^{h+k_{\delta}}\delta_{\ell}q^{\ell}, \lambda(\delta-1)\right]\cap \mathcal{A}_{k_{\delta}}(s_{\delta}) \cap \mathcal{D}_{\lambda} .
\end{equation}
For any integer $a\in \left[\sum\limits_{\ell=h-k_{\delta}+1}^{h+k_{\delta}}\delta_{\ell}q^{\ell},\lambda(\delta-1)\right]$  with $q$-adic expansion $\sum\limits_{\ell=0}^{h+k_{\delta}}a_{\ell}q^{\ell}$, we have 
\begin{equation*}V(\sum\limits_{\ell=h-k_{\delta}+1}^{h+k_{\delta}}\delta_{\ell}q^{\ell})\leq V(a)\leq V(\lambda(\delta-1)). \end{equation*}
Noting that 
\begin{equation}\notag
V(\lambda(\delta-1))=(\mathbf{0}_{h-k_{\delta}},\delta_{h+k_{\delta}},\ldots,\delta_{0})
\end{equation}
and 
\begin{equation}\notag
V(\sum\limits_{\ell=h-k_{\delta}+1}^{h+k_{\delta}}\delta_{\ell}q^{\ell})=(\mathbf{0}_{h-k_{\delta}},\delta_{h+k_{\delta}},\ldots,\delta_{h-k_{\delta}+1}, \mathbf{0}_{h-k_{\delta}+1}), 
\end{equation}
it follows that  \begin{equation}\label{vvv}
V(a)=(\mathbf{0}_{h-k_{\delta}},a_{h+k_{\delta}},\ldots,a_0)=(\mathbf{0}_{h-k_{\delta}},\delta_{h+k_{\delta}},\ldots,\delta_{h-k_{\delta}+1} , a_{h-k_{\delta}},\ldots,a_0).
\end{equation}
Recalling the definition of  $s_{\delta}$,   it follows that $s_{\delta}$ is  the smallest integer in $[-k_{\delta}+1,k_{\delta}]$ such that $a_{h+s_{\delta}}>0$.  
 By Remark  \ref{cor2},  this implies  $a\not \in \mathcal{A}_{k_{\delta}}(i)$ for any integer $i\neq s_{\delta}$.  Therefore, 
$\left[\sum\limits_{\ell=h-k_{\delta}+1}^{h+k_{\delta}}\delta_{\ell}q^{\ell},\delta-1\right]\cap \mathcal{A}_{k_{\delta}}(i)=\varnothing$ for any integer $i\neq s_{\delta}.$  Then applying  Lemma  \ref{th2},  we can conclude that  (\ref{ha2}) holds.  

Now, let us count the number of integers in the set $\left[\sum\limits_{\ell=h-k_{\delta}+1}^{h+k_{\delta}}\delta_{\ell}q^{\ell},\delta-1\right]\cap \mathcal{A}_{k_{\delta}}(s_{\delta})\cap \mathcal{D}_{\lambda}.$
    Recall that $w_{\delta}= \sum\limits_{\ell=h+s_{\delta}}^{h+k_{\delta}}\delta_{\ell}q^{\ell}$, $\mu(\delta)=\min\left \{\sum\limits_{\ell=0}^{h-k_{\delta}}\delta_{\ell}q^{\ell}, \sum\limits_{\ell=h-s_{\delta}+1}^{h+k_{\delta}}\delta_{\ell}q^{\ell-(h-s_{\delta}+1)}  \right\}$, and $\alpha(a)=\sum\limits_{\ell=0}^{k_{\delta}+s_{\delta}-1}a_{\ell}q^{\ell}$  for each integer $a\in \mathcal{A}_{k_{\delta}}(s_{\delta})$ with $q$-adic expansion $\sum\limits_{\ell=0}^{h+k_{\delta}}a_{\ell}q^{\ell}$. 
 We conclude from equation (\ref{vvv}) and  the definition of $\mathcal{A}_{k_{\delta}}(s_{\delta})$  that an integer  $a\in \left[\sum\limits_{\ell=h-k_{\delta}+1}^{h+k_{\delta}}\delta_{\ell}q^{\ell},\delta-1\right]\cap \mathcal{A}_{k_{\delta}}(s_{\delta})\cap \mathcal{D}_{\lambda}$ if and only if $a$ satisfies the decomposition 
\begin{equation*}a=w_{\delta}+\alpha(a)\end{equation*}
with $   \alpha(a) \in \{ \alpha\in \left[1,    \mu(\delta)\right]: q\nmid \alpha \hbox{ and }\lambda\mid \alpha+w_{\delta}\}. 
$
Consequently, we have 
\begin{equation}\notag  \left|\left[\sum\limits_{\ell=h+s_{\delta}}^{h+k_{\delta}}\delta_{\ell}q^{\ell}, \lambda(\delta-1)\right]\cap \mathcal{A}_{k_{\delta}}\cap \mathcal{D}_{\lambda}\right| = 
   \left| \{\alpha\in \left[1,\mu(\delta)\right]: q\nmid \alpha \hbox{ and } \lambda \mid \alpha+w_{\delta} \}\right|.
\end{equation} 
Then by applying Lemma \ref{le1}, equation  (\ref{as1e0}) follows.

Suppose that $m$ is even. The equality in (\ref{as1e0-1}) can be obtained by  employing  a similar argument as above.  It remains to  establish equation (\ref{as1e0-2}). By applying Lemma \ref{corr},  
we can conclude that 
$a\in \left[\sum\limits_{\ell=h}^{h+k_{\delta}}\delta_{\ell}q^{\ell}, \lambda(\delta-1)\right]\cap \mathcal{H}\cap \mathcal{D}_{\lambda}$ if and only if $V(a)$ has the form 
\begin{equation}\notag
(\mathbf{0}_{h-k_{\delta}-1},\delta_{h+k_{\delta}},\ldots,\delta_h, \mathbf{0}_{h-k_{\delta}-1},\delta_{h+k_{\delta}},\ldots,\delta_h)
\end{equation}
 with $\delta_h>0$,   $\sum\limits_{\ell=h}^{h+k_{\delta}}\delta_{\ell}q^{\ell-h}\leq\sum\limits_{\ell=0}^{h-1}\delta_{\ell}q^{\ell}$ and $\lambda \mid \sum\limits_{\ell=h}^{h+k_{\delta}}\delta_{\ell}q^{\ell-h}+\sum\limits_{\ell=h}^{h+k_{\delta}}\delta_{\ell}q^{\ell}$. Since $\lambda\mid q-1$, the condition   $\lambda \mid \sum\limits_{\ell=h}^{h+k_{\delta}}\delta_{\ell}q^{\ell-h}+\sum\limits_{\ell=h}^{h+k_{\delta}}\delta_{\ell}q^{\ell}$ is equivalent $\lambda \mid 2\sum\limits_{\ell=h}^{h+k_{\delta}}\delta_{\ell}q^{\ell}.$ 
 Therefore,  an integer satisfying the above condition exists and is unique if and only if  $\delta_h>0$, $\sum\limits_{\ell=h}^{h+k_{\delta}}\delta_{\ell}q^{\ell-h}\leq \sum\limits_{\ell=0}^{h-1}\delta_{\ell}q^{\ell}$ and $\lambda \mid 2\sum\limits_{\ell=h}^{h+k_{\delta}}\delta_{\ell}q^{\ell}$. 
It follows that equation (\ref{as1e0-2}) holds. 
\end{proof}

\section{Proof of Assertion \ref{as4}}

\begin{proof} 
Suppose that $m$ is odd. By definition, an integer
$
a \in \left[q^{h+k_{\delta}}, \; \sum_{\ell=h-k_{\delta}+1}^{h+k_{\delta}} \delta_{\ell} q^{\ell}\right) \cap \mathcal{A}_{k_{\delta}}(i)
$
if and only if the following conditions are satisfied:
\begin{itemize}
    \item[(\expandafter{\romannumeral1})] $V(a)$ is of the form given in~(\ref{p33}) and satisfies~(\ref{p31}) and~(\ref{p32}) for $k = k_{\delta}$;
    \item[(\expandafter{\romannumeral2})] $V(q^{h+k_{\delta}}) \leq V(a) < V\left(\sum_{\ell=h-k_{\delta}+1}^{h+k_{\delta}} \delta_{\ell} q^{\ell}\right)$.
\end{itemize}
Notice  that
\begin{equation*}
V(q^{h+k_{\delta}}) = (\mathbf{0}_{h-k_{\delta}}, 1, \mathbf{0}_{h+k_{\delta}})
\quad \text{and} \quad
V\left(\sum_{\ell=h-k_{\delta}+1}^{h+k_{\delta}} \delta_{\ell} q^{\ell}\right) = (\mathbf{0}_{h-k_{\delta}}, \delta_{h+k_{\delta}}, \ldots, \delta_{h-k_{\delta}+1}, \mathbf{0}_{h-k_{\delta}+1}).
\end{equation*}
Therefore, the form of $V(a)$ in~(\ref{p33}) together with the inequality $a_0 > 0$ in~(\ref{p32}) imply that condition~(\expandafter{\romannumeral2}) is satisfied if and only if
\begin{equation*}
(1, \mathbf{0}_{2k_{\delta}-1}) \leq (a_{h+k_{\delta}}, \ldots, a_{h+i}, \mathbf{0}_{k_{\delta}+i-1}) < (\delta_{h+k_{\delta}}, \ldots, \delta_{h-k_{\delta}+1}).
\end{equation*}
Since $s_{\delta}$ is the smallest integer in $[-k_{\delta}+1, k_{\delta}]$ such that $\delta_{h+s_{\delta}} > 0$, this is further equivalent to
\begin{equation*}
q^{k_{\delta}-i} \leq \sum_{\ell=h+i}^{h+k_{\delta}} a_{\ell} q^{\ell-h-i} < \sum_{\ell=h+s_{\delta}}^{h+k_{\delta}} \delta_{\ell} q^{\ell-h-i}.
\end{equation*}

Recall the definition of the set $\mathcal{T}_i(\delta)$ and Remarks \ref{rm3} --  \ref{rm4}. We now  
 can  conclude that for each integer $i\in \left[-k_{\delta}+1, k_{\delta}\right]$, 
 an integer $a\in \left[q^{h+k_{\delta}},  \sum\limits_{\ell=h-k_{\delta}+1}^{h+k_{\delta}}\delta_{\ell}q^{\ell}\right)\cap \mathcal{A}_{k_{\delta}}(i)\cap \mathcal{D}_{\lambda}$ 
if and only if  $a$ admits the  decomposition 
\begin{equation*}a=t(a)\cdot q^{h+i}+\alpha(a)\end{equation*}
 with 
 \begin{equation}\notag
\begin{cases}
 \lambda \mid t(a)+\alpha(a),\\
t(a)\in \mathcal{T}_i(\delta),\\
1\leq \alpha(a)\leq \lfloor t(a)\cdot q^{2i-1}\rfloor \hbox{ and } q\nmid \alpha(a).
\end{cases}
\end{equation}

We may now conclude that, for each integer 
$i \in [-k_{\delta}+1,\, k_{\delta}]$, an integer
\[
a \in \left[q^{h+k_{\delta}}, \sum_{\ell=h-k_{\delta}+1}^{h+k_{\delta}}
\delta_{\ell} q^{\ell}\right)
\cap \mathcal{A}_{k_{\delta}}(i) \cap \mathcal{D}_{\lambda}
\]
exists if and only if $a$ admits a decomposition of the form
\[
a = t(a)\, q^{h+i} + \alpha(a),
\]
where
\[
\begin{cases}
\lambda \mid t(a) + \alpha(a),\\
t(a) \in \mathcal{T}_i(\delta),\\
1 \le \alpha(a) \le \lfloor t(a)\, q^{2i-1} \rfloor
\quad \text{and} \quad q \nmid \alpha(a).
\end{cases}
\]

Consequently, we can apply Lemma \ref{le1} to conclude  that 
\begin{equation}\notag
\begin{split}
\left| \left[q^{h+k_{\delta}},  \sum\limits_{\ell=h-k_{\delta}+1}^{h+k_{\delta}}\delta_{\ell}q^{\ell}\right)\cap \mathcal{A}_{k_{\delta}}(i)\cap \mathcal{D}_{\lambda}\right|& =\sum\limits_{t\in \mathcal{T}_i(\delta)}\left| \left\{ \alpha\in [1, \lfloor tq^{2i-1} \rfloor]: q\nmid \alpha \hbox{ and }\lambda\mid \alpha + t      \right\}\right|\\
&=\sum\limits_{t\in \mathcal{T}_i(\delta)}\left[\lfloor \frac{\lfloor tq^{2i-1}\rfloor+t}{\lambda}\rfloor -\lfloor \frac{\lfloor tq^{2i-2}\rfloor+t}{\lambda}\rfloor\right].
\end{split}
\end{equation}
Then applying Lemma  \ref{th2}, it follows that (\ref{as2e}) holds. 

Suppose that $m$ is even. We can first utilize a similar argument as above to obtain (\ref{as2e2}).  Next, we demonstrate that equation (\ref{as2e3}) holds. By applying Lemma \ref{corr}, we can conclude that an integer $a\in \left[q^{h+k_{\delta}}, \sum\limits_{\ell=h}^{h+k_{\delta}}
        \delta_{\ell}q^{\ell}\right)\cap  \mathcal{H} \cap \mathcal{D}_{\lambda}$ if and only if $a$ can be decomposed as  decomposition 
        \begin{equation}\label{de}
a=\sum\limits_{\ell=0}^{k_{\delta}}a_{\ell}q^{\ell}\cdot q^{h}+\sum\limits_{\ell=0}^{k_{\delta}}a_{\ell}q^{\ell}
        \end{equation}
        with 
        $q\nmid \sum\limits_{\ell=0}^{k_{\delta}}a_{\ell}q^{\ell}$, $q^{h+k_{\delta}}\leq a<\sum\limits_{\ell=h}^{h+k_{\delta}}\delta_{\ell}q^{\ell}$ and $\lambda\mid a$.  
        It is easy to see that  $q^{h+k_{\delta}}\leq a <\sum\limits_{\ell=h}^{h+k_{\delta}}\delta_{\ell}q^{\ell}$ holds if and only if 
$q^{k_{\delta}}\leq \sum\limits_{\ell=0}^{k_{\delta}}a_{\ell}q^{\ell}\leq \sum\limits_{\ell=h}^{h+k_{\delta}}\delta_{\ell}q^{\ell-h}-1$. Furthermore, 
  since $\lambda\mid q-1$,  the condition $\lambda\mid a $ is satisfied if and only if $\lambda\mid 2\sum\limits_{\ell=0}^{k_{\delta}}a_{\ell}q^{\ell}$.   Therefore,      an integer $a\in \left[q^{h+k_{\delta}}, \sum\limits_{\ell=h}^{h+k_{\delta}}
        \delta_{\ell}q^{\ell}\right)\cap  \mathcal{H} \cap \mathcal{D}_{\lambda}$ if and only if $a$ admits  the decomposition given in (\ref{de})
with 
\begin{equation}\notag
    \begin{cases}
    q^{k_{\delta}}\leq \sum\limits_{\ell=0}^{k_{\delta}}a_{\ell}q^{\ell}\leq \sum\limits_{\ell=h}^{h+k_{\delta}}\delta_{\ell}q^{\ell-h}-1;\\
        q\nmid \sum\limits_{\ell=0}^{k_{\delta}}a_{\ell}q^{\ell};\\
        \lambda\mid 2\sum\limits_{\ell=0}^{k_{\delta}}a_{\ell}q^{\ell}.
    \end{cases}
\end{equation}
 Consequently, 
\begin{equation}\notag
\begin{split}
   \left|\left[q^{h+k_{\delta}}, \sum\limits_{\ell=h}^{h+k_{\delta}}
        \delta_{\ell}q^{\ell}\right)\cap  \mathcal{H} \cap \mathcal{D}_{\lambda}\right| = \left|\left \{\alpha \in \left[q^{k_{\delta}}, \sum\limits_{\ell=h}^{h+k_{\delta}}\delta_{\ell}q^{\ell-h}-1\right]: \lambda \mid 2\alpha \hbox{ and }  q \nmid \alpha   \right\}\right|
        \end{split}
\end{equation}
By applying Lemma \ref{lll8}, it follows that equation (\ref{as2e3}) holds. 
\end{proof}

\section{Proof of Assertion \ref{as2}}
\begin{proof}
The arguments used to establish (\ref{as3ee0}) and (\ref{as3ee05}) are analogous.  Therefore, we only demonstrate that (\ref{as3ee0}) holds when $m$ is odd. Using reasoning similar to that in the proof of Assertion 2, we can conclude that an integer $a \in \mathcal{A}_{k_{\delta}}(i) \cap \mathcal{D}_{\lambda}$ if and only if $a$ admits the decomposition 
\begin{equation*}
a = t(a) \cdot q^{h+i} + \alpha(a)
\end{equation*}
with 
\begin{equation*}
\begin{cases}
\lambda \mid t(a) + \alpha(a); \\[1ex]
t(a) \in \left[q^{k-i},\, q^{k-i+1} - 1 \right] \cap \{ t \in \mathbb{Z} : q \nmid t \}; \\[1ex]
1 \leq \alpha(a) \leq \left\lfloor t(a) \cdot q^{2i-1} \right\rfloor  \text{ and }  q \nmid \alpha(a).
\end{cases}
\end{equation*}
 Consequently, by applying Lemma \ref{le1}, we derive 
\begin{equation}\notag
\begin{split}
\left| \mathcal{A}_{k_{\delta}}(i)\cap \mathcal{D}_{\lambda}\right|& =\sum\limits_{t=q^{k-i},q\nmid t}^{q^{k-i+1}-1}\left| \left\{ \alpha\in \left[1, \lfloor tq^{2i-1} \rfloor\right]: q\nmid \alpha \hbox{ and }\lambda\mid \alpha  + t     \right\}\right|\\
&=\sum\limits_{t=q^{k-i},q\nmid t}^{q^{k-i+1}-1} \left[\lfloor \frac{\lfloor tq^{2i-1}\rfloor+t}{\lambda}\rfloor -\lfloor \frac{\lfloor tq^{2i-2}\rfloor+t}{\lambda}\rfloor\right].
\end{split}
\end{equation}
Then applying Lemma \ref{le5}, we obtain equation (\ref{as3ee0}).

\end{proof}

\section{Proof of Assertion \ref{as3}}
\begin{proof}
Following a similar approach to the proof of Assertion \ref{as2}, we can obtain 
\begin{equation}\notag
|B_k(0) \cap \mathcal{D}_{\lambda}| = 
    \sum_{t=q^{k}, q \nmid t}^{q^{k+1}-1} \left[ \lfloor \frac{2t}{\lambda} \rfloor - \lfloor \frac{\left\lfloor {t}{q}^{-1} \right\rfloor + t}{\lambda}\rfloor \right]. 
\end{equation}
Then applying Lemmas \ref{ll9} and \ref{lema12},  the assertion follows.
\end{proof}

\section{Proof of Assertion \ref{as3.5}}\label{APM}
\begin{proof}
 By applying Lemma \ref{corr}, we can conclude that an integer $a\in \mathcal{H}\cap [q^{h+k},q^{h+k+1})\cap \mathcal{D}_{\lambda}$ if and only if 
 $a$ admits the decomposition   \begin{equation}\label{as5eh}
a=\sum\limits_{\ell=0}^{k}a_{\ell}q^{\ell}\cdot  q^{h}+\sum\limits_{\ell=0}^{k}a_{\ell}q^{\ell}.
 \end{equation} 
 with $a_0>0$, $a_{k}>0$, and $\lambda\mid a$. 
 Since $\lambda\mid q-1$, it follows that $\lambda\mid a$ is equivalent to  $\lambda \mid 2\sum\limits_{\ell=0}^{k}a_{\ell}q^{\ell}$. In addition, it is straightforward to verify that the conditions $a_{0}>0$ and $a_{k}>0$ are satisfied if and only if $q^k\leq \sum\limits_{\ell=0}^{k}a_{\ell}q^{\ell} \leq q^{k+1}-1$ and $q\nmid \sum\limits_{\ell=0}^{k}a_{\ell}q^{\ell}$. Therefore,   an integer $a\in \mathcal{H}\cap [q^{h+k},q^{h+k+1})\cap \mathcal{D}_{\lambda}$ if and only if $a$ has the decomposition as given in  (\ref{as5eh})
 with 
 \begin{equation}\notag
     \begin{cases}
     q\nmid \sum\limits_{\ell=0}^{k}a_{\ell}q^{\ell};\\
     q^k\leq \sum\limits_{\ell=0}^{k}a_{\ell}q^{\ell}\leq q^{k+1}-1;\\
         \lambda \mid 2\sum\limits_{\ell=0}^{k}a_{\ell}q^{\ell}.
     \end{cases}
 \end{equation}
Consequently, 
 \begin{equation}\notag
     \left|\mathcal{H}\cap \left[q^{h+k},q^{h+k+1}\right)\cap \mathcal{D}_{\lambda}\right| =\left| \left\{ \alpha \in \left[q^k, q^{k+1}-1\right]: \lambda\mid 2\alpha \hbox{ and }q \nmid \alpha\right\}\right|.
 \end{equation} 
Then by applying Lemma \ref{lll8},  the assertion follows.  
 \end{proof}

\end{document}